\documentclass[a4paper]{article}
\usepackage{graphics, color, latexsym, amssymb, amsmath, tabularx, epsfig, xspace} 
\usepackage{algorithm}
\usepackage[latin1]{inputenc}
\usepackage[T1]{fontenc}
\usepackage{amsfonts}
\usepackage{lmodern}
\usepackage{tikz}
\usepackage{cancel}
\usepackage{booktabs}
\usepackage{enumitem}
\usepackage{amsthm}

\usepackage{geometry}
\usetikzlibrary{arrows,decorations,backgrounds}
\usetikzlibrary{patterns}
\usetikzlibrary{positioning}
\usetikzlibrary{shapes}
\usetikzlibrary{calc}
\usepackage{color}
\definecolor{TUGreen}{rgb}{0.517,0.721,0.094}
\definecolor{TUOrange}{rgb}{1.0,0.7176,0.0}
\definecolor{BrightGray}{gray}{0.9}
\definecolor{DarkGray}{gray}{0.2}
\definecolor{white}{rgb}{1,1,1}
\definecolor{black}{rgb}{0,0,0}
\definecolor{red}{rgb}{1,0,0}

\usepackage{xcolor}
\definecolor{itemizeblue}{rgb}{0.2,0.2,0.7}
\definecolor{shorse}{rgb}{0.84,0.84,0.94}
\definecolor{shorselight}{rgb}{0.93,0.93,0.98}
\definecolor{shorsedark}{rgb}{0.76,0.76,0.91}
\usetikzlibrary{topaths}
\usetikzlibrary{decorations.pathreplacing}

\usepackage{subcaption}
\newcommand{\dnode}[1]{\textnormal{d-node}(#1)}
\newcommand{\vars}[1]{\textnormal{vars}(#1)}
\newcommand{\rootset}[1]{\textnormal{root}(#1)}
\newcommand{\nodefunc}[1]{\textnormal{node}(#1)}

\newtheorem{proposition}{Proposition}
\newtheorem{theorem}{Theorem}
\newtheorem{lemma}{Lemma}
\newtheorem{corollary}{Corollary}
\newtheorem{simulation}{Simulation}

\theoremstyle{remark}
\newtheorem*{proofidea}{Proof idea}

\theoremstyle{definition}
\newtheorem{definition}{Definition}

\renewenvironment{proof}[1][Proof]{
                \begin{trivlist}
                \item \upshape \bfseries #1.
                \upshape\mdseries}
               {\nopagebreak[4]\hspace*{\fill}\mbox{$\Box$}\end{trivlist}}

\setitemize{itemsep=-2pt}
\setenumerate{itemsep=-2pt}

\begin{document}

\title{On the Relative Succinctness of Sentential Decision Diagrams}

\author{Beate Bollig\thanks{TU Dortmund, LS2 Informatik, Germany, Email: beate.bollig@tu-dortmund.de} \and Matthias Buttkus\thanks{TU Dortmund, LS2 Informatik, Germany, Email: matthias.buttkus@tu-dortmund.de}}

%\institute{Beate Bollig \at TU Dortmund, LS2 Informatik, Germany\\
%\institute{Beate Bollig \at TU Dortmund, LS2 Informatik, Germany\\
%\email{beate.bollig@tu-dortmund.de}\\
%and\\ 
%Matthias Buttkus \at TU Dortmund, LS2 Informatik, Germany\\
%\email{matthias.buttkus@tu-dortmund.de}

\maketitle

\begin{abstract}

Sentential decision diagrams (SDDs) introduced by Darwiche in 2011
%as representation type of propositional knowledge bases
are a promising representation type used in knowledge compilation.
The relative succinctness of representation types is an important subject in this area.
The aim of the paper is to identify which kind of Boolean functions can be represented by SDDs of small size
with respect to the number of variables the functions are defined on.
For this reason the sets of Boolean functions representable by different representation 
types in polynomial size are investigated and SDDs are compared with representation types from the 
classical knowledge compilation map of Darwiche and Marquis.
Ordered binary decision diagrams (OBDDs) which are a popular data structure for Boolean functions
are one of these representation types. 
SDDs are more general than OBDDs by definition but only recently, a Boolean function was 
presented with polynomial SDD size but exponential OBDD size. This result is strengthened in several ways.
The main result is a quasipolynomial simulation of SDDs by equivalent unambiguous nondeterministic 
OBDDs, a nondeterministic variant where there exists exactly one accepting computation for each satisfying input.
As a side effect an open problem about the relative succinctness between SDDs and 
free binary decision diagrams (FBDDs) which are more general than OBDDs is answered.\\

\end{abstract}

%%%%%%%%%%%%%%%%%%%%%%%%%%%%%%%%%%%%%%%%%%%%%%%%%%%%%%%%%%%%%%%%%%%%%%%%%%%
% Section 1
\section{Introduction}\label{sec1}
%\paragraph{{\bf Motivation and results}}

Knowledge compilation is an area of research with a long tradition in artificial
intelligence (see, e.g., \cite{CD97}).
An input formula is converted into a representation of the Boolean function
that the formula defines from which some tasks can (hopefully) be done efficiently.
%Efficiency depends on the size of the representation and the time required to construct it.
Developing their knowledge compilation map
Darwiche and Marquis identified sets of
useful queries and transformations in the area of knowledge compilation
and compared systematically different representation types w.r.t.\ their succinctness
and efficient support of these operations \cite{DM02}.
One aim of their work was 
%the investigation of the relative succinctness of different representation types and 
to decide whether representations can be transformed into equivalent ones of another
representation type at the cost of increasing the representation size at most polynomially.
Here we continue this part of their work. 
Sentential decision diagrams, or SDDs for short, introduced by Darwiche \cite{Dar11}
are a promising representation type for propositional knowledge bases in artificial intelligence.
Our main motivation in the paper is to characterize which kind of Boolean functions can be represented
by SDDs of small size.

%\paragraph{{\bf Contribution and known results}}
\paragraph{{\bf Contribution and related work}}

For a representation type $\mathcal{M}$ let $\mathcal{P}(\mathcal{M})$
be the set of all Boolean functions representable by $\mathcal{M}$
in polynomial size w.r.t.\ the number of Boolean variables the functions are defined on.
We call $\mathcal{P}(\mathcal{M})$ a complexity class.
Our aim is to characterize the complexity class $\mathcal{P}($SDD) as precisely as possible.
For the formal definitions of the following representation types see Section \ref{sec2}.

If one likes to have representations of small size for Boolean functions, 
circuits are the most powerful model. 
The desire to find representation types with better algorithmic properties 
leads to restricted circuits.
Decomposable negation normal form circuits, or DNNFs for short, 
introduced  by Darwiche \cite{Dar01}
are the most general one of these representation types discussed in this paper.
The subcircuits leading into each $\wedge$-gate (conjunction) are defined on 
disjoint sets of variables. Darwiche also defined deterministic DNNFs, or $d$-DNNFs for short,
where the subcircuits leading into each $\vee$-gate (disjunction) never simultaneously evaluate 
to the function value $1$. This restriction allows polynomial-time equivalence testing \cite{HD07}.

In his seminal paper Bryant showed that ordered binary decision diagrams, or OBDDs for short,
are well suited as data structure for Boolean functions \cite{Bry86}.
Since some important functions have exponential OBDD size,
many variants and extensions have been considered (for an extensive discussion 
see, e.g., the monograph of Wegener \cite{Weg00}). 
Besides nondeterministic variants and co-nondeterministic variants,
free binary decision diagrams (FBDDs) and $k$-OBDDs, 
for constant $k$, have been investigated. FBDDs and $k$-OBDDs are by definition 
more general than OBDDs.

%DNNFs are the most general model discussed in this paper.
SDDs 
%introduced by Darwiche \cite{Dar11}
are restricted $d$-DNNFs more general than OBDDs.
%that have applications in artificial intelligence.
Recently, Bova 
%showed that the so-called hidden weighted bit function (for a formal definition see also Section
%\ref{sec2}) is 
provided a function
in $\mathcal{P}($SDD) whose OBDD size is exponential \cite{Bov16}.
This result is strengthened by our proof that there exist Boolean functions
representable by SDDs of polynomial size but with exponential FBDD size
(see Section \ref{sec:storage}).
This result answers a question posed by Beame and Liew (see {\it Discussion} in \cite{BL15})
in the affirmative
whether SDDs are ever more concise than so-called decision-DNNFs which are also restricted $d$-DNNFs
considered in database theory in the context of probabilistic databases. (See, e.g., \cite{OD14} for a discussion
on the importance of decision DNNFs in model counting, the problem to compute the number of satisfying 
assignments of a Boolean formula.) There exists a
quasipolynomial simulation of decision-DNNFs by equivalent FBDDs \cite{BLR13}.
Moreover, Beame and Liew showed that SDDs are sometimes exponentially less concise than FBDDs \cite{BL15}.
Therefore, we can conclude that SDDs and FBDDs are incomparable w.r.t.\ polynomial-size 
representations (see also Figure \ref{fig:landscape2}).
In other words, $\mathcal{P}($SDD) is not a subset of $\mathcal{P}($FBDD) and vice versa.
Furthermore, we prove that SDDs are even more powerful w.r.t.\ polynomial-size representations than $k$-OBDDs,
where $k$ is a constant
%%where $k$ is independent on the number of Boolean variables,
%a generalization of OBDDs 
(see Section \ref{sec:comparison}). 
For this result we use a polynomial transformation from $k$-OBDDs for $k$ into equivalent
unambiguous nondeterministic OBDDs.
Until now it is open whether the set of Boolean functions representable by
polynomial-size unambiguous nondeterministic OBDDs,
or $\vee_1$-OBDDs for short,
that have exactly one accepting path for every satisfying input is a subset of $\mathcal{P}($SDD)
(see also Figure \ref{fig:landscape1}).
One of our main results is the proof  
that every Boolean function
$f$ for which $f$ and its negated function $\overline{f}$ can be represented by 
polynomial-size unambiguous nondeterministic OBDDs w.r.t.\ the same variable ordering can also be 
represented by SDDs of polynomial size (see Section \ref{section:simulation_vee_one_obdds_by_sdds}).
This result is sufficient to prove that $\mathcal{P}(k$-OBDD)$\subseteq \mathcal{P}($SDD).
Adapting a result from Sauerhoff that nondeterministic OBDDs where all nondeterministic decisions
are made at the beginning of the computations are less powerful w.r.t.\ polynomial-size representation
than general nondeterministic OBDDs \cite{Sau03b}, we can strengthen our result to
$\mathcal{P}(k$-OBDD)$\subsetneq \mathcal{P}($SDD).

Razgon proved a quasipolynomial separation between decision-DNNFs and nondeterministic FBDDs,
or $\vee$-FBDDs for short, \cite{Raz16}. He presented a Boolean function
with polynomial decision-DNNF size but only quasipolynomial nondeterministic FBDD size.
A careful inspection of his results (Theorem 2 and 3 in \cite{Raz16}) in combination with
a result from Darwiche (Theorem 13 in \cite{Dar11}) also leads to a quasipolynomial separation 
between SDDs and nondeterministic FBDDs. Since FBDDs are more general than OBDDs this
is also a quasipolynomial separation between SDDs and nondeterministic OBDDs. 
Recently, strengthening his result, Razgon presented a quasipolynomial separation between SDDs and a representation typ
more general than nondeterministic OBDDs \cite{Raz17}.
The second main result of our paper is the proof that
SDDs can be simulated with only a quasipolynomial size increase
by equivalent unambiguous nondeterministic OBDDs (see Sections \ref{section:simulation_sdnnfs_by_vee_obdds}
and \ref{section:simulating_structured_d_dnnfs}).
This simulation yields directly lower bounds on the SDD size of Boolean functions $f$ from
unambiguous nondeterministic OBDD lower bounds for $f$.
Because of Razgon's quasipolynomial separation \cite{Raz17} 
our result is tight.
For our simulation we extend ideas described independently by Beame and Liew  and by Razgon for a quasipolynomial transformation
from DNNFs to equivalent nondeterministic FBDDs \cite{BL15,Raz15}.
We prove that so-called structured DNNFs can be simulated by equivalent nondeterministic
OBDDs with only a quasipolynomial increase in representation size. Moreover, if the structured DNNF is deterministic
the result is an unambiguous nondeterministic OBDD. Since SDDs are restricted deterministic structured DNNFs, we are done.
%This transformation is tight because of a recently established quasipolynomial separation of nondeterministic 
%OBDDs from SDDs \cite{Raz17}.
%We complete his quasipolynomial separation by a new simulation of SDDs by equivalent unambiguous nondeterministic
%OBDDs with only a quasipolynomial increase in representation size in general.

Figure \ref{fig:landscape1} and \ref{fig:landscape2} illustrate the relative succinctness of some of the representation
types mentioned above.
$\mathcal{P}($OBDD)$\subsetneq \mathcal{P}($SDD) was shown in \cite{Bov16}. It is known that 
$\mathcal{P}($SDD)$\not\subseteq \mathcal{P}(\vee_1-$OBDD) (see \cite{Raz17} and \cite{Dar11,Raz16}).
We prove that the separation between 
$\mathcal{P}($SDD) and $\mathcal{P}(\vee_1-$OBDD) is only quasipolynomial.
The question whether $\mathcal{P}(\vee_1-$OBDD)$ \subsetneq \mathcal{P}($SDD) is open.

$\mathcal{P}($SDD)$\not\subseteq \mathcal{P}(\vee-$FBDD) can be proved with results in \cite{Dar11,Raz16}
but the separation is only quasipolynomial.  An exponential separation exists between 
$\mathcal{P}($SDD) and $\mathcal{P}($FBDD) and vice versa (see Section \ref{sec:storage} and \cite{BL15}).
%It is not hard to see that the quasipolynomial simulation of DNNFs by nondeterministic FBDDs
%proved by Beame and Liew \cite{BL15} leads for d-DNNFs to unambiguous nondeterministic FBDDs.
%Combining this observation with an exponential separation between
%nondeterministic FBDDs and unambiguous nondeterministic FBDDs \cite{Sau3a} we can conclude that
%$\mathcal{P}(\vee-$FBDD) $\not\subseteq \mathcal{P}($SDD).

\paragraph{{\bf Remarks}}
SDDs are structured w.r.t.\ so-called vtrees whose leaves are labeled by Boolean variables
and OBDDs respect  so-called variable orderings which are lists of variables (see Section \ref{sec2}).
Xue, Choi, and Darwiche showed a Boolean function whose SDD size w.r.t.\ a given vtree
$T$ is linear but whose OBDD size w.r.t.\ a variable ordering that 
corresponds to 
%an in-order traversal of 
a left-right traversal of the leaves in $T$ is exponential (Theorem 1 in \cite{XCD12}).
Their result demonstrates that for a space-efficient simulation of SDDs by equivalent unambiguous nondeterministic
OBDDs the choice of the variable ordering is not trivial. 
As a side effect, our quasipolynomial simulation of SDDs by equivalent unambiguous nondeterministic OBDDs 
presented in Section \ref{section:simulation_sdnnfs_by_vee_obdds} and in Section 
\ref{section:simulating_structured_d_dnnfs} generates a variable ordering from a given vtree. 
%Our quasipolynomial simulation 
%presented in Section \ref{section:simulation_sdnnfs_by_vee_obdds} and in Section 
%\ref{section:simulating_structured_d_dnnfs} generates 
For the SDD given in \cite{XCD12} it generates a variable ordering
for which the represented function has polynomial OBDD size.

\vspace{0.2cm}
Only recently, Cali, Capelli, and Razgon investigated two restricted variants of decision DNNFs, 
so-called structured decision DNNFs and so-called decomposable $\wedge$-OBDDs 
which are OBDDs augmented with decomposable $\wedge$-nodes \cite{CCR17}.
Since our quasipolynomial simulation of SDDs by equivalent unambiguous nondeterministic OBDDs 
generates a variable ordering from a given vtree, 
our constructon can be used to show that each structured decision
DNNF can be seen as a decomposable  $\wedge$-OBDD
of the same asymptotical size. This answers the question in \cite{CCR17} in the affirmative
whether a polynomial transformation from structured decision DNNFs to equivalent 
decomposable $\wedge$-OBDDs exists.
Moreover, our simulation shows that every function representable by 
decomposable $\wedge$-OBDDs can be represented by
OBDDs with only a quasipolynomial increase in representation size in general
(a fact already mentioned in \cite{OD15} but without proof).

\begin{figure}[!b]
\centering
\resizebox{0.6\textwidth}{!}{
\begin{tikzpicture}[scale=1.5]

\def\kreisAussen{(0,0) ellipse (4.5cm and 2.5cm)};
\def\kreisBlau{(-1,0) ellipse (2.5cm and 1.5cm)};
\def\kreisRot{(1,0) ellipse (2.5cm and 1.5cm)};
\def\kreisGelb{(0.2,0) ellipse (1cm and 0.75cm)};

\draw \kreisAussen;
\draw[dotted] \kreisBlau;
\draw[dashed] \kreisRot;
\draw[dashed, thick] \kreisGelb;

\node at (0,-2.7) {\large{$\mathcal{P}(d$-DNNF$)$}};
\node at (-1.3,-1.7) {\large{$\mathcal{P}($SDD$)$}};
\node at (1.3,-1.7) {\large{$\mathcal{P}(\lor_1$-OBDD$)$}};
\node at (0.1,-0.95) {\large{$\mathcal{P}($OBDD$)$}};

\node[draw,circle,inner sep=2pt] (raz) at (-2.5,0) {};
\node[draw,circle,inner sep=2pt] (here) at (-1.1,0) {};
\node[draw,circle,inner sep=2pt] (bl) at (2.5,0) {};

\node (qm) at (-3,2.5) {\large{\cite{Raz17} and \cite{Dar11,Raz16}}}; 
\draw[-stealth] (qm) to (raz);

\node (h) at (-1,3) {\large{\cite{Bov16}}};
\draw[-stealth] (h) to (here);

%\node[text width=3.5cm] (citation) at (3.75,2.5) {\large{?}};
\node (citation) at (3.75,2.5) {\large{?}};
\draw[-stealth] (citation) to (bl);

\end{tikzpicture}
}
\caption{On the relative succinctness of SDDs and (unambiguous nondeterministic) OBDDs.
%$\mathcal{P}($OBDD)$\subsetneq \mathcal{P}($SDD) was proved in \cite{Bov18}. It is known that 
%$\mathcal{P}($SDD)$\notsubseteq \mathcal{P}(\vee_1-$OBDD) \cite{Raz17} and \cite{Dar11,Raz16}
%but $\mathcal{P}(\vee_1-$OBDD)$ \subsetneq \mathcal{P}($SDD) is open.
}\label{fig:landscape1}
\end{figure}
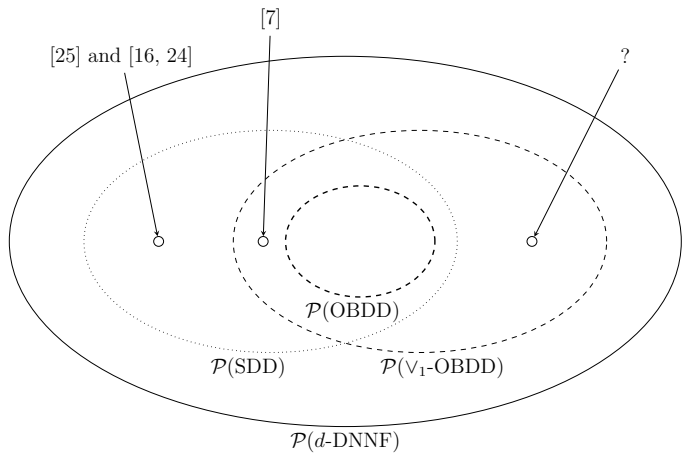

\begin{figure}[!t]
\centering
\resizebox{0.6\textwidth}{!}{
\begin{tikzpicture}[scale=1.5]

\def\kreisAussen{(0,0) ellipse (4.5cm and 2.5cm)};
\def\kreisBlau{(-1,0) ellipse (2.5cm and 1.5cm)};
\def\kreisRot{(1,0) ellipse (2.5cm and 1.5cm)};
\def\kreisGelb{(0.2,0) ellipse (1cm and 0.75cm)};
\def\kreisGrau{(-1.1,0) ellipse (1cm and 0.75cm)};

\draw \kreisAussen;
\draw[dotted] \kreisBlau;
\draw[dashed] \kreisRot;
\draw[dashed,thick] \kreisGelb;
\draw[dotted,thick] \kreisGrau;

\node at (0,-2.7) {\large{$\mathcal{P}($DNNF$)$}};
\node at (-1.3,-1.7) {\large{$\mathcal{P}(d $-DNNF$)$}};
\node at (1.3,-1.7) {\large{$\mathcal{P}(\lor$-FBDD$)$}};
\node at (-1.6,-0.95) {\large{$\mathcal{P}($SDD$)$}};
\node at (0.2,-0.95) {\large{$\mathcal{P}($FBDD$)$}};

\node[draw,circle,inner sep=2pt] (?) at (-1.8,0) {};
\node[draw,circle,inner sep=2pt] (here) at (-1.1,0) {};
\node[draw,circle,inner sep=2pt] (bl) at (0.5,0) {};

\node (qm) at (-2.5,2.5) {\large{\cite{Dar11,Raz16}}};
\draw[-stealth] (qm) to (?);
\node (h) at (-1,3) {\large{new}};
\draw[-stealth] (h) to (here);

\node (citation) at (2.5,2.8) {\large{\cite{BL15}}}; 
\draw[-stealth] (citation) to (bl);
\end{tikzpicture}
}
\caption{On the relative succinctness of \textup{SDDs} and \textup{FBDDs}.
%In Section \ref{sec4} we prove that the class
%$\mathcal{P}(\text{SDD})\setminus \mathcal{P}(\text{FBDD})$ is not empty.
}\label{fig:landscape2}
\end{figure}
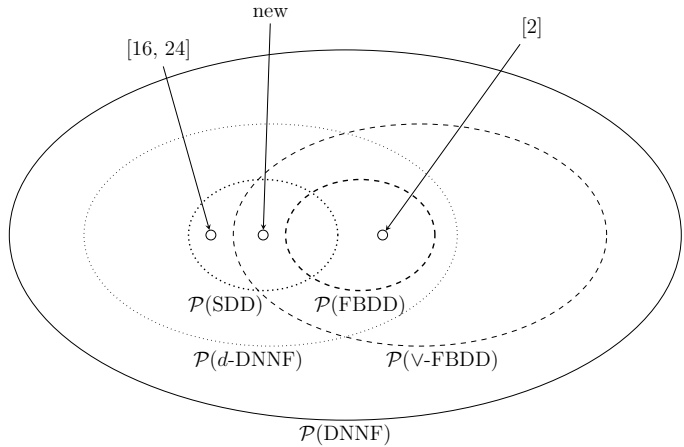

\paragraph{{\bf Organization of the paper}}

The rest of the paper is organized as follows.
In Section \ref{sec2} we recall the main definitions concerning binary decision diagrams and
decomposable negation normal forms. Moreover, important Boolean functions
which are discussed later on in the paper are formally defined.
The next sections contain our main results.
%This transformation is tight because of a recently established quasipolynomial separation of nondeterministic 
%OBDDs from SDDs \cite{Raz17}.
%On the other hand, 
In Section \ref{section:simulation_vee_one_obdds_by_sdds} it is shown that every Boolean function
$f$ for which $f$ and its negated function $\overline{f}$ can be represented by 
polynomial-size unambiguous nondeterministic OBDDs w.r.t.\ the same variable ordering can also be 
represented by SDDs of polynomial size.
Section \ref{section:simulation_sdnnfs_by_vee_obdds} and
Section \ref{section:simulating_structured_d_dnnfs} 
are devoted to the new quasipolynomial transformation from structured (deterministic) DNNFs into
equivalent (unambiguous) nondeterministic OBDDs. 
Section \ref{sec:storage} uses 
%this characterization 
the results from Section \ref{section:simulation_vee_one_obdds_by_sdds}
to derive small size SDDs for an 
important class of Boolean functions called strorage access functions.
Moreover, we obtain 
as a corollary the result that there are functions with polynomial SDD size but exponential FBDD size.
%a characterization of Boolean functions representable by SDDs of polynomial size.
%As a corollary we obtain the result that there are functions with polynomial SDD size but exponential FBDD size.
The proof that SDDs are more powerful w.r.t.\ polynomial-size representations than $k$-OBDDs for constant $k$,
a generalization of OBDDs, is shown in Section $\ref{sec:comparison}$.
This is done by demonstrating that Boolean functions representable by $k$-OBDDs of polynomial size,
where $k$ is a constant, can be represented by equivalent restricted unambiguous nondeterministic OBDDs
of polynomial size.
%fulfill the characterization described in Section $\ref{sec3}$.
%Finally, in Section $\ref{sec5}$ Beame's and Liew's simulation of DNNFs by equivalent
%nondeterministic FBDDs \cite{BL15} is inspected. It also
%yields a simulation of deterministic DNNFs by unambiguous nondeterministic FBDDs
%which have exactly one accepting path for every satisfying input.
Finally, we finish the paper with some open questions.
For readability some tedious technical proofs are delegated into the appendix.

%Furthermore, we give a short summary on the known results how the sizes of the two models
%are related. Section \ref{} presents some results on 
%%%%%%%%%%%%%%%%%%%%%%%%%%%%%%%%%%%%%%%%%%%%%%%%%%%%%%%%%%%%%%%%%%%%%%%%%%%
% Section 2
\section{Preliminaries}\label{sec2}

In the following we assume familiarity with fundamental concepts on circuits
(otherwise see, e.g., \cite{Vol99} and \cite{Weg87} for more details). 
In this section, we briefly recall the main notions concerning binary decision diagrams and decomposable
negation normal forms,
discuss the relation between ordered binary decision diagrams and sentential decision diagrams,
and introduce some Boolean functions.

\subsection{\bf Binary decision diagrams}

In complexity theory binary decision diagrams or in this area more often called branching programs
are a well established 
representation type for discrete functions and
the binary decision diagram size of a Boolean function is known 
to be a measure for the space complexity of nonuniform Turing machines and 
known to lie between the circuit size of the considered function 
and its $\{\wedge, \vee, \neg\}$-formula size (see, e.g., \cite{Weg87,Weg00}).

Since binary decision diagrams are a nonuniform model of computation, 
usually sequences of binary decision diagrams
$G=(G_n)$ representing sequences of Boolean functions $f=(f_n)$
are considered, where $f_n$ is defined on $n$ variables and $n\in \mathbb{N}$. 
In the following we simplify the notation for all nonuniform computation models
because the meaning is clear from the context.
Moreover, in the remaining part of the paper the size of a representation for a Boolean function 
refers to the number of variables
the function is defined on if nothing else is explicitly mentioned.

\begin{definition}[BDDs]
A {\em binary decision diagram} \textup{(BDD)}
on a variable set $X=\{x_1, \ldots, x_n\}$
is a directed acyclic graph with one source and sinks labeled
by the constants $0$ and $1$, respectively.
Each 
%non-sink 
internal node \textup{(}or {\em decision node}\textup{)} is labeled
by a Boolean variable
and has two outgoing edges, one labeled by $0$ and the other by $1$.
A {\em nondeterministic binary decision diagram}
\textup{(}$\vee$-\textup{BDD}\textup{)}
is a binary decision diagram with some
additional nodes called {\em nondeterministic nodes} \textup{(}$\vee$-nodes\textup{)}
whose outgoing edges are unlabeled.
%which have out-degree 2.

An input $b\in \{0,1\}^n$ activates all edges consistent with $b$,
i.e., the edges labeled by $b_i$ which leave nodes labeled by $x_i$
\textup{(}and all unlabeled edges in a nondeterministic binary decision diagram\textup{)}.
A {\em computation path} for an input $b$ in a \textup{BDD} is
a directed path of edges activated by the input $b$ that leads from the source to a sink.
A computation path 
for an input $b$ 
that leads to the 1-sink is called
{\em accepting path} for $b$.

Let $B_n$ denote the set of all Boolean functions defined on $n$ variables.
%$f: \{0,1\}^n \rightarrow \{0,1\}$. 
A \textup{(}nondeterministic\textup{)} \textup{BDD}
represents the function $f\in B_n$ 
for which $f(b)=1$ iff there exists an accepting path for the input $b$.
A nondeterministic \textup{BDD} is {\em unambiguous nondeterministic}, or a $\vee_1$-BDD for short,
iff there exists at most one accepting path for every input.

The {\em size} of a \textup{(}nondeterministic\textup{)} binary decision diagram $G$ is the number of
its nodes and is denoted by $|G|$.
The {\em \textup{(}nondeterministic\textup{)} binary decision diagram size} of a Boolean function $f$ is the size of a smallest
\textup{BDD} representing $f$.
\end{definition}

Our definition of the (nondeterministic) binary decision diagram size as the number of nodes 
and not the number of edges is justified 
because both numbers are polynomially related.

%At each node $v$ in a BDD a Boolean function $f_v\in B_n$
%is represented.
%A $c$-sink represents the constant function $c$.
%If $f_{v^1_0}, f_{v^2_0}, \ldots, f_{v^k_0}$ and $f_{v^1_1}, f_{v^2_1}, \ldots, f_{v^\ell_1}$ are the functions
%at the 0- or 1-successors of $v$, respectively, and $v$ is labeled by $x_i$,
%for an assignment $b$ to the input variables $f_v(b)$ is defined by Shannon's decomposition rule
%$$
%%f_v(b):= 
%\left(\overline{b_i} \wedge (f_{v^1_0}(b)\vee f_{v^2_0}(b) \vee \ldots \vee f_{v^k_0}(b)) \right)
%            \vee \left(b_i \wedge (f_{v^1_1}(b)\vee f_{v^2_1}(b) \vee \ldots \vee f_{v^\ell_1}(b))\right).
%$$

In many applications, such as symbolic verification or the analysis of circuits and automata,
data structures for Boolean functions are necessary that represent important functions in small size
and allow the efficient execution of important operations 
(for the choice of these operations and a discussion see, e.g., Section 10.2 in \cite{BSSW10} and \cite{Weg94}).
Since satisfiability test and equality check are two important operations that are NP-hard for general 
BDDs, restricted variants are considered.
FBDDs (with some restrictions) and $k$-OBDDs, where $k$ does not depend on the number of Boolean variables 
the represented function is defined on, 
allow polynomial time algorithms for important operations.
OBDDs introduced by Bryant \cite{Bry86} are restricted FBDDs and restricted $k$-OBDDs.
\begin{definition}
\begin{enumerate}
\item[\textup{(}i\textup{)}] A {\em free binary decision diagram} \textup{(FBDD)} or {\em read-once branching program}
           is a \textup{BDD} where each directed path contains for each variable at most one node 
           labeled by this variable.
            (See Figure \ref{bdds} for an example of an \text{FBDD}.)
%\item[(ii)] An ordered \textup{BDD} \textup{(OBDD)} 
%             is an \textup{FBDD} where the orderings of the
% variables on all paths are consistent with one ordering.
\item[\textup{(}ii\textup{)}] An {\em ordered binary decision diagram} \textup{(OBDD)} 
            is a binary decision diagram where on each directed
            path the node labels of the decision nodes are a subsequence of a given variable ordering
            $x_{\pi(1)}, x_{\pi(2)}, \ldots, x_{\pi(n)}$, where $\pi$ is a permutation on $\{1, \ldots, n\}$.
            (See Figure \ref{dsa} for an example of an \text{OBDD}.)
%\item[(iii)] A $k$-\textup{OBDD} consists of $k$ layers of \textup{OBDDs} respecting the
% same ordering.
\item[\textup{(}iii\textup{)}] A {\em $k$-\textup{OBDD}} is a binary decision diagram that can be partitioned 
             into $k$ layers. Each layer is an \textup{OBDD} \textup{(}with possibly many 
             sources\textup{)} such that the edges leaving the $i$-th layer, $1\leq i < k$, reach only nodes of a
             layer $j>i$ and the sinks. Moreover, all \textup{OBDDs} respect the same variable ordering 
             which means that on all directed paths in a layer 
            the node labels of the decision nodes are a subsequence of a given variable ordering and this
            ordering is the same for all layers. 
            (See Figure \ref{bdds} for an example of a $2$-\text{OBDD}.)
\end{enumerate}
\end{definition}

Nondeterministic 
variants of restricted BDDs can be defined similarly as for BDDs.
In the rest of the paper we consider $k$-OBDDs, where $k$ is a constant, 
if nothing else is mentioned.
Since a variable ordering can be identified with the corresponding permutation,
$\pi$ also denotes the ordering of the variables by abuse of notation.
%$\pi$ is sometimes by abuse of notation called variable ordering.
%if the meaning is clear from the context.

A {\em $1$-input} or {\em satisfying input} for a function $f$ is an assignment to the input variables 
whose function value is $1$, in other words this assignment is mapped to $1$ by $f$.
A function is {\em satisfiable} if there exists a satisfying input for $f$.
In the following, by abuse of notation we say that a (nondeterministic) BDD $G$ has a $1$-input or a satisfying input
if $G$ does not represent the constant $0$ function.

Since OBDDs are restricted FBDDs and restricted $k$-OBDDs by definition, 
$\mathcal{P}(\textup{OBDD})\subseteq \mathcal{P}(\textup{FBDD})$ and 
$\mathcal{P}(\textup{OBDD})\subseteq \mathcal{P}(k\textup{-OBDD})$. 
Moreover, we know that
$\mathcal{P}(\textup{OBDD})\subsetneq \mathcal{P}(\textup{FBDD})$ and 
$\mathcal{P}(\textup{OBDD})\subsetneq \mathcal{P}(k\textup{-OBDD})$.
The hidden weighted bit 
function HWB$_n$ defined below is an example of a Boolean function representable by
$2$-OBDDs and FBDDs of size $\mathcal{O}(n^2)$ but its OBDD size is $\Omega(2^{n/5})$
(\cite{BLSW99} and \cite{SW95}).
It is well-known that the complexity classes 
$\mathcal{P}(\textup{FBDD})$ 
and 
$\mathcal{P}(k\textup{-OBDD})$
%$\mathcal{P}(k\textup{-OBDD})$, where $k$ is a constant, 
are incomparable which means 
$\mathcal{P}(\textup{FBDD}) \not\subseteq \mathcal{P}(k\textup{-OBDD})$ and
$\mathcal{P}(k\textup{-OBDD})\not\subseteq \mathcal{P}(\textup{FBDD})$.
Moreover,
there are Boolean functions representable in polynomial size by one model 
but only in exponential size by the other one and vice versa (see, e.g., \cite{Weg00}).
The same result holds for $\mathcal{P}($FBDD) and $\mathcal{P}(\vee_1$-OBDD). 
%where $\vee_1$-OBDD denotes an unambiguous nondeterministic OBDD.

\subsection{\bf Decomposable negation normal forms}

Many known representations of propositional knowledge bases 
are restricted negation normal form circuits (NNFs)
and correspond to specific properties on NNFs \cite{DM02}. 
Decomposability and determinism are two of these fundamental properties.

\begin{definition}[NNFs]
A {\em negation normal form circuit} on a variable set 
%$X=\{x_1, \ldots, x_n\}$
$X$
is a Boolean circuit over fanin 2 conjunction and unbounded fanin disjunction gates, 
labeled by $\wedge$ and $\vee$,
whose inputs are labeled by literals 
%$x_i$ and $\overline{x}_i$, $i\in \{1, \ldots, n\}$,
$x$ and $\overline{x}$, $x\in X$,
and the Boolean constants $0$ and $1$.
The {\em size} of an \textup{NNF} $C$, denoted by $|C|$, is the number of its gates.
The {\em \textup{NNF} size} of a Boolean function $f$ is the size of a smallest
negation normal form circuit representing $f$.
The Boolean function $f_C:\{0,1\}^X\rightarrow \{0,1\}$ represented by $C$ is defined in the usual way.
For an \textup{NNF} $C$ and a gate $g$ in $C$ the subcircuit rooted at $g$ is denoted by $C_g$.
An \textup{NNF} is {\em decomposable}, or a \textup{DNNF} for short, 
iff the children of each $\wedge$-gate are reachable from disjoint sets of input variables.
A set of Boolean functions $\{f_1, \ldots, f_\ell\}$ on the same variable set is {\em disjoint}
if each pair of functions $f_i,f_j$, $i\not= j$, 
is not simultaneously satisfiable.
A \textup{DNNF} is {\em deterministic}, or a $d$-\textup{DNNF} for short, iff the functions 
computed at the children of each $\vee$-gate are disjoint.
%not simultaneously satisfiable.
\end{definition}

Our assumption that each $\wedge$-gate has only fan-in $2$ is justified because it
affects the NNF size only polynomially.

Sentential decision diagrams introduced by Darwiche \cite{Dar11} 
result from so-called structured decomposability and strong determinism.
They are restricted $d$-DNNFs and a generalization of OBDDs.

\begin{definition}
For a variable set $X$ let $\bot:\{0,1\}^X\rightarrow \{0,1\}$ and
$\top:\{0,1\}^X\rightarrow \{0,1\}$ denote the constant $0$ function
and constant $1$ function, respectively.
A set of Boolean functions $\{f_1, \ldots, f_\ell\}$ on the same variable set is called a {\em partition}
iff the functions $\{f_1, \ldots, f_\ell\}$ are disjoint, none of the functions is the constant $0$ function $\bot$,
and $\bigvee\limits_{i=1}^\ell f_i=\top$.
\end{definition}

\begin{definition}\label{def:sdds}
A  {\em vtree} for a variable set $X$ is a full, rooted binary tree 
whose leaves are in one-to-one correspondence with the variables in $X$.
A {\em sentential decision diagram} $C$, or \textup{SDD} for short, respecting a vtree $T$
on the variable set $X=\{x_1, \ldots, x_n\}$
is defined inductively in the follwing way:
\begin{itemize}
\item $C$ represents $\bot$ or $\top$ or $C$ represents a projective function $p(X)=x_i$ or $p(X)=\overline{x}_i$,
      $1\leq i \leq n$.
\item The output gate of $C$ is a disjunction whose inputs are wires from $\wedge$-gates $g_1, \ldots, g_\ell$,
      where each $g_i$ has wires from $p_i$ and $s_i$,
      $v$ is an internal node in $T$ with children $v_L$ and $v_R$,  
      $C_{p_1}, \ldots, C_{p_\ell}$
      %$p_1, \ldots, p_\ell$ 
       are \textup{SDDs} that respect the subtree of $T$ rooted at $v_L$,
      $C_{s_1}, \ldots, C_{s_\ell}$
      %$s_1, \ldots, s_\ell$ 
      are \textup{SDDs} that respect the subtree of $T$ rooted at $v_R$,
      and the functions represented by $C_{p_1}, \ldots, C_{p_\ell}$ are a partition.
\end{itemize}
%A constant or literal \textup{SDD} is called \textup{terminal}. Otherwise it is called a \textup{decomposition}.
\end{definition}

Vtrees were introduced by Pipatsrisawat and Darwiche \cite{PD08}.
The ordering w.r.t.\ a vtree and the so-called \emph{partition property} 
ensure that SDDs are decomposable and deterministic and therefore, restricted $d$-DNNFs. 
The partition property is also called \emph{strong determinism}.
It ensures that $\mathcal{P}($SDD) is closed under negation which means
that for each function $f$ representable by polynomial-size SDDs also the negated function
$\overline{f}$ is in $\mathcal{P}($SDD).
To the best of our knowledge it is open whether SDDs are even more restricted 
in the sense of polynomial-size representations
than structured $d$-DNNFs which are $d$-DNNFs respecting a vtree.

\begin{definition}\label{definition:dnnf_respecting_vtree}
For a node $u$ let $vars(u)$ denote the set of variables that appear in a subgraph rooted at $u$.
Let $T$ be a vtree for the set of variables $X$ and $\mathcal{D}$ be a $\textnormal{DNNF}$. 
$\mathcal{D}$ \emph{respects} the vtree $T$, if for every $\wedge$-node $u$ of $\mathcal{D}$ 
with children $u_l, u_r$, there is a node $v$ of $T$ with children $v_l, v_r$ such that 
$\vars{u_l} \subseteq \vars{v_l}$ and $\vars{u_r} \subseteq \vars{v_r}$.

A (deterministic) $\textnormal{DNNF}$ that respects a given vtree $T$ is called a (deterministic) 
$\textnormal{DNNF}_T$. 
Moreover, a \emph{structured} (deterministic) $\textnormal{DNNF}$, or (deterministic) SDNNF for short,
is a (deterministic) $\textnormal{DNNF}_T$ for an arbitrary vtree $T$.
\end{definition}

Note that for each $\wedge$-node $u$ in Definition \ref{definition:dnnf_respecting_vtree} 
there is only one node $v$ of $T$ fulfilling the requirement mentioned above. 
We call $v$ the \emph{decomposition node} of $u$ and $\dnode{u} = v$. 

%In the rest of the paper, we look at SDDs as a class of Boolean circuits.
In the rest of the paper, we look at (restricted) NNFs as classes of Boolean circuits.

\subsection{\bf On the relation between OBDDs and SDDs}

A vtree is linear if for every internal node one child is a leaf.
It is right-linear if for every internal node the left child is a leaf.
%Similarly, left-linear vtrees are defined.
In the following let $T_\pi$ be a vtree whose 
%in-order traversal which means a 
left-right traversal of the leaves in $T$
corresponds to
the variable ordering $\pi$.
%A variable ordering $x_{\pi(1)}, x_{\pi(2)}, \ldots, x_{\pi(n)}$ can be represented by the 
%right-linear vtree $T_\pi$ whose in-order traversal corresponds to the variable ordering. 
OBDDs are based on the Shannon decomposition 
$$f=\overline{x}_if_{|x_i=0} \vee x_i f_{|x_i=1},$$
where $f_{|x_i=c}$ denotes the subfunction of $f$ obtained by replacing the Boolean variable $x_i$ 
by the Boolean constant $c$. Since the subfunctions $f_{|x_i=0}$ and $f_{|x_i=1}$ do not essentially depend
on the variable $x_i$, i.e., there is no assignment to the remaining variables such that 
the function values for $x_i=0$ and $x_i=1$ differ, and
the disjunction of the projective functions $p_0=\overline{x}_i$ and $p_1=x_i$ 
is the constant function $\top$ but their conjunction is the function $\bot$,
%Already Darwiche \cite{Dar11} noticed that 
OBDDs respecting the variable ordering $\pi$ can be seen as restricted SDDs w.r.t.\ the right-linear
vtree $T_\pi$ and vice versa (see also \cite{Dar11}).
Figure \ref{dsa} shows an OBDD for a Boolean function 
w.r.t.\ the variable ordering $\pi= a_1,a_0, x_0, x_1, x_2, x_3$, 
Figure \ref{vtree} illustrates the corresponding right-linear vtree $T_\pi$ and an SDD respecting $T_\pi$
for the same Boolean function.

Structured decomposability on the notion of vtrees was originally introduced by 
Pipatsrisawat and Darwiche \cite{PD08}
but without distinction between the left and right child of a node.
Xue, Choi, and Darwiche showed that switching the left and right child of a vtree node
may lead to an exponential change in the size of the corresponding SDDs \cite{XCD12}.
%They showed that simply swapping a pair of children can sometimes lead to exponential differences in 
%the corresponding SDD size.
An SDD w.r.t.\ a 
%left-linear 
linear vtree $T_{\pi}$ can be seen as an unambiguous nondeterministic OBDD repecting
$\pi$. 
Since it is well-known that
$\mathcal{P}(\textup{OBDD})\subsetneq \mathcal{P}(\vee_1$-\textup{OBDD}), 
it is not astonishing that swapping the children of nodes in a vtree may lead to an exponential blow-up
in the representation size. We will see in Section \ref{sec:comparison} that SDDs respecting 
linear vtrees can represent all Boolean functions in $\mathcal{P}(k$-\textup{OBDD}) in polynomial size.

\subsection{\bf Storage access functions}

%%In the following let $|z|_2:= \sum_{i=0}^{m-1} z_i2^i$, where $z=(z_{m-1}, \ldots, z_0)\in \{0,1\}^m$.
%We investigate Boolean functions modeling different aspects of storage access.
In the BDD literature Boolean functions modeling different aspects of storage access are well investigated.
A storage access sometimes also called pointer function outputs a single bit of the input
for which 
% as a result, 
%and 
the
address or index 
%of this bit 
is also computed from the input.
A very simple one is 
the {\it multiplexer} function MUX$_n$ (alternative names are  {\it direct storage access function}
%DSA$_n$ 
or {\it index function}) 
%IND$_n$). 
that is defined on $n+k$ variables
 $a_{k-1},\dots ,a_0,x_0,\dots ,x_{n-1}$, where $n=2^k$.
The function is given as $\text{MUX}_n(a,x)=x_{|a|_2}$, 
where $|a|_2$ is the number in $\mathbb{N}$ whose binary
 representation equals $(a_{k-1},\dots ,a_0)$. (See Figure \ref{dsa} for an example of an OBDD representing
MUX$_4$.)

\begin{figure}[!ht]
\centering
\resizebox{0.3\textwidth}{!}{
%\resizebox{0.3\textwidth}{!}{
\begin{tikzpicture}

\node[draw, circle] at (0,0) (a) {$a_1$};
\node[draw, circle] at (-1.5,-2) (b1) {$a_0$};
\node[draw, circle] at (1.5,-2) (b2) {$a_0$};

\draw[-stealth, dashed] (a) -- (b1);
\draw[-stealth] (a) -- (b2);

\node[draw, circle] at (-2.25, -4) (c1) {$x_0$};
\node[draw, circle] at (-0.75, -5) (c2) {$x_1$};
\node[draw, circle] at (0.75, -6) (c3) {$x_2$};
\node[draw, circle] at (2.25, -7) (c4) {$x_3$};

\draw[-stealth, dashed] (b1) -- (c1);
\draw[-stealth] (b1) -- (c2);
\draw[-stealth, dashed] (b2) -- (c3);
\draw[-stealth] (b2) -- (c4);

%\node[draw, rectangle, minimum size=0.65cm] at (-1.5, -9) (d0) {$0$};
%\node[draw, rectangle, minimum size=0.65cm] at (1.5, -9) (d1) {$1$};
\node[draw, rectangle, minimum size=0.62cm] at (-1.5, -9) (d0) {$0$};
\node[draw, rectangle, minimum size=0.62cm] at (1.5, -9) (d1) {$1$};

\draw[-stealth, dashed] (c1) -- (d0.125);
\draw[-stealth, dashed] (c2) -- (d0.100);
\draw[-stealth, dashed] (c3) -- (d0.75);
\draw[-stealth, dashed] (c4) -- (d0.50);
\draw[-stealth] (c1) -- (d1.125);
\draw[-stealth] (c2) -- (d1.100);
\draw[-stealth] (c3) -- (d1.75);
\draw[-stealth] (c4) -- (d1.50);

\end{tikzpicture}
}
\caption{An OBDD for the Boolean function MUX$_4$ w.r.t.\ the 
variable ordering $a_1,a_0, x_0, x_1, x_2, x_3$.
Dashed lines represent edges with label 0 and solid ones represent edges with label 1.
}\label{dsa}
\end{figure}
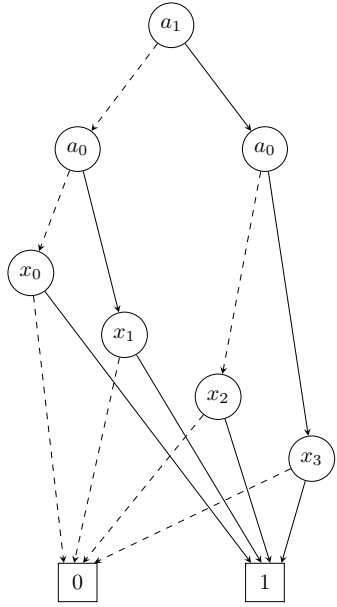

\begin{figure}[!ht]
%\resizebox{0.9\textwidth}{!}{
\resizebox{0.85\textwidth}{!}{
\begin{minipage}[b]{0.3\linewidth}
\centering
\begin{tikzpicture}[minimum size = 0.75cm, level distance=1.5cm, sibling distance=1cm, every node/.style={thin},
    strong/.style={edge from parent/.style={thick,draw}}]
\node[draw,circle] (wurzel) {}
child[strong] {
	node[draw,circle]{$a_1$}
}
child[strong] {
        node[draw,circle] {}
        child {
                node[draw,circle] {$a_0$}
        }
        child {
                node[draw,circle] {}
                 child {
					node[draw,circle] {$x_0$}
				}
				child {
					node[draw,circle] {}
					 child {
						node[draw,circle] {$x_1$}
					}
					 child {
						node[draw,circle] {}
						 child {
						 	node[draw,circle] {$x_2$}
						 }
						 child {
						 	node[draw,circle] {$x_3$}
						 }					}
				}
        }
};
\end{tikzpicture}
%\caption{Text für Baum links}\label{fig:links}
\end{minipage}
\hphantom{XXXX}
\hfill
\begin{minipage}[b]{0.69\linewidth}
\centering
\begin{tikzpicture}[minimum size = 0.75cm]
\tikzset{level 1/.style={sibling distance=4.5cm},
	level 2/.style={sibling distance=2cm},
	level 3/.style={sibling distance=2cm},
	level 4/.style={sibling distance=1cm},
	every node/.style={thin},
    strong/.style={edge from parent/.style={thick,draw}}};
	
\node[draw,circle] (wurzel) {$\lor$}
child[strong] {
	node[draw,circle]{$\land$}
	 child {
	 	node[draw,circle] {$\lor$}
	 	 child {
	 	 	node[draw,circle] {$\land$}
	 	 	 child {
	 	 	 	node[draw,circle] {$x_0$}
	 	 	 }
	 	 	 child {
	 	 	 	node[draw,circle] {$\overline{a}_0$}
	 	 	 }
	 	 }
	 	 child {
	 	 	node[draw,circle] {$\land$}
	 	 	 child {
	 	 	 	node[draw,circle] {$x_1$}
	 	 	 }
	 	 	 child {
	 	 	 	node[draw,circle] {$a_0$}
	 	 	 }
	 	 }
	 }
	 child {
	 	node[draw,circle] {$\overline{a}_1$}
	 }
}
child[strong] {
        node[draw,circle] {$\land$}
        child {
                node[draw,circle] {$\lor$}
                 child {
                 	node[draw,circle] {$\land$}
                 	 child {
                 	 	node[draw,circle] {$x_2$}
                 	 }
                 	 child {
                 	 	node[draw,circle] {$\overline{a}_0$}
                 	 }
                 }
                 child {
                 	node[draw,circle] {$\land$}
                 	 child {
                 	 	node[draw,circle] {$x_3$}
                 	 }
                 	 child {
                 	 	node[draw,circle] {$a_0$}
                 	 }
                 }
        }
        child {
                node[draw,circle] {$a_1$}
        }
};
\end{tikzpicture}
%\caption{Text für Baum rechts}\label{fig:rechts}
\end{minipage}
}
\caption{A right-linear vtree 
%$T$ 
whose 
left-right traversal of the leaves corresponds to the variable ordering $a_1,a_0,x_0,x_1,x_2,x_3$
and an \textup{SDD} for the Boolean function \textup{MUX}$_4$ w.r.t.\ this vtree.
% $T$.
}\label{vtree}
\end{figure}
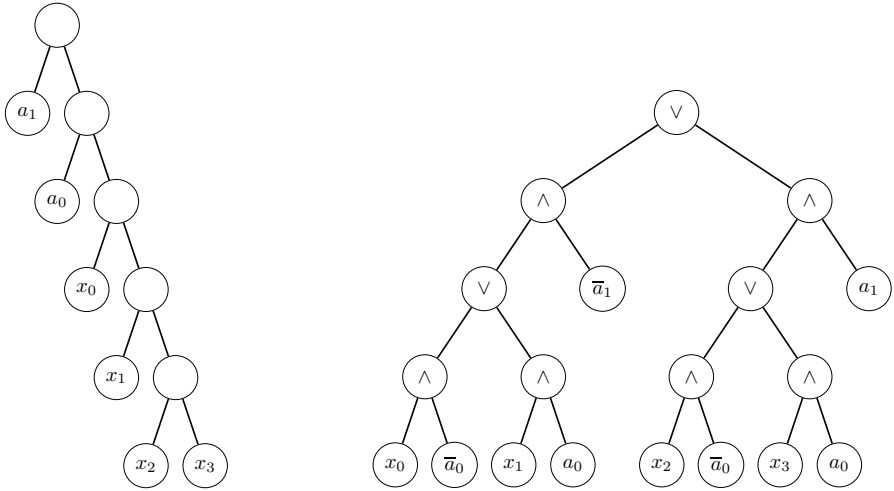

The following three Boolean functions 
%well-known in the BDD literature
are generalized storage access functions, 
where variables may serve as address as well as data variables.
The {\it hidden weigthed bit function} HWB$_n$ is defined by
$$
\textup{HWB}_n(x_1, \ldots, x_n)= x_{\|x\|},
$$
where $\|x\|= x_1 + \cdots + x_n$ is the number of variables set to $1$ in the input $x$ and $x_0:=0$
which means that the output is $0$ if $x_1 + \cdots + x_n=0$.
HWB$_n$ is an example of a function with a clear and simple structure,
nevertheless the OBDD size is exponential \cite{Bry91}. 
(See Figure \ref{bdds} for restricted BDDs representing the function HWB.)

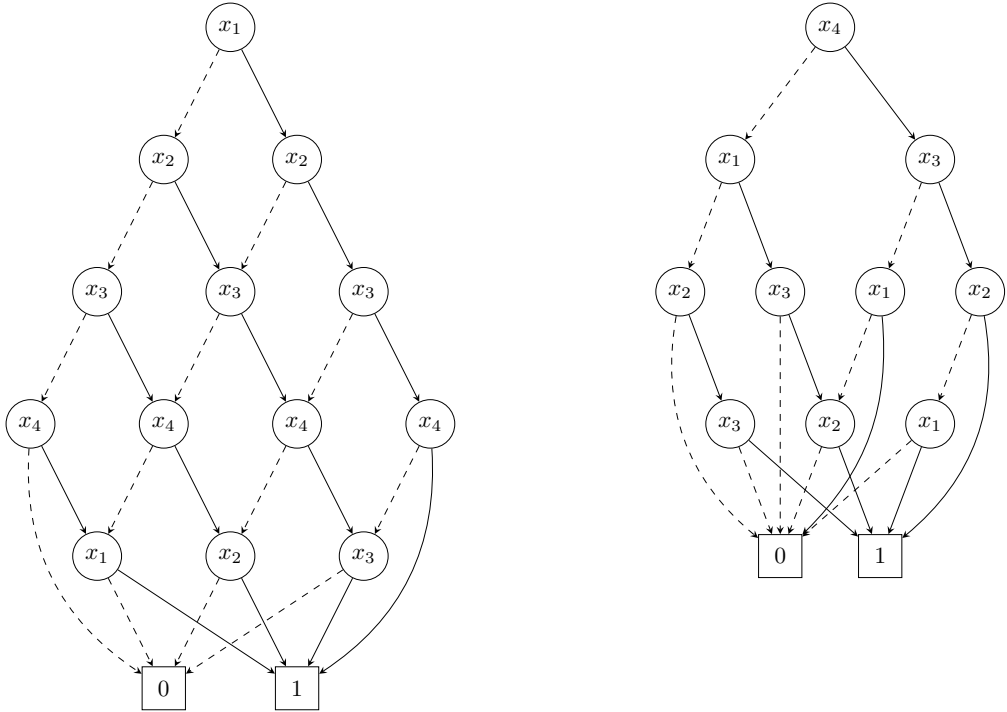
\begin{figure}[!ht]
\centering
%\resizebox{1.0\textwidth}{!}{
\resizebox{0.89\textwidth}{!}{
\begin{tikzpicture}

\node[draw, circle] at (0,0) (a) {$x_1$};

\node[draw, circle] at (-1,-2) (b) {$x_2$};
\node[draw, circle] at (1,-2) (c) {$x_2$};

\node[draw, circle] at (-2,-4) (d) {$x_3$};
\node[draw, circle] at (0,-4) (e) {$x_3$};
\node[draw, circle] at (2,-4) (f) {$x_3$};

\node[draw, circle] at (-3,-6) (g) {$x_4$};
\node[draw, circle] at (-1,-6) (h) {$x_4$};
\node[draw, circle] at (1,-6) (i) {$x_4$};
\node[draw, circle] at (3,-6) (j) {$x_4$};

\node[draw, circle] at (-2,-8) (k) {$x_1$};
\node[draw, circle] at (0,-8) (l) {$x_2$};
\node[draw, circle] at (2,-8) (m) {$x_3$};

\node[draw, rectangle, minimum size=0.65cm] at (-1,-10) (null) {$0$};
\node[draw, rectangle, minimum size=0.65cm] at (1,-10) (eins) {$1$};

\draw[-stealth, dashed] (a) -- (b);
\draw[-stealth] (a) -- (c);

\draw[-stealth, dashed] (b) -- (d);
\draw[-stealth] (b) -- (e);
\draw[-stealth, dashed] (c) -- (e);
\draw[-stealth] (c) -- (f);

\draw[-stealth, dashed] (d) -- (g);
\draw[-stealth] (d) -- (h);
\draw[-stealth, dashed] (e) -- (h);
\draw[-stealth] (e) -- (i);
\draw[-stealth, dashed] (f) -- (i);
\draw[-stealth] (f) -- (j);
\draw[-stealth] (g) -- (k);
\draw[-stealth, dashed] (h) -- (k);
\draw[-stealth] (h) -- (l);
\draw[-stealth, dashed] (i) -- (l);
\draw[-stealth] (i) -- (m);
\draw[-stealth, dashed] (j) -- (m);

\draw[-stealth, dashed] (g) to[bend right] (null);
\draw[-stealth, dashed] (k) -- (null);
\draw[-stealth] (k) -- (eins);
\draw[-stealth, dashed] (l) -- (null);
\draw[-stealth] (l) -- (eins);
\draw[-stealth, dashed] (m) -- (null);
\draw[-stealth] (m) -- (eins);
\draw[-stealth] (j) to[bend left] (eins);

\begin{scope}[xshift=9cm,yshift=0cm]
\node[draw, circle] at (0,0) (a) {$x_4$};
\node[draw, circle] at (-1.5,-2) (b1) {$x_1$};
\node[draw, circle] at (1.5,-2) (b2) {$x_3$};

\draw[-stealth, dashed] (a) -- (b1);
\draw[-stealth] (a) -- (b2);

\node[draw, circle] (c1) at (-2.25,-4) {$x_2$};
\node[draw, circle] (c2) at (-0.75,-4) {$x_3$};
\node[draw, circle] (c3) at (0.75,-4) {$x_1$};
\node[draw, circle] (c4) at (2.25,-4) {$x_2$};

\draw[-stealth, dashed] (b1) -- (c1);
\draw[-stealth] (b1) -- (c2);
\draw[-stealth, dashed] (b2) -- (c3);
\draw[-stealth] (b2) -- (c4);

\node[draw, circle] (d1) at (-1.5, -6) {$x_3$};
\node[draw, circle] (d2) at (0,-6) {$x_2$};
\node[draw, circle] (d3) at (1.5, -6) {$x_1$};
\node[draw, rectangle, minimum size=0.65cm] at (-0.75, -8) (e0) {$0$};
\node[draw, rectangle, minimum size=0.65cm] at (0.75, -8) (e1) {$1$};

\draw[-stealth, dashed] (c1) to[bend right] (e0);
\draw[-stealth] (c1) -- (d1);
\draw[-stealth, dashed] (c2) -- (e0);
\draw[-stealth] (c2) -- (d2);
\draw[-stealth, dashed] (c3) -- (d2);
\draw[-stealth] (c3) to[bend left=25] (e0);
\draw[-stealth, dashed] (c4) -- (d3);
\draw[-stealth] (c4) to[bend left] (e1);

\draw[-stealth, dashed] (d1) -- (e0);
\draw[-stealth] (d1) -- (e1);
\draw[-stealth, dashed] (d2) -- (e0);
\draw[-stealth] (d2) -- (e1);
\draw[-stealth, dashed] (d3) -- (e0);
\draw[-stealth] (d3) -- (e1);

\end{scope}

\end{tikzpicture}
}
\caption{A $2$-OBDD w.r.t.\ the variable ordering 
$x_1,x_2,x_3,x_4$ and an FBDD for the function HWB$_4$.
Dashed lines
represent edges with label 0 and solid ones edges with label 1. (See also \cite{BSSW10}.)}\label{bdds}
\end{figure}

The {\it indirect storage access function} ISA$_n$ can be described in the following way.
Let $n=2^k$, $k=2^\ell$, and $m=n/k=2^{k-\ell}$. ISA$_n$ is defined on $n+k-\ell$ Boolean variables,
%$a=(a_{k-\ell -1}, \ldots, a_0)$ and $x=(x_0, \ldots, x_{n-1})$. 
an address vector $a=(a_{k-\ell -1}, \ldots, a_0)$ and a vector
$x=(x_0, \ldots, x_{n-1})$. 
The address vector is interpreted
as the binary number with value $|a|_2$ pointing to a block
$x(a)=(x_{|a|_2k}, \ldots, x_{(|a|_2 +1)k-1})$.
%, where the indices are taken modulo $n$.
Then
$$
\textup{ISA}_n(a,x)=x_{|x(a)|_2}.
$$
The function ISA$_n$ has small size representation for BDD models like FBDDs and $2$-OBDDs 
but its OBDD size is exponential \cite{BHR95}.
To be more precisely its FBDD and $2$-OBDD size is $\mathcal{O}(n^2)$ but its OBDD size is 
$\Omega(2^{\lfloor n/\log n \rfloor})$.
%The indirect storage access function is a typical example for a function 
%that is hard for deterministic OBDDs, which means that its OBDD size is exponential.

Another kind of storage access or pointer function is the following one.
Let $p$ be the smallest prime larger than $n$.
The function {\it weighted sum} WS$_n$ is defined  by
$$
\textup{WS}_n(x_1, \ldots, x_n)= x_s,
$$
where $s$ is the sum of all $ix_i$ in the field $\mathbb{Z}_p$, $1\leq i \leq n$, if this sum is between $1$ and $n$ 
and $1$ otherwise. 
The weighted sum function was introduced and analyzed by 
Savick\'y and \u{Z}\'ak \cite{SZ00} in order to prove a lower bound of order $2^{n-o(1)}$ on the 
FBDD size of a Boolean function. It is not difficult to see that the $2$-OBDD size of WS$_n$ is $\mathcal{O}(n^2)$.
%%%%%%%%%%%%%%%%%%%%%%%%%%%%%%%%%%%%%%%%%%%%%%%%%%%%%%%%%%%%%%%%%%%%%%%%%%%
%\section{Simulating $\vee_1$-OBDDs by SDDs}
\section{Simulating Unambiguous Nondeterministic OBDDs by SDDs}\label{section:simulation_vee_one_obdds_by_sdds}
In this section, we will examine the relationship between unambiguous nondeterministic OBDDs and SDDs. More precisely, we will derive a way of representing a Boolean function $f$ as an SDD provided that $f$ and $\overline{f}$ can both be represented by unambiguous nondeterministic OBDDs which respect a common variable ordering.

\subsection{\bf Main ideas and simulation}
%\label{subsection:simulation_vee_one_obdds_by_sdds_main_ideas}

Let $\mathcal{F}_u$ denote the subgraph of a given BDD $\mathcal{F}$ rooted at node $u$ and let $f_u$ be the Boolean function which is represented by $\mathcal{F}_u$.
In order to avoid corner cases, we will assume that the given unambiguous nondeterministic OBDDs are of the following form.

\begin{definition}
	Let $\mathcal{F}$ be an unambiguous nondeterministic OBDD. We call $\mathcal{F}$ \emph{simple}, if
	\begin{itemize}
		\item there exist no edges between $\vee$-nodes,
		\item all $\vee$-nodes have at least two children,
		\item no $\vee$-node is connected to a sink, and 
		\item for each inner node $u$ of $\mathcal{F}$ holds that $\mathcal{F}_u$ does not represent the constant function $\top$ or $\bot$.
	\end{itemize}
\end{definition}

%\noindent
Observe that for each unambiguous nondeterministic OBDD that has polynomial size there exists a simple one of polynomial size representing the same function. Furthermore, we will assume w.l.o.g. that the variable ordering is given by the list of variables $x_1, \dots, x_n$ in the rest of this section. Next, we will present the main ideas of the simulation.\\

%\noindent
Let $f$ and $\overline{f}$ be Boolean functions that can be represented by unambiguous nondeterministic OBDDs $\mathcal{F}$ and $\overline{\mathcal{F}}$, respectively. Moreover, assume $\mathcal{F}$ and $\overline{\mathcal{F}}$ respect a common variable ordering. Darwiche already mentioned how a (deterministic) OBDD can be converted to an equivalent SDD respecting a right-linear vtree \cite{Dar11}. Therefore, the main question is how to deal with the occurrence of $\vee$-nodes in $\mathcal{F}$. Let $f_u$ be the Boolean function that is computed at an $\vee$-node $u$ of $\mathcal{F}$. Since $u$ can occur at any position in the given unambiguous nondeterministic OBDD $\mathcal{F}$, we would like to derive a way of representing $f_u$ by an SDD. Let $f_{u_1}, \dots, f_{u_k}$ be the functions that are represented at the child nodes of $u$. Due to the assumed variable ordering, we know that the functions $f_{u_1}, \dots, f_{u_k}$ essentially depend on a subset of variables $Y = \{x_i, \dots, x_n\} \subseteq X$ for $i \geq 1$. The function $f_u$ can be represented by $f_u = (f_{u_1} \wedge \top) \vee (f_{u_2} \wedge \top) \vee \dots \vee (f_{u_k} \wedge \top)$. However, for an SDD representing $f_u$ in such a way it would not be guaranteed that $f_{u_1}, \dots, f_{u_k}$ form a partition. Hence, the main idea is to find further functions represented at inner nodes of $\mathcal{F}$ and $\overline{\mathcal{F}}$ which essentially depend on $Y$ and together with $f_{u_1}, \dots, f_{u_k}$ yield a partition.\\

%\noindent
We use the notation $f_{|x_1 = c_1, \dots, x_{i-1} = c_{i-1}}$ for the subfunction that emerges of $f$ by replacing all occurrences of $x_1, \dots, x_{i-1}$ by constants $c_1, \dots, c_{i-1} \in \{0,1\}$. Now, observe that the subfunctions $f_{|x_1 = c_1, \dots, x_{i-1} = c_{i-1}}$ and $\overline{f}_{|x_1 = c_1, \dots, x_{i-1} = c_{i-1}}$ yield a partition for arbitrary assignments of the variables $x_1, \dots, x_{i-1}$. Fix an $\vee$-node $u$ of $\mathcal{F}$. Define $\beta(u)$ to be the set of variable assignments over $X \backslash Y = \{x_1, \dots, x_{i-1}\}$ which can be extended by an assignment of the variables of $Y = \{x_i, \dots, x_n\}$ such that there exists an accepting path containing $u$ for the resulting assignment in $\mathcal{F}$. For an arbitrary assignment $\beta \in \beta(u)$ with $\beta = (\beta_1, \dots, \beta_{i-1}) \in \{0,1\}^{i-1}$ we get the relation $f_u \leq f_{|x_1 = \beta_1, \dots, x_{i-1} = \beta_{i-1}}$ which means that the satisfying assignments of $f_u$ are a subset of the satisfying assignments of $f_{|x_1 = \beta_1, \dots, x_{i-1} = \beta_{i-1}}$ and $f_u \neq \bot$.\\

%\noindent
Next, we want to identify all nodes $u_1', \dots, u_l'$ in $\mathcal{F}$ for a fixed $\beta \in \beta(u)$ such that $f_{u_j'} \leq f_{|x_1 = \beta_1, \dots, x_{i-1} = \beta_{i-1}}$ and $f_{u_j'} \neq \bot$ hold. In order to get these nodes, we consider each node $u'$ in $\mathcal{F}$ with $\vars{u'} \subseteq Y$ such that there is no other node $u''$ fulfilling $\vars{u'} \subset \vars{u''} \subseteq Y$ and $\mathcal{F}_{u'}$ is a subgraph of $\mathcal{F}_{u''}$. Afterwards, for each resulting candidate $u'$ we check whether $\beta$ can be extended by an assignment of the variables of $Y$ such that there is an accepting path in $\mathcal{F}$ containing $u'$. If $u'$ is an $\vee$-node, we add the children of $u'$ instead to our set of nodes since we want to resolve $\vee$-nodes of $\mathcal{F}$.\\

%\noindent
Let $f_{u_1'}, \dots, f_{u_l'}$ be the Boolean functions that are represented at the nodes $u_1', \dots, u_l'$ in $\mathcal{F}$. Then, $f_{|x_1 = \beta_1, \dots, x_{i-1} = \beta_{i-1}} = f_{u_1} \vee \dots \vee f_{u_k} \vee f_{u_1'} \vee \dots \vee f_{u_l'}$. Analogously, we identify nodes $v_1, \dots, v_m$ of $\overline{\mathcal{F}}$ such that $\overline{f}_{v_j} \leq \overline{f}_{|x_1 = \beta_1, \dots, x_{i-1} = \beta_{i-1}}$ and $\overline{f}_{v_j} \neq \bot$ where the functions represented at the nodes $v_1, \dots, v_m$ of $\overline{\mathcal{F}}$ are denoted by $\overline{f}_{v_1}, \dots, \overline{f}_{v_m}$. Hence, we get $\overline{f}_{|x_1 = \beta_1, \dots, x_{i-1} = \beta_{i-1}} = \overline{f}_{v_1} \vee \dots \vee \overline{f}_{v_m}$. Now, we are able to represent the function calculated at the $\vee$-node $u$ as 
\vspace{-0.18cm}
\begin{eqnarray}
	\label{equation:vee_one_obdd_to_sdd}
	f_u &=& (f_{u_1} \wedge \top) \vee \dots \vee (f_{u_k} \wedge \top) \vee (f_{u_1'} \wedge \bot) \vee \dots \vee (f_{u_l'} \wedge \bot) \vee\\ 
	&&(\overline{f}_{v_1} \wedge \bot) \vee \dots \vee (\overline{f}_{v_m} \wedge \bot) \,.\nonumber
\end{eqnarray}

\vspace{-0.2cm}
%\noindent
We know that the functions $f_{u_1}, \dots, f_{u_k}, f_{u_1'}, \dots, f_{u_l'}, \overline{f}_{v_1}, \dots, \overline{f}_{v_m}$ yield a partition because $f_{|x_1 = \beta_1, \dots, x_{i-1} = \beta_{i-1}}$ and $\overline{f}_{|x_1 = \beta_1, \dots, x_{i-1} = \beta_{i-1}}$ are a partition and $\mathcal{F}$ and $\overline{\mathcal{F}}$ are unambiguous nondeterministic.\\

%\noindent
Finally, we have a look at how to construct an SDD representing $f_u$. Suppose there are already SDDs representing $f_{u_1}, \dots, f_{u_k}, f_{u_1'}, \dots, f_{u_l'}, \overline{f}_{v_1}, \dots, \overline{f}_{v_m}$ and respecting a vtree $T$. Now, we construct an SDD $C$ representing $f_u$ composed like in Equation \ref{equation:vee_one_obdd_to_sdd} from the given SDDs. $C$ respects a new vtree $T'$ which is structured in the following way. The left subtree of $T'$ is $T$. The right subtree of $T'$ is just a leaf labeled by a help variable $h_{x_i, \dots, x_n}$. We need this help variable since $f_{u_1}, \dots, f_{u_k}, f_{u_1'}, \dots, f_{u_l'}, \overline{f}_{v_1}, \dots, \overline{f}_{v_m}$ and $\bot, \top$ formally have to be defined on disjoint variable sets.\\

%\noindent
If the sub-OBDDs $\mathcal{F}_{u_1}, \dots, \mathcal{F}_{u_k}, \mathcal{F}_{u_1'}, \dots, \mathcal{F}_{u_l'}, \overline{\mathcal{F}}_{v_1}, \dots, \overline{\mathcal{F}}_{v_m}$ contain $\vee$-nodes as well, we apply the described idea recursively in order to get the needed SDDs. Observe that all functions that are represented at $\vee$-nodes of the mentioned sub-OBDDs essentially depend on a proper subset of variables $Y' = \{x_j, \dots, x_n\} \subset Y$ for $j > i$ since by assumption there are no edges between $\vee$-nodes. Hence, the termination of the recursion is guaranteed.\\

%\noindent
Next, we will define some notation in order to prove that the described selection of functions always yields a partition. Afterwards, we will give the formal definition of the simulation. We start with the set $\beta(u)$. 

\begin{definition}
	Let $\mathcal{F}$ be an unambiguous nondeterministic OBDD on the variable set $X = \{x_1, \dots, x_n\}$ respecting the variable ordering $\pi = \textnormal{id}$. Furthermore, let $u$ be a node of $\mathcal{F}$ and $Y = \{x_i, \dots, x_n\} \subseteq X$ is chosen with the maximum value of $i \in \{1, \dots, n\}$ fulfilling $\vars{u} \subseteq Y$. Then, $\beta(u)$ is defined as the set of variable assignments over $X \backslash Y$ which can be extended by an assignment of $Y$ such that there exists an accepting path in $\mathcal{F}$ containing $u$.
\end{definition}

%\noindent
The following definition helps us to identify all nodes $u'$ of $\mathcal{F}$ for a fixed $\beta \in \beta(u)$ at which parts of $f_{|x_1 = \beta_1, \dots, x_{i-1} = \beta_{i-1}}$ will be computed. 

\begin{definition}
		Let $\mathcal{F}$ be an unambiguous nondeterministic OBDD on the variable set $X = \{x_1, \dots, x_n\}$ respecting the variable ordering $\pi = \textnormal{id}$. In addition, let $Y = \{x_i, \dots, x_n\} \subseteq X$ and $\beta$ be a variable assignment over $X \backslash Y$. We call a node $u$ of $\mathcal{F}$ with $\vars{u} \subseteq Y$ \emph{maximal w.r.t.} $Y$, if there exists no other node $u'$ in $\mathcal{F}$ such that $\vars{u} \subset \vars{u'} \subseteq Y$ and $\mathcal{F}_u$ is a subgraph of $\mathcal{F}_{u'}$. Moreover, let $R(\mathcal{F}, \beta)$ be the set of all inner nodes $u$ of $\mathcal{F}$ such that $u$ is maximal w.r.t. $Y$ and $\beta$ can be extended by an assignment of $Y$ with the result that there is an accepting path for the extended assignment in $\mathcal{F}$ containing $u$.
\end{definition}

%\noindent
Since we want to resolve $\vee$-nodes of $\mathcal{F}$, we will replace $\vee$-nodes in the following way.

\begin{definition}
	Let $R^{+}(\mathcal{F}, \beta)$ be the set of nodes arising from $R(\mathcal{F}, \beta)$, if every $\vee$-node will be replaced by its children.
\end{definition}

%\noindent
The next lemma will be used in our simulation of unambiguous nondeterministic OBDDs by SDDs in order to get a partition for Boolean functions that are represented at $\vee$-nodes of $\mathcal{F}$.

\begin{lemma}
	\label{lemma:partitionlemma_vee_one_obdd_to_sdd}
	Let $\mathcal{F}$ and $\overline{\mathcal{F}}$ be unambiguous nondeterministic \textnormal{OBDDs}  respecting the variable ordering $\pi = \textnormal{id}$ and representing the Boolean functions $\Phi_{\mathcal{F}}$ and $\Phi_{\overline{\mathcal{F}}}$ such that $\Phi_{\mathcal{F}} = \overline{\Phi_{\overline{\mathcal{F}}}}$. Let $u$ be an $\vee$-node of $\mathcal{F}$ and $\beta \in \beta(u)$. Furthermore, the sets $R^{+}(\mathcal{F}, \beta) = \{u_1, \dots, u_k\}$ and $R^{+}(\overline{\mathcal{F}}, \beta) = \{v_1, \dots, v_l\}$ are given. Let $\Phi_{u_i}$ and $\Phi_{v_j}$ with $i \in [k]$ and $j \in [l]$ be the functions that are represented at the nodes $u_i$ of $\mathcal{F}$ and $v_j$ of $\overline{\mathcal{F}}$, respectively. Then, the set of functions $\Phi = \{\Phi_{u_1}, \dots, \Phi_{u_k}, \Phi_{v_1}, \dots, \Phi_{v_l}\}$ is a partition.
\end{lemma}

\begin{proofidea}
	First, we have to show that the set of functions $\Phi$ contains at least two elements. Otherwise, $\Phi$ cannot yield a partition. For this purpose, it can be shown that the children of the $\vee$-node $u$ are elements of $R^{+}(\mathcal{F}, \beta)$. Next, we have to prove that $\Phi$ fulfills all partition properties. One can show that the violation of at least one property will lead to a contradiction. The entire proof can be found in Appendix A.
%\ref{appendix:proof_of_partition_lemma}. 
\end{proofidea}

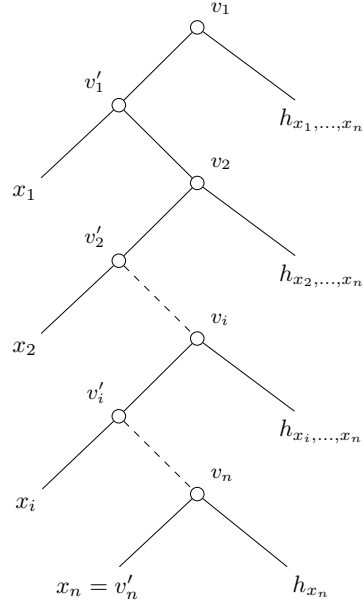
\begin{figure*}[!t]
	\centering{
		\resizebox{!}{0.35\textheight}{
			\begin{tikzpicture}[>= stealth']

\tikzstyle{vertex} = [draw, circle, inner sep = 2pt]

% Knoten
\node[vertex] (v1) at (0,0) [label={[black, font = \small] above right:{$v_1$}}] {};

\node[vertex] (v1s) [below left = of v1, label={[black, font = \small] above left:{$v_1'$}}] {};
\node (h1n) [below right = of v1] {$h_{x_{1}, \dots, x_{n}}$};

\node (x1) [below left = of v1s] {$x_1$};
\node[vertex] (v2) [below right = of v1s, label={[black, font = \small] above right:{$v_2$}}] {};

\node[vertex] (v2s) [below left = of v2, label={[black, font = \small] above left:{$v_2'$}}] {};
\node (h2n) [below right = of v2] {$h_{x_{2}, \dots, x_{n}}$};

\node (x2) [below left = of v2s] {$x_2$};
\node[vertex] (vi) [below right = of v2s, label={[black, font = \small] above right:{$v_i$}}] {};

\node[vertex] (vis) [below left = of vi, label={[black, font = \small] above left:{$v_i'$}}] {};
\node (hin) [below right = of vi] {$h_{x_{i}, \dots, x_{n}}$};

\node (xi) [below left = of vis] {$x_i$};
\node[vertex] (vn1) [below right = of vis, label={[black, font = \small] above right:{$v_{n}$}}] {};

\node (vn1s) [below left = 1 and 0.65 of vn1] {$x_n = v_n'$};
\node (hn1n) [below right = 1 and 1.2 of vn1] {$h_{x_{n}}$};

% Kanten
\draw (v1) to (v1s);
\draw (v1) to (h1n);

\draw (v1s) to (x1);
\draw (v1s) to (v2);

\draw (v2) to (v2s);
\draw (v2) to (h2n);

\draw (v2s) to (x2);
\draw[dashed] (v2s) to (vi);

\draw (vi) to (vis);
\draw (vi) to (hin);

\draw (vis) to (xi);
\draw[dashed] (vis) to (vn1);

\draw (vn1) to (vn1s);
\draw (vn1) to (hn1n);

\end{tikzpicture}
		}
	}
	\caption{The vtree $T$ for the set of variables $X \cup H$.}
	\label{figure:vtree_vee_one_obdds_to_sdds}
\end{figure*}

%\noindent
Now, we give the formal definition of the simulation.

\begin{simulation}
	\label{simulation:vtree_vee_one_obdds_to_sdds}
	Let $f \in B_n$ be a Boolean function such that $f$ and $\overline{f}$ can be represented by unambiguous nondeterministic \textnormal{OBDDs} respecting the variable ordering $\pi = \textnormal{id}$. Let $\mathcal{F}$ and $\overline{\mathcal{F}}$ be those $\vee_1$-\textnormal{OBDDs}. We construct an \textnormal{SDD} $C$ representing $f$ from the $\vee_1$-\textnormal{OBDDs} $\mathcal{F}$ and $\overline{\mathcal{F}}$ in the following.
	
	First, in order to define the vtree $T$ that will be respected by $C$ we augment $X = \{x_1, \dots, x_n\}$ by help variables $H = \{h_{x_1, \dots, x_n}, h_{x_2, \dots, x_n}, \dots, h_{x_n}\}$. We define the vtree $T$ for the set of variables $X \cup H$ as depicted in Figure \ref{figure:vtree_vee_one_obdds_to_sdds}:
	
	\begin{itemize}
		\item $T$ consists of the inner nodes $v_1, \dots, v_{n}, v_1', \dots, v_{n}'$ and leaves for the variables of $X \cup H$.
		\item The node $v_1$ is the root of $T$.
		\item For all $i \in \{1, \dots, n\}:$ $(v_i, v_i')$ and $(v_i, h_{x_{i}, \dots, x_{n}})$ are edges of $T$.
		\item For all $i \in \{1, \dots, n-1\}:$ $(v_i', x_i)$ and $(v_i', v_{i+1})$ are edges of $T$.
		\item The node $v_n'$ is equal to the leaf labeled by $x_n$.
	\end{itemize}
	
	Let $(V,E)$ and $(\overline{V},\overline{E})$ be the sets of nodes and edges of the $\vee_1$-\textnormal{OBDDs} $\mathcal{F}$ and $\overline{\mathcal{F}}$, respectively. Furthermore, let $X' \subseteq X$ be the set of variables for which there is decision node of $\mathcal{F}$ or $\overline{\mathcal{F}}$ labeled by a variable of $X'$. Moreover, we have $Y = V \cup \overline{V}$ and $Z = \{\wedge_0, \wedge_1, \emptyset\} \cup X' \cup Y$. The nodes of $C$ are tuples $(u,v) \in Y \times Z$. We construct $C$ respecting $T$ by mapping nodes and edges of $\mathcal{F}$ and $\overline{\mathcal{F}}$ to nodes and edges of $C$ according to the following cases:
	
	\vspace{-0.09cm}
	\begin{itemize}
		\item[(a)] For each decision node $u \in (V \cup \overline{V})$ for a variable $x_i \in X$ which is only connected to sinks, add a decision node $(u,\emptyset)$ to $C$ that is labeled by a literal $\overline{x_i}$ or $x_i$ according to the semantics of $u$.
		\item[(b)] For each decision node $u \in (V \cup \overline{V})$ for a variable $x_i \in X$ which is not only connected to sinks, add the $\vee$-node $(u, \emptyset)$, both $\wedge$-nodes $(u, \wedge_0)$, $(u, \wedge_1)$ and the decision nodes $(u,\overline{x_i})$, $(u,x_i)$ that are labeled by $\overline{x_i}$ and $x_i$, respectively. In addition, add the following edges to $C$:
		\begin{itemize}
			\item $((u, \emptyset), (u, \wedge_0))$ and $((u, \emptyset), (u, \wedge_1))$,
			\item $((u, \wedge_0), (u, \overline{x_i}))$ and $((u, \wedge_1), (u, x_i))$,
			\item the $0$-edge $(u, u_0) \in (E \cup \overline{E})$ is mapped to edge $((u, \wedge_0), (u_0, \emptyset))$,
			\item the $1$-edge $(u, u_1) \in (E \cup \overline{E})$ is mapped to edge $((u, \wedge_1), (u_1, \emptyset))$.
		\end{itemize}
		The case of $u \in V$ is depicted in Figure \ref{figure:vee_one_obdds_to_sdds_case_b}.
		
		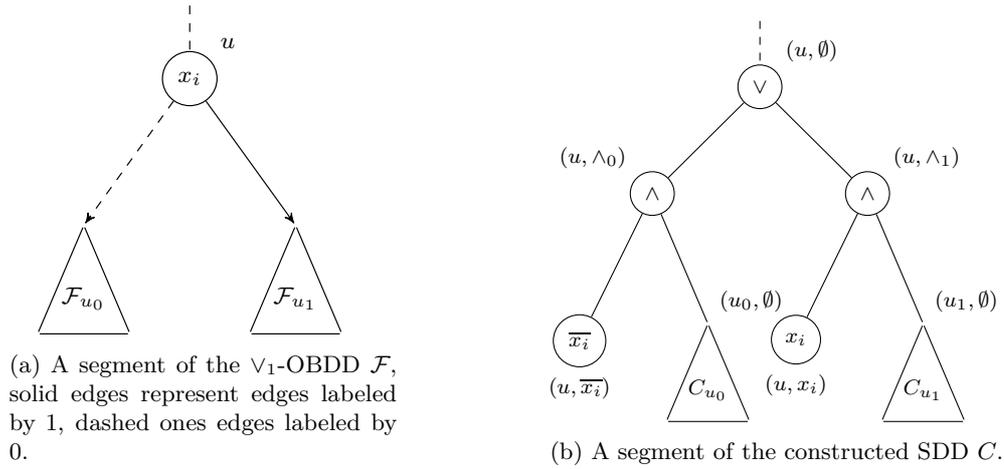
\begin{figure}[!t]
			\centering
			\begin{subfigure}[b]{0.34\textwidth}
				\hspace*{1em}\resizebox{!}{0.2\textheight}{\begin{tikzpicture}[>= stealth', ]

\tikzstyle{vertex} = [draw, circle, font = \footnotesize]
\tikzstyle{sink} = [draw, rectangle]
\tikzstyle{empty} = [inner sep = 0pt]

% -1. Ebene
\node (a) at (0,1) {};

% 0. Ebene
\node[vertex] (x_i) at (0,0) [label={[font = \footnotesize, label distance = 0.01cm] above right:{$u$}}] {$x_i$};

% 1. Ebene
\node[empty] (v2_top) [below left = 1.5 and 1 of x_i, label={[font = \footnotesize, label distance = 0.6cm] below:{$\mathcal{F}_{u_0}$}}] {};
\node[empty] (v2_left) [below left = 1.25 and 0.5 of v2_top] {};
\node[empty] (v2_right) [below right = 1.25 and 0.5 of v2_top] {};

\node[empty] (v3_top) [below right = 1.5 and 1 of x_i, label={[font = \footnotesize, label distance = 0.6cm] below:{$\mathcal{F}_{u_1}$}}] {};
\node[empty] (v3_left) [below left = 1.25 and 0.5 of v3_top] {};
\node[empty] (v3_right) [below right = 1.25 and 0.5 of v3_top] {};

% Kanten
\draw[dashed] (a) to (x_i);

\draw[->, dashed] (x_i) to (v2_top);
\draw[->] (x_i) to (v3_top);

\draw (v2_left) to (v2_right);
\draw (v2_top) to (v2_left);
\draw (v2_top) to (v2_right);

\draw (v3_left) to (v3_right);
\draw (v3_top) to (v3_left);
\draw (v3_top) to (v3_right);

\end{tikzpicture}}
				\caption{A segment of the $\vee_1$-OBDD $\mathcal{F}$, solid edges represent edges labeled by $1$, dashed ones edges labeled by $0$.}
				\label{figure:vee_one_obdds_to_sdds_case_b_links}
			\end{subfigure}	
			\begin{subfigure}[b]{0.65\textwidth}
				\hspace*{5em}\resizebox{!}{0.24\textheight}{\begin{tikzpicture}[>= stealth', ]

\tikzstyle{vertex} = [draw, circle, font = \footnotesize]
\tikzstyle{sink} = [draw, rectangle]
\tikzstyle{empty} = [inner sep = 0pt]

% -1. Ebene
\node (a) at (0,1) {};

% 0. Ebene
\node[vertex] (vee) at (0,0) [label={[font = \footnotesize, label distance = 0.01cm] above right:{$(u, \emptyset)$}}] {$\vee$};

% 1. Ebene
\node[vertex] (wedge1)  [below left = of vee, label={[font = \footnotesize, label distance = 0.01cm] above left:{$(u, \wedge_0)$}}] {$\wedge$};
\node[vertex] (wedge2) [below right = of vee, label={[font = \footnotesize, label distance = 0.01cm] above right:{$(u, \wedge_1)$}}] {$\wedge$};

% 2. Ebene
\node[empty] (v2_top) [below right = 1.5 and 0.5 of wedge1, label={[font = \footnotesize, label distance = 0.6cm] below:{$C_{u_0}$}}, label={[font = \footnotesize, label distance = 0.01cm] above right:{$(u_0, \emptyset)$}}] {};
\node[empty] (v2_left) [below left = 1.25 and 0.5 of v2_top] {};
\node[empty] (v2_right) [below right = 1.25 and 0.5 of v2_top] {};

\node[empty] (v3_top) [below right = 1.5 and 0.5 of wedge2, label={[font = \footnotesize, label distance = 0.6cm] below:{$C_{u_1}$}}, label={[font = \footnotesize, label distance = 0.01cm] above right:{$(u_1, \emptyset)$}}] {};
\node[empty] (v3_left) [below left = 1.25 and 0.5 of v3_top] {};
\node[empty] (v3_right) [below right = 1.25 and 0.5 of v3_top] {};

\node[vertex] (nx_i) [below left = 1.5 and 0.5 of wedge1, label={[font = \footnotesize, label distance = 0.01cm] below:{$(u, \overline{x_i})$}}] {$\overline{x_i}$};

\node[vertex] (x_i) [below left = 1.5 and 0.5 of wedge2, label={[font = \footnotesize, label distance = 0.01cm] below:{$(u, x_i)$}}] {$x_i$};

% Kanten
\draw[dashed] (a) to (vee);

\draw (vee) to (wedge1);
\draw (vee) to (wedge2);

\draw (wedge1) to (v2_top);
\draw (wedge1) to (nx_i);

\draw (wedge2) to (v3_top);
\draw (wedge2) to (x_i);

\draw (v2_left) to (v2_right);
\draw (v2_top) to (v2_left);
\draw (v2_top) to (v2_right);

\draw (v3_left) to (v3_right);
\draw (v3_top) to (v3_left);
\draw (v3_top) to (v3_right);

\end{tikzpicture}}
				\caption{A segment of the constructed SDD $C$.}
				\label{figure:vee_one_obdds_to_sdds_case_b_rechts}
			\end{subfigure}
			\caption[Skizze zur Konstruktion eines SDD $C$ aus $\vee_1$-OBDDs (a)]{Case (\textit{b}) in Simulation
                         \ref{simulation:vtree_vee_one_obdds_to_sdds}.}
			\label{figure:vee_one_obdds_to_sdds_case_b}
		\end{figure}
		
		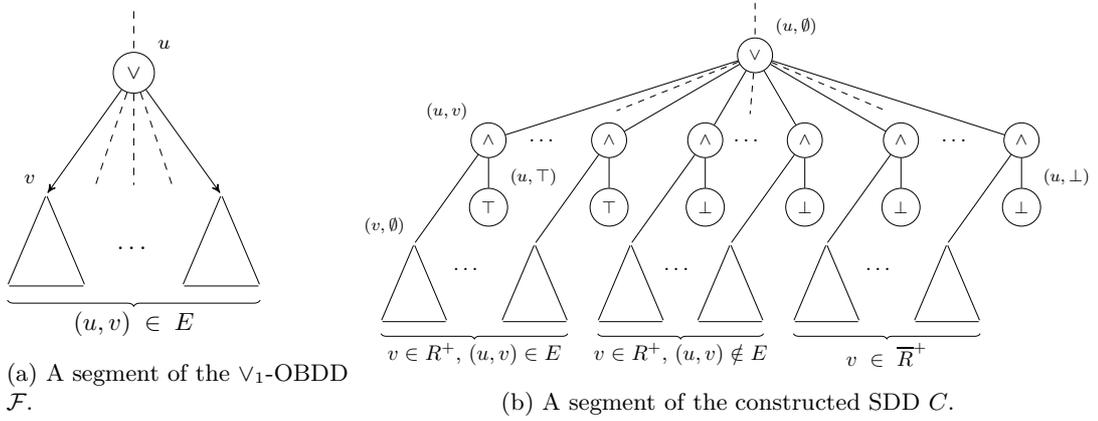
\begin{figure}[!t]
			\centering
			\begin{subfigure}[b]{0.3\textwidth}
				\resizebox{!}{0.2\textheight}{\begin{tikzpicture}[>= stealth', ]

\tikzstyle{vertex} = [draw, circle, font = \footnotesize]
\tikzstyle{sink} = [draw, rectangle]
\tikzstyle{empty} = [inner sep = 0pt]

% -1. Ebene
\node (a) at (0,1) {};

% 0. Ebene
\node[vertex] (vee) at (0,0) [label={[font = \footnotesize, label distance = 0.01cm] above right:{$u$}}] {$\vee$};

% 1. Ebene
\node[empty] (v2_top) [below left = 1.5 and 1 of vee, label={[font = \footnotesize, label distance = 0.01cm] above left:{$v$}}] {};
\node[empty] (v2_left) [below left = 1.25 and 0.5 of v2_top] {};
\node[empty] (v2_right) [below right = 1.25 and 0.5 of v2_top] {};

\node[empty] (v3_top) [below right = 1.5 and 1 of vee] {};
\node[empty] (v3_left) [below left = 1.25 and 0.5 of v3_top] {};
\node[empty] (v3_right) [below right = 1.25 and 0.5 of v3_top] {};

\node (dots1) [right = 0.5 of v2_top] {};
\node (dots2) [below = 1.3 of vee] {};
\node (dots3) [left = 0.5 of v3_top] {};

\node (dots4) [below = 0.5 of dots2] {$\dots$};

% Kanten
\draw[dashed] (a) to (vee);

\draw[->] (vee) to (v2_top);
\draw[->] (vee) to (v3_top);

\draw[dashed] (vee) to (dots1);
\draw[dashed] (vee) to (dots2);
\draw[dashed] (vee) to (dots3);

\draw (v2_left) to (v2_right);
\draw (v2_top) to (v2_left);
\draw (v2_top) to (v2_right);

\draw (v3_left) to (v3_right);
\draw (v3_top) to (v3_left);
\draw (v3_top) to (v3_right);

\draw[decoration={brace,mirror,raise=5pt},decorate]
(v2_left.south) -- node[below=6pt, text width = 3cm, align = center] {$(u,v) \in E$} (v3_right.south);

\end{tikzpicture}}
				\caption{A segment of the $\vee_1$-OBDD $\mathcal{F}$.}
				\label{figure:vee_one_obdds_to_sdds_case_c_links}
			\end{subfigure}	
			\begin{subfigure}[b]{0.65\textwidth}
				\resizebox{!}{0.22\textheight}{\begin{tikzpicture}[>= stealth', ]

\tikzstyle{vertex} = [draw, circle, font = \footnotesize]
\tikzstyle{sink} = [draw, rectangle]
\tikzstyle{empty} = [inner sep = 0pt]

\tikzset{draw half paths/.style 2 args={%
		decoration={show path construction,
			lineto code={
				\draw [#1] (\tikzinputsegmentfirst) -- 
				($(\tikzinputsegmentfirst)!0.7!(\tikzinputsegmentlast)$);
				\draw [#2] ($(\tikzinputsegmentfirst)!0.7!(\tikzinputsegmentlast)$)
				-- (\tikzinputsegmentlast);
			}
		}, decorate
	}}

% -1. Ebene
\node (a) at (0,1) {};

% 0. Ebene
\node[vertex] (vee) at (0,0) [label={[font = \footnotesize, label distance = 0.01cm] above right:{$(u, \emptyset)$}}] {$\vee$};

% 1. Ebene
\node[vertex] (wedge1) [below left = 1 and 4 of vee, label={[font = \footnotesize, label distance = 0.01cm] above left:{$(u, v)$}}] {$\wedge$};
\node[vertex] (wedge2) [below left = 1 and 2 of vee] {$\wedge$};
\node[vertex] (wedge5) [below right = 1 and 2 of vee] {$\wedge$};
\node[vertex] (wedge6) [below right = 1 and 4 of vee] {$\wedge$};

\node[vertex] (wedge3) [right = 1 of wedge2] {$\wedge$};
\node[vertex] (wedge4) [left = 1  of wedge5] {$\wedge$};

% 2. Ebene
\node[empty] (v1_top) [below left = 1.5 and 1 of wedge1, , label={[font = \footnotesize, label distance = 0.01cm] above left:{$(v, \emptyset)$}}] {};
\node[empty] (v1_left) [below left = 1.25 and 0.5 of v1_top] {};
\node[empty] (v1_right) [below right = 1.25 and 0.5 of v1_top] {};

\node[empty] (v2_top) [below left = 1.5 and 1 of wedge2] {};
\node[empty] (v2_left) [below left = 1.25 and 0.5 of v2_top] {};
\node[empty] (v2_right) [below right = 1.25 and 0.5 of v2_top] {};

\node[empty] (v3_top) [below left = 1.5 and 1 of wedge3] {};
\node[empty] (v3_left) [below left = 1.25 and 0.5 of v3_top] {};
\node[empty] (v3_right) [below right = 1.25 and 0.5 of v3_top] {};

\node[empty] (v4_top) [below left = 1.5 and 1 of wedge4] {};
\node[empty] (v4_left) [below left = 1.25 and 0.5 of v4_top] {};
\node[empty] (v4_right) [below right = 1.25 and 0.5 of v4_top] {};

\node[empty] (v5_top) [below left = 1.5 and 1 of wedge5] {};
\node[empty] (v5_left) [below left = 1.25 and 0.5 of v5_top] {};
\node[empty] (v5_right) [below right = 1.25 and 0.5 of v5_top] {};

\node[empty] (v6_top) [below left = 1.5 and 1 of wedge6] {};
\node[empty] (v6_left) [below left = 1.25 and 0.5 of v6_top] {};
\node[empty] (v6_right) [below right = 1.25 and 0.5 of v6_top] {};

\node[vertex] (top1) [below = 0.5 of wedge1, label={[font = \footnotesize, label distance = 0.01cm] above right:{$(u, \top)$}}] {$\top$};
\node[vertex] (top2) [below = 0.5 of wedge2] {$\top$};
\node[vertex] (bot1) [below = 0.5 of wedge3] {$\bot$};
\node[vertex] (bot2) [below = 0.5 of wedge4] {$\bot$};
\node[vertex] (bot3) [below = 0.5 of wedge5] {$\bot$};
\node[vertex] (bot4) [below = 0.5 of wedge6, label={[font = \footnotesize, label distance = 0.01cm] above right:{$(u, \bot)$}}] {$\bot$};

\node (dots1) [below right = 0.25 and 0.5 of v1_top] {$\dots$};
\node (dots2) [below right = 0.25 and 0.4 of v3_top] {$\dots$};
\node (dots3) [below right = 0.25 and 0.5 of v5_top] {$\dots$};

\node (dots4) [right = 0.25 of wedge1] {$\dots$};
\node (dots5) [right = 0.05 of wedge3] {$\dots$};
\node (dots6) [right = 0.25 of wedge5] {$\dots$};

% Kanten
\draw[dashed] (a) to (vee);

\draw (vee) to (wedge1);
\draw (vee) to (wedge2);
\draw (vee) to (wedge3);
\draw (vee) to (wedge4);
\draw (vee) to (wedge5);
\draw (vee) to (wedge6);

\draw (wedge1) to (v1_top);
\draw (wedge2) to (v2_top);
\draw (wedge3) to (v3_top);
\draw (wedge4) to (v4_top);
\draw (wedge5) to (v5_top);
\draw (wedge6) to (v6_top);

\draw (wedge1) to (top1);
\draw (wedge2) to (top2);
\draw (wedge3) to (bot1);
\draw (wedge4) to (bot2);
\draw (wedge5) to (bot3);
\draw (wedge6) to (bot4);

\draw (v1_left) to (v1_right);
\draw (v1_top) to (v1_left);
\draw (v1_top) to (v1_right);

\draw (v2_left) to (v2_right);
\draw (v2_top) to (v2_left);
\draw (v2_top) to (v2_right);

\draw (v3_left) to (v3_right);
\draw (v3_top) to (v3_left);
\draw (v3_top) to (v3_right);

\draw (v4_left) to (v4_right);
\draw (v4_top) to (v4_left);
\draw (v4_top) to (v4_right);

\draw (v5_left) to (v5_right);
\draw (v5_top) to (v5_left);
\draw (v5_top) to (v5_right);

\draw (v6_left) to (v6_right);
\draw (v6_top) to (v6_left);
\draw (v6_top) to (v6_right);

\path [draw half paths={dashed}{draw=none}] (vee) -- (dots4);
\path [draw half paths={dashed}{draw=none}] (vee) -- (dots5);
\path [draw half paths={dashed}{draw=none}] (vee) -- (dots6);

\draw[decoration={brace,mirror,raise=5pt},decorate]
(v1_left.south) -- node[below=6pt, text width = 3cm, align = center] {$v \in R^+$, $(u,v) \in E$} (v2_right.south);

\draw[decoration={brace,mirror,raise=5pt},decorate]
(v3_left.south) -- node[below=6pt, text width = 3cm, align = center] {$v \in R^+$, $(u,v) \notin E$} (v4_right.south);

\draw[decoration={brace,mirror,raise=5pt},decorate]
(v5_left.south) -- node[below=6pt, text width = 3cm, align = center] {$v \in \overline{R}^+$} (v6_right.south);

\end{tikzpicture}}
				\caption{A segment of the constructed SDD $C$.}
				\label{figure:vee_one_obdds_to_sdds_case_c_rechts}
			\end{subfigure}
			\caption[Skizze zur Konstruktion eines SDD $C$ aus $\vee_1$-OBDDs (b)]{Case (\textit{c})
				in Simulation \ref{simulation:vtree_vee_one_obdds_to_sdds}.}
			\label{figure:vee_one_obdds_to_sdds_case_c}
		\end{figure}
		
		\item[(c)] For each $\vee$-node $u \in (V \cup \overline{V})$, add an $\vee$-node $(u,\emptyset)$ to $C$. Let $\beta \in \beta(u)$ be a (partial) variable assignment (uniquely chosen). If $u \in V$ holds, let $R^+ = R^+(\mathcal{F}, \beta)$ and $\overline{R}^+ = R^+(\overline{\mathcal{F}}, \beta)$. Otherwise, let $\overline{R}^+ = R^+(\mathcal{F}, \beta)$ and $R^+ = R^+(\overline{\mathcal{F}}, \beta)$. For each node $v \in (R^+ \cup \overline{R}^+)$, add an $\wedge$-node $(u,v)$ to $C$. Moreover, add the nodes $(u,\bot)$ and $(u,\top)$ to $C$ which are labeled by the constants $\bot$ and $\top$, respectively. For each $u \in V$, add the following edges to $C$:
		\begin{itemize}
			\item For each node $v \in R^+$ fulfilling $(u,v) \in E$ insert
			\begin{itemize}
				\item $((u, \emptyset), (u, v))$
				\item $((u, v), (v, \emptyset))$
				\item $((u, v), (u, \top))$
			\end{itemize}
			\item For each node $v \in R^+$ fulfilling $(u,v) \notin E$ and each $v \in \overline{R}^+$ insert
			\begin{itemize}
				\item $((u, \emptyset), (u, v))$
				\item $((u, v), (v, \emptyset))$
				\item $((u, v), (u, \bot))$
			\end{itemize}
		\end{itemize}
		The case of $u \in V$ is depicted in Figure \ref{figure:vee_one_obdds_to_sdds_case_c}. If $u \in \overline{V}$ holds, then the edges will be inserted analogously by replacing the set of edges $E$ by $\overline{E}$ in the given description.
	\end{itemize}
	
	%\noindent
	Furthermore, for each sink $u \in (V \cup \overline{V})$ we add a node $(u, \emptyset)$ labeled by the respective constant to $C$. The root of $C$ is given by $(\textnormal{root}(\mathcal{F}), \emptyset)$. Finally, we remove all nodes and edges from the resulting \textnormal{SDD} $C$ which cannot be reached from $\textnormal{root}(C) = (\textnormal{root}(\mathcal{F}), \emptyset)$.
	
\end{simulation}

%\noindent
In Figures \ref{figure:example_vee_one_obdd_to_sdd} and \ref{figure:example_vee_one_obdd_to_sdd_finished}, we give an example for the proposed simulation of unambiguous nondeterministic OBDDs by SDDs. Figure \ref{figure:example_vee_one_obdd_to_sdd} depicts two unambiguous nondeterministic OBDDs $\mathcal{F}$ and $\overline{\mathcal{F}}$ representing Boolean functions $f$ and $\overline{f}$, respectively. Whereas Figure \ref{figure:example_vee_one_obdd_to_sdd_finished} shows the SDD $C$ constructed by the simulation. 

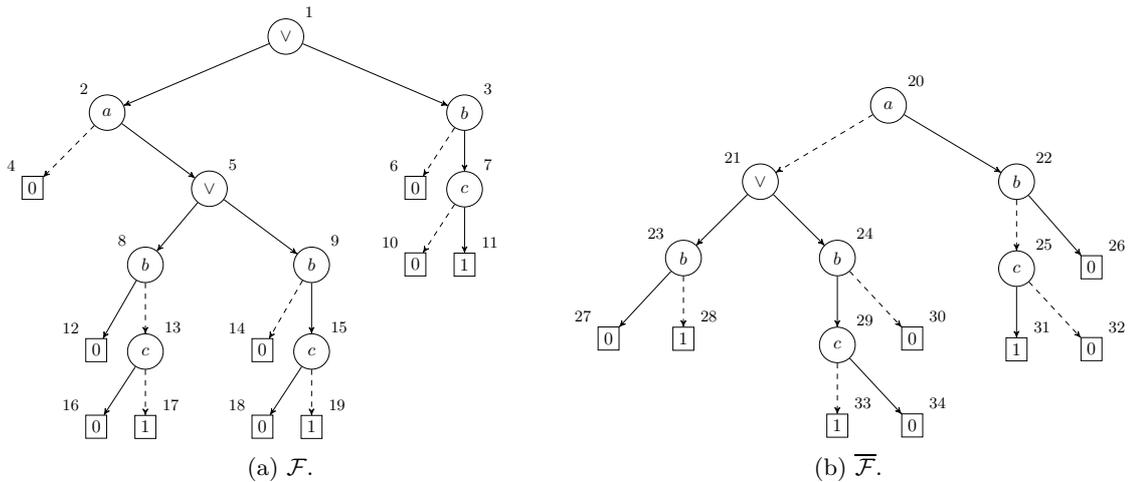
\begin{figure}[!b]
	\centering
	\begin{subfigure}[b]{0.49\textwidth}
		\resizebox{!}{0.25\textheight}{\begin{tikzpicture}[>= stealth']

\tikzstyle{vertex} = [draw, circle, minimum size=0.7cm]
\tikzstyle{sink} = [draw, rectangle]

% Knoten
% 1. Ebene
\node[vertex] (v1) at (0,0) [label={[black, font = \small] above right:{$1$}}] {$\vee$};

% 2. Ebene
\node[vertex] (v1m) [below left = 1 and 3 of v1, label={[black, font = \small] above left:{$2$}}] {$a$};
\node[vertex] (v1p) [below right = 1 and 3 of v1, label={[black, font = \small] above right:{$3$}}] {$b$};

% 3. Ebene
\node[sink] (vbm) [below left = 1 and 1 of v1m, label={[black, font = \small] above left:{$4$}}] {$0$};
\node[vertex] (v3) [below right = 1 and 1.5 of v1m, label={[black, font = \small] above right:{$5$}}] {$\vee$};
\node[vertex] (v2) [below = 0.8 of v1p, label={[black, font = \small] above right:{$7$}}] {$c$};
\node[sink] (vmx) [below left = 1 and 0.5 of v1p, label={[black, font = \small] above left:{$6$}}] {$0$};

% 4. Ebene
\node[vertex] (v3m) [below left = 1 and 0.75 of v3, label={[black, font = \small] above left:{$8$}}] {$b$};
\node[vertex] (v3p) [below right = 1 and 1.5 of v3, label={[black, font = \small] above right:{$9$}}] {$b$};
\node[sink] (vcm) [below = 0.9 of v2, label={[black, font = \small] above right:{$11$}}] {$1$};

\node[sink] (vm5) [below left = 1 and 0.5 of v2, label={[black, font = \small] above left:{$10$}}] {$0$};

% 5. Ebene
\node[vertex] (vcp) [below = 1 of v3p, label={[black, font = \small] above right:{$15$}}] {$c$};
\node[vertex] (vxp) [below = 1 of v3m, label={[black, font = \small] above right:{$13$}}] {$c$};

\node[sink] (vm1) [below left = 1.2 and 0.5 of v3m, label={[black, font = \small] above left:{$12$}}] {$0$};
\node[sink] (vm2) [below left = 1.2 and 0.5 of v3p, label={[black, font = \small] above left:{$14$}}] {$0$};

% 6. Ebene
\node[sink] (vm) [below left = 1 and 0.5 of vcp, label={[black, font = \small] above left:{$18$}}] {$0$};
\node[sink] (vp) [below = 0.9 of vcp, label={[black, font = \small] above right:{$19$}}] {$1$};

\node[sink] (vmm) [below left = 1 and 0.5 of vxp, label={[black, font = \small] above left:{$16$}}] {$0$};
\node[sink] (vpp) [below = 0.9 of vxp, label={[black, font = \small] above right:{$17$}}] {$1$};

% Kanten
% 1. -> 2.
\draw[->] (v1) to (v1m);
\draw[->] (v1) to (v1p);

% 2. -> 3.
\draw[->, dashed] (v1m) to (vbm);
\draw[->] (v1m) to (v3);
\draw[->] (v1p) to (v2);
\draw[->, dashed] (v1p) to (vmx);

% 3. -> 4.
\draw[->] (v3) to (v3m);
\draw[->] (v3) to (v3p);
\draw[->] (v2) to (vcm);
\draw[->, dashed] (v3m) to (vxp);
\draw[->, dashed] (v2) to (vm5);

% 4. -> 5.
\draw[->] (v3p) to (vcp);
\draw[->] (v3m) to (vm1);
\draw[->, dashed] (v3p) to (vm2);

% 5. -> 6.
\draw[->] (vcp) to (vm);
\draw[->, dashed] (vcp) to (vp);

\draw[->] (vxp) to (vmm);
\draw[->, dashed] (vxp) to (vpp);

\end{tikzpicture}}
		\caption{$\mathcal{F}$.}
		\label{figure:example_vee_one_obdd_to_sdd_left}
	\end{subfigure}	
	\begin{subfigure}[b]{0.49\textwidth}
		\resizebox{!}{0.21\textheight}{\begin{tikzpicture}[>= stealth']

\tikzstyle{vertex} = [draw, circle, minimum size=0.7cm]
\tikzstyle{sink} = [draw, rectangle]

% Knoten
% 1. Ebene
\node[vertex] (a) at (0,0) [label={[black, font = \small] above right:{$20$}}] {$a$};

% 2. Ebene
\node[vertex] (v1) [below left = 1 and 2 of a, label={[black, font = \small] above left:{$21$}}] {$\vee$};
\node[vertex] (b1) [below right = 1 and 2 of a, label={[black, font = \small] above right:{$22$}}] {$b$};

% 3. Ebene
\node[vertex] (b2) [below left = of v1, label={[black, font = \small] above left:{$23$}}] {$b$};
\node[vertex] (b3) [below right = of v1, label={[black, font = \small] above right:{$24$}}] {$b$};
\node[vertex] (c1) [below = of b1, label={[black, font = \small] above right:{$25$}}] {$c$};
\node[sink] (z0) [below right = 1.2 and 1 of b1, label={[black, font = \small] above right:{$26$}}] {$0$};

% 4. Ebene
\node[sink] (z1) [below left = 1.1 and 1 of b2, label={[black, font = \small] above left:{$27$}}] {$0$};
\node[sink] (e0) [below = of b2, label={[black, font = \small] above right:{$28$}}] {$1$};
\node[vertex] (c2) [below = of b3, label={[black, font = \small] above right:{$29$}}] {$c$};
\node[sink] (z2) [below right = 1.1 and 1 of b3, label={[black, font = \small] above right:{$30$}}] {$0$};
\node[sink] (e1) [below = of c1, label={[black, font = \small] above right:{$31$}}] {$1$};
\node[sink] (z4) [below right = 1.1 and 1 of c1, label={[black, font = \small] above right:{$32$}}] {$0$};

% 5. Ebene
\node[sink] (e2) [below = of c2, label={[black, font = \small] above right:{$33$}}] {$1$};
\node[sink] (z5) [below right = 1.1 and 1 of c2, label={[black, font = \small] above right:{$34$}}] {$0$};

% Kanten
\draw[->, dashed] (a) to (v1);
\draw[->] (a) to (b1);

\draw[->] (v1) to (b2);
\draw[->] (v1) to (b3);

\draw[->, dashed] (b1) to (c1);
\draw[->] (b1) to (z0);

\draw[->, dashed] (b2) to (e0);
\draw[->] (b2) to (z1);

\draw[->] (b3) to (c2);
\draw[->, dashed] (b3) to (z2);

\draw[->] (c2) to (z5);
\draw[->, dashed] (c2) to (e2);

\draw[->] (c1) to (e1);
\draw[->, dashed] (c1) to (z4);

\end{tikzpicture}}
		\caption{$\overline{\mathcal{F}}$.}
		\label{figure:example_vee_one_obdd_to_sdd_right}
	\end{subfigure}
	\caption{Unambiguous nondeterministic OBDDs $\mathcal{F}$ and $\overline{\mathcal{F}}$ representing the Boolean functions $f(a,b,c) = (b \wedge c) \vee (a \wedge ((\overline{b} \wedge \overline{c}) \vee (b \wedge \overline{c})))$ and $\overline{f}$. Solid edges represent edges labeled by $1$, dashed ones edges labeled by $0$.}
	\label{figure:example_vee_one_obdd_to_sdd}
\end{figure}

\begin{figure}[!ht]
	\centering{
		\hspace*{-1cm}
		\resizebox{!}{0.38\textheight}{
			\input{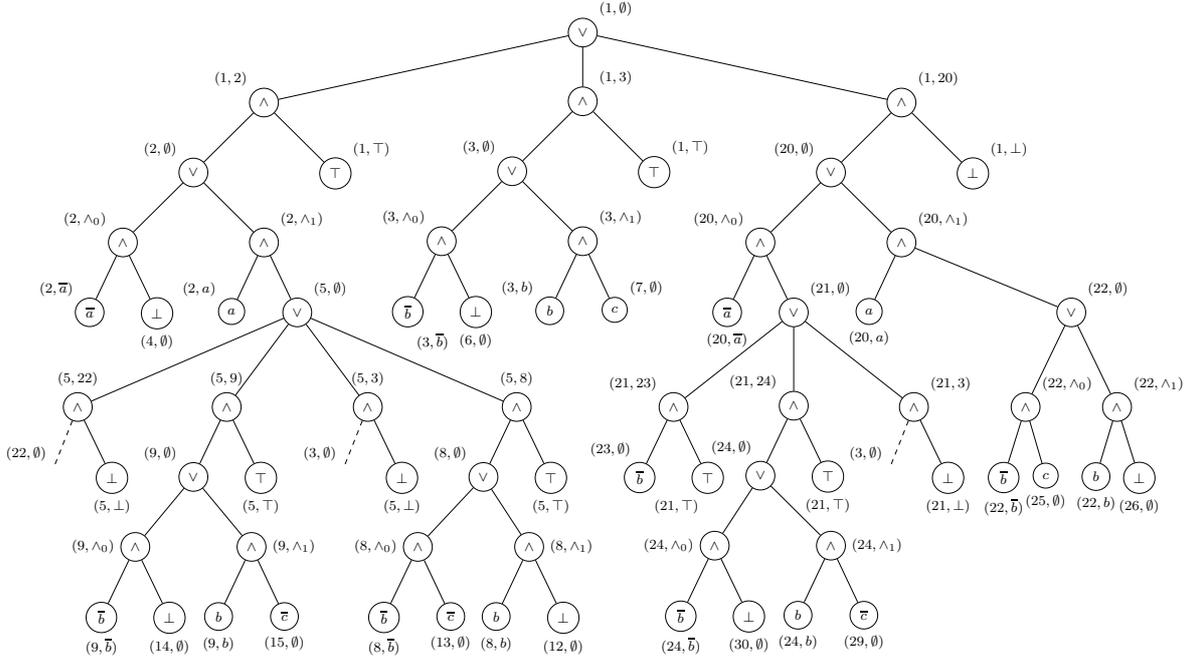}
		}
	}
	\caption{The SDD $C$ which also represents $f$ constructed by Simulation \ref{simulation:vtree_vee_one_obdds_to_sdds} with input $\mathcal{F}$ and $\overline{\mathcal{F}}$. The dashed lines depict connections to sub-SDDs that are already shown in the diagram.}
	\label{figure:example_vee_one_obdd_to_sdd_finished}
\end{figure}

\subsection{\bf Size, correctness, and equivalence}\label{subsection:simulation_vee_one_obdds_by_sdds_size_correctness}

We get a relationship between the sizes of the given unambiguous nondeterministic OBDDs and the constructed SDD by the following lemma which states that the increase in size is at most quadratic in $|\mathcal{F}| + |\overline{\mathcal{F}}|$.

\begin{lemma}
	\label{lemma:size_vee_one_obdd_to_sdd}
	Let $\mathcal{F}$ and $\overline{\mathcal{F}}$ be unambiguous nondeterministic \textnormal{OBDDs}  respecting the variable ordering $\pi = \textnormal{id}$ and representing Boolean functions $f, \overline{f} \in B_n$. Additionally, let $|\mathcal{F}| = N_1$, $|\overline{\mathcal{F}}| = N_2$, $N = N_1 + N_2$, and $X' \subseteq X$ be the set of variables for which there is decision node of $\mathcal{F}$ or $\overline{\mathcal{F}}$ labeled by a variable of $X'$. Then, the \textnormal{SDD} $C$ resulting from Simulation \ref{simulation:vtree_vee_one_obdds_to_sdds} contains at most $2N^2 + 3N$ nodes.
\end{lemma}

\begin{proof}
	The nodes of $C$ are tuple $(u,v) \in Y \times Z$. By definition of $Y$ and $Z$ in Simulation \ref{simulation:vtree_vee_one_obdds_to_sdds} we have $Y = |\mathcal{F}| + |\overline{\mathcal{F}}| = N_1 + N_2$ and $Z = N_1 + N_2 + |X'| + 3$. Hence, $C$ contains at most $(N_1 + N_2) \cdot (N_1 + N_2 + |X'| + 3)$ nodes. Furthermore, by assumption $\mathcal{F}$ and $\overline{\mathcal{F}}$ contain at least one node for each variable $x \in X'$. Therefore, we also have $N_1 + N_2 \geq |X'|$. Altogether, we get the following quadratic upper bound:
	\begin{eqnarray*}
		|C| &\leq& (N_1 + N_2) \cdot (N_1 + N_2 + |X'| + 3) \\
		&=& N \cdot (N + |X'| + 3) \\
		&\leq& N \cdot (2N + 3) = 2N^2 + 3N \in \mathcal{O}(N^2)\,.
	\end{eqnarray*} 
	
\end{proof}

%\noindent
Simulation \ref{simulation:vtree_vee_one_obdds_to_sdds} maps each node $u \in (V \cup \overline{V})$ to a node $(u, \emptyset)$ of $C$. In order to show that $C$ is a syntactically correct SDD computing the same function as $\mathcal{F}$, we will prove that each node $(u, \emptyset)$ of $C$ is the root of a syntactically correct SDD $C_{(u, \emptyset)}$ which computes the same function as $\mathcal{F}_u$ or $\overline{\mathcal{F}}_u$. For this purpose, we map each node $u \in (V \cup \overline{V})$ to a node $v$ of $T$ such that we can show that $C_{(u, \emptyset)}$ respects subtree $T_v$.

\begin{definition}
	Let $T$ be the vtree as defined in Simulation \ref{simulation:vtree_vee_one_obdds_to_sdds} and $u \in (V \cup \overline{V})$ be an inner node of the given $\vee_1$-$\textnormal{OBDDs}$. We use the function \emph{node} in order to map inner nodes of $\mathcal{F}$ and $\overline{\mathcal{F}}$ to nodes of $T$ in the following way:
	\begin{eqnarray*}
		\textnormal{node}(u) :=
		\begin{cases}
			v_i & \hspace{-0.25cm}\mbox{, } u \textnormal{ is an} \vee\textnormal{-node, } x_i \in \vars{u}, \,\nexists x_j \in \vars{u} \textnormal{ such that } \\
			& \hspace{-0.05cm} x_j < x_i \textnormal{ w.r.t. } \pi. \\
			v_i' & \hspace{-0.25cm}\mbox{, } u \textnormal{ is not an} \vee\textnormal{-node, } x_i \in \vars{u}, \,\nexists x_j \in \vars{u} \textnormal{ such that } \\
			& \hspace{-0.05cm} x_j < x_i \textnormal{ w.r.t. } \pi. \\
		\end{cases}
	\end{eqnarray*}
\end{definition}

%\noindent
Now, we are ready to prove the stated properties of the SDDs $C_{(u, \emptyset)}$.

\begin{lemma}
	\label{lemma:hauptlemma_vee_one_obdd_to_sdd}
	Let $\mathcal{F}$ and $\overline{\mathcal{F}}$ be unambiguous nondeterministic \textnormal{OBDDs}  respecting the variable ordering $\pi = \textnormal{id}$, representing Boolean functions $f, \overline{f} \in B_n$. Let $C$ be the $\textnormal{SDD}$ resulting from Simulation \ref{simulation:vtree_vee_one_obdds_to_sdds}. Then, each node $(u, \emptyset)$ of $C$ is the root of a syntactically correct \textnormal{SDD} $C_{(u,\emptyset)}$ respecting the vtree $T_{v}$ of the inner node $v = \textnormal{node}(u)$. Moreover, $C_{(u,\emptyset)}$ represents the same Boolean function as $\mathcal{F}_u$ or $\overline{\mathcal{F}}_u$.
\end{lemma}

\begin{proofidea}
	Consider the different cases how the node $(u,\emptyset)$ was added to $C$ by the given simulation. We give a proof by induction on the depth $l$ of the subgraph $C_{(u,\emptyset)}$ of the SDD $C$ in Appendix B.
%\ref{appendix:hauptlemma_vee_one_obdd_to_sdd}. 
\end{proofidea}

%\noindent
As a consequence of Lemma \ref{lemma:hauptlemma_vee_one_obdd_to_sdd}, we know that $C$ is a syntactically correct SDD representing the same Boolean function as $\mathcal{F}$.

\begin{corollary}
	\label{corollary:syntaktische_und_semantische_korrektheit_vee_one_obdd_to_sdd}
	Let $\mathcal{F}$ and $\overline{\mathcal{F}}$ be unambiguous nondeterministic \textnormal{OBDDs} respecting the variable ordering $\pi = \textnormal{id}$, representing Boolean functions $f, \overline{f} \in B_n$. Then, $C$ is a syntactically correct $\textnormal{SDD}$ respecting the vtree $T$ as defined in the simulation. Furthermore, $C$ represents $f$.
\end{corollary}

\begin{proof}
	The root of $C$ is given by the node $(\rootset{\mathcal{F}}, \emptyset)$ as depicted in Simulation \ref{simulation:vtree_vee_one_obdds_to_sdds}. We use Lemma \ref{lemma:hauptlemma_vee_one_obdd_to_sdd} in order to see that $C = C_{(\rootset{\mathcal{F}}, \emptyset)}$ is a syntactically correct $\textnormal{SDD}$ respecting the vtree $T_v$ with $v = \nodefunc{\rootset{\mathcal{F}}}$ and representing the same Boolean function as $\mathcal{F}$. Here we have $v = v_i$ or $v = v_i'$ for $i \in \{1,\dots, n\}$. Thus, $C$ is also respecting $T$. 
\end{proof}

\begin{theorem}\label{thm:transformation_into_sdd}
Let $f$ be a Boolean function such that $f$ and $\overline{f}$ can be represented by polynomial-size 
unambiguous nondeterministic \textnormal{OBDDs} respecting the same variable ordering. 
Then, $f$ can also be represented by polynomial-size \textnormal{SDDs}.
\end{theorem}

\begin{proof}
	By assumption there exist polynomial-size unambiguous nondeterministic OBDDs $\mathcal{F}$ and $\overline{\mathcal{F}}$ respecting the same variable ordering and representing $f$ and $\overline{f}$, respectively. We use Simulation \ref{simulation:vtree_vee_one_obdds_to_sdds} in order to get the SDD $C$. On the one hand we know by Lemma \ref{lemma:hauptlemma_vee_one_obdd_to_sdd} that $C$ is syntactically correct and represents the same function as $\mathcal{F}$. On the other hand we know by Lemma \ref{lemma:size_vee_one_obdd_to_sdd} that the increase in size is at most quadratic in $|\mathcal{F}| + |\overline{\mathcal{F}}|$. 
\end{proof}

%\noindent
If we only have a representation of $f$ as a polynomial-size unambiguous nondeterministic OBDD, 
we can modify Simulation \ref{simulation:vtree_vee_one_obdds_to_sdds} in order to get an equivalent
structured d-DNNF representing $f$ in polynomial size.

\begin{corollary}
	Let $f$ be a Boolean function representable by polynomial-size unambiguous nondeterministic \textnormal{OBDDs}. 
Then, $f$ can also be represented by structured \textnormal{d-DNNFs} of polynomial size.
\end{corollary}

\begin{proofidea}
	We can modify Simulation \ref{simulation:vtree_vee_one_obdds_to_sdds} such that in case (c) only edges to children of $\vee$-nodes will be added to the SDD $C$. For this purpose, we do not have to determine the sets $R^+$ and $\overline{R}^+$. Furthermore, we do not need an unambiguous nondeterministic OBDD representing $\overline{f}$ as input because we do not need a partition in order to represent Boolean functions that are computed at $\vee$-nodes of $\mathcal{F}$. 
\end{proofidea}
%%%%%%%%%%%%%%%%%%%%%%%%%%%%%%%%%%%%%%%%%%%%%%%%%%%%%%%%%%%%%%%%%%%%%%%%%%%%%%%%%%%%%%%%%%%%%%%%%%%%%%%%%%%%%%%%%%%%%%
%\section{Simulating SDNNFs by $\vee$-OBDDs}
\section{Simulating Structured DNNFs by Nondeterministic OBDDs}
\label{section:simulation_sdnnfs_by_vee_obdds}
In recent works it was shown how $\textnormal{DNNFs}$ can be simulated by equivalent nondeterministic $\textnormal{FBDDs}$ with an increase in size that remains bounded by a quasipolynomial factor \cite{BL15,Raz15}. These results were obtained by adapting a quasipolynomial simulation of decision-DNNFs by equivalent FBDDs proposed by Beame et al. \cite{BLR13}. In this section, we introduce another adaption in order to get a quasipolynomial simulation of structured DNNFs by equivalent nondeterministic OBDDs. Moreover, Razgon recently proved that there exists a quasipolynomial separation of SDDs (which are a subclass of d-SDNNFs) and nondeterministic OBDDs \cite{Raz17}. Therefore, the achieved upper bound concerning the increase in size is tight.

\subsection{Recap and main ideas}\label{subsection:simulation_sdnnfs_by_vee_obdds}

At the beginning, we will briefly recap the idea of constructing a nondeterministic FBDD $\mathcal{F}$ 
that computes the same Boolean function as a given DNNF $\mathcal{D}$ \cite{BL15,Raz15}. 
In order to construct $\mathcal{F}$ we have to remove all $\wedge$-nodes of $\mathcal{D}$ and replace them by decision nodes. Suppose we have an $\wedge$-node $u$ of $\mathcal{D}$ and its child nodes $u_l, u_r$. First, we need to find equivalent nondeterministic FBDDs $\mathcal{F}_{u_l}$ and $\mathcal{F}_{u_r}$ for the subgraphs $\mathcal{D}_{u_l}$ and $\mathcal{D}_{u_r}$, respectively. Next up, we need to combine these nondeterministic FBDDs in order to get a larger one for the expression $\Phi_u = \Phi_{u_l} \wedge \Phi_{u_r}$. For this purpose, redirect all $1$-sinks of $\mathcal{F}_{u_l}$ to the root of $\mathcal{F}_{u_r}$. That way we will get the needed conjunction of the given functions. Note that we get a syntactically correct nondeterministic FBDD by this conjunction since $\Phi_{u_l}$ and $\Phi_{u_r}$ depend on disjoint sets of variables because of the decomposability of $\mathcal{D}$. In general this first approach fails since the node $u_l$ can serve as input for more than one node. Then, it is not clear how to redirect the $1$-sinks of $\mathcal{F}_{u_l}$. Therefore, we make copies of subgraphs of $\mathcal{D}$ whenever the mentioned problem arises. Moreover, the children of $\wedge$-nodes will be reordered to bound the blow in size. An outgoing edge of an $\wedge$-node will be classified as a \emph{light edge}, if the subgraph of $\mathcal{D}$ that is connected by this edge does not contain more $\wedge$-nodes than the subgraph which is connected via the other edge. The latter will then be called a \emph{heavy edge}. If $(u,u_l)$ is the light edge of $u$, we redirect the $1$-sinks of $\mathcal{F}_{u_l}$ to the root of $\mathcal{F}_{u_r}$. As a consequence, each variable mentioned in $\mathcal{F}_{u_l}$ will be queried before every other variable mentioned in $\mathcal{F}_{u_r}$.\\

%\noindent
For the following adaption it is crucial to observe that the order in which the functions $\Phi_{u_l}$ and $\Phi_{u_r}$ will be evaluated (and therefore the order of queried variables) essentially depends on the definition of light and heavy edges. On the one hand, we will modify the presented definition of light and heavy edges with the aid of the vtree of a given SDNNF in order to obtain a variable ordering for the constructed nondeterministic OBDD. 
On the other hand, this new definition of light and heavy edges also ensures that the increase in size 
remains bounded by a quasipolynomial factor.
While the light and heavy edges of an $\wedge$-node are determined individually in the simulation of DNNFs 
by nondeterministic FBDDs, we will follow a more global approach using the information of a vtree to get a 
variable ordering. 
%We need the following definitions to specify the idea.

%\begin{definition}[\cite{PD08}]
%	\label{definition:dnnf_respecting_vtree}
%	Let $T$ be a vtree for the set of variables $X$ and $\mathcal{D}$ be a $\textnormal{DNNF}$. We say that $\mathcal{D}$ \emph{respects} the vtree $T$, if for every $\wedge$-node $u$ of $\mathcal{D}$ with children $u_l, u_r$, there is a node $v$ of $T$ with children $v_l, v_r$ such that $\vars{u_l} \subseteq \vars{v_l}$ and $\vars{u_r} \subseteq \vars{v_r}$.
%\end{definition}

%\noindent
%Observe that for each $\wedge$-node $u$ in definition \ref{definition:dnnf_respecting_vtree} there is only one node $v$ of $T$ fulfilling the mentioned condition. Therefore, we call $v$ the \emph{decomposition node} of $u$ (notation $\dnode{u} = v$). 

%\begin{definition}[\cite{PD08}]
%	The set of all $\textnormal{DNNFs}$ that respect a given vtree $T$ is denoted by $\textnormal{DNNF}_T$. Moreover, \emph{structured} $\textnormal{DNNF}$ (SDNNF) is the set containing all $\textnormal{DNNF}_T$ for any $T$.
%\end{definition}

%\noindent
%Now, 
We know that the variables which can appear in the subgraphs $\mathcal{D}_{u_l}$ and $\mathcal{D}_{u_r}$ 
of an $\wedge$-node $u$ in a DNNF$_T$ with decomposition node $v$ are restricted to the variables 
mentioned in $T_{v_l}$ and $T_{v_r}$, respectively. The key idea is to globally define 
the light and heavy edges of all $\wedge$-nodes of a $\textnormal{DNNF}_T$ which have the same decomposition node. 
We introduce the following quantities to formalize this approach.

\begin{definition}
	Let $T$ be a vtree for the set of variables $X$ and $\mathcal{D}$ be a $\textnormal{DNNF}_T$. Furthermore, let $v$ be an inner node of $T$ and $v_l, v_r$ its children. We define the following sets and quantities:
	\begin{eqnarray*}
		A^v &:=& \{u \;|\; u \textnormal{ is an $\wedge$-node of $\mathcal{D}$, } \;\dnode{u} = v.\}, \\
		M^v &:=& |A^v|, \quad 
		M^v_l := \sum_{w \in T_{v_l}} M^w, \quad
		M^v_r := \sum_{w \in T_{v_r}} M^w\,.
	\end{eqnarray*}
\end{definition}

%\noindent
Our aim is to determine in a common way for all $\wedge$-nodes of a set $A^v$ which subgraph can be reached via a light or heavy edge. Hereby, we achieve that all nondeterministic OBDDs representing a function $\Phi_u = \Phi_{u_l} \wedge \Phi_{u_r}$ for $u \in A^v$ will respect the same variable ordering. With an eye toward the size of the constructed nondeterministic OBDD, we will classify the edges as follows.

\begin{definition}
	Let $T$ be a vtree for the set of variables $X$ and $\mathcal{D}$ be a $\textnormal{DNNF}_T$. Moreover, let $u$ be an $\wedge$-node of $\mathcal{D}$ with children $u_l, u_r$ and $\dnode{u} = v$ for a node $v$ of $T$. We classify the edges $(u, u_l)$ and $(u, u_r)$ in the following way: If $M^v_l \leq M^v_r$ holds, we call $(u, u_l)$ a \emph{light edge} and $(u, u_r)$ a \emph{heavy edge}. Otherwise, we classify the edges vice versa. We call the remainder of the edges of $\mathcal{D}$ \emph{neutral edges}.
\end{definition}

%\noindent
In order to define the light and heavy edges we used the fact that given an $\wedge$-node $u$ of $\mathcal{D}$ with $\dnode{u} = v$ the number of $\wedge$-nodes that can occur in the subgraphs $\mathcal{D}_{u_l}$ and $\mathcal{D}_{u_r}$ is restricted by $M^v_l$ and $M^v_r$, respectively. Thus, each time we cross a light edge on a path from the root 
to a leaf the number of $\wedge$-nodes that can possibly occur in the next lower subgraph will be halved. Next, we will use the quantities $M^v_l$ and $M^v_r$ in the same way to define a variable ordering.

\begin{definition}
	Let $\mathcal{D}$ be a $\textnormal{DNNF}_T$ and $T$ be a vtree for the set of variables $X = \{x_1, \dots, x_n\}$. For a pair of variables $x_i, x_j \in X$ with $i \neq j$ let $v$ be the unique node of $T$ with children $v_l,v_r$ such that $x_i \in \vars{v_l}$ and $x_j \in \vars{v_r}$ holds. Then, we order $x_i < x_j$, if and only if $M^v_l \leq M^v_r$. Otherwise, we arrange $x_j < x_i$. We define $\pi(\mathcal{D},T)$ to be the variable ordering induced by the previously defined relation $<$.
\end{definition}

%\noindent
So, why do we get a variable ordering by the defined relation? Intuitively, starting from the root $v$ of a given vtree $T$ we order the variables that occur in $T_{v_l}$ and $T_{v_r}$ such that each variable of $\vars{v_l}$ precedes each variable of $\vars{v_r}$ w.r.t. to $<$ or vice versa. Afterwards, we recursively proceed with the nodes $v_l$ and $v_r$. Later on, we will formally prove that $\pi(\mathcal{D},T)$ is the variable ordering of the constructed nondeterministic OBDD $\mathcal{F}$. We need the following sets in order to define the simulation.

\begin{definition}[\cite{BLR13,BL15}]
	Fix a $\textnormal{DNNF}_T$ $\mathcal{D}$. For a node $u$ in $\mathcal{D}$ and a path $P$ from the root to $u$, let $S(P)$ be the set of light edges along $P$ and $S(u) := \{ S(P) \;|\; P \textnormal{ is a path from the root to } u\}$.
\end{definition}

%\noindent
While we adjusted the definitions of light and heavy edges, we will use the same simulation proposed by Beame et al. in order to construct the nondeterministic OBDD \cite{BLR13,BL15}. We will interpret a leaf of the given $\textnormal{DNNF}_T$ $\mathcal{D}$ labeled by a variable $x_i \in X$ as a decision node that points to a $0$-sink if $x_i = 0$ and to a $1$-sink if $x_i = 1$, and vice versa for a leaf labeled by $\overline{x_i}$. Now, by the following simulation we get a nondeterministic OBDD with additional unlabeled nodes that can be removed in a further step.

\begin{simulation}[\cite{BLR13,BL15}]
	\label{simulation:structured_dnnf_to_vee_obdd}
	Let $\mathcal{D}$ be a $\textnormal{DNNF}_T$ and $T$ a vtree for the set of variables $X = \{x_1, \dots, x_n\}$. We will construct a nondeterministic $\textnormal{OBDD}$ $\mathcal{F}$ that computes the same Boolean function as $\mathcal{D}$. Its nodes are pairs $(u,s)$ where $u$ is a node of $\mathcal{D}$ and the set of light edges $s$ belongs to $S(u)$. The nodes $u' = (u,s)$ of $\mathcal{F}$ will be labeled in the following way:
	
	\begin{itemize}
		\item[(i)] If $u$ is a decision node for a variable $x_i \in X$ in $\mathcal{D}$, then $u'$ is a decision node for the same variable in $\mathcal{F}$.
		\item[(ii)] If $u$ is an $\wedge$-node in $\mathcal{D}$, then $u'$ remains unlabeled in $\mathcal{F}$.
		\item[(iii)] If $u$ is an $\vee$-node in $\mathcal{D}$, then $u'$ is also an $\vee$-node in $\mathcal{F}$.
		\item[(iv)] If $u$ is a $0$-sink in $\mathcal{D}$, then $u'$ is also a $0$-sink in $\mathcal{F}$.
		\item[(v)] If $u$ is a $1$-sink in $\mathcal{D}$ and $s = \emptyset$, then $u'$ is also a $1$-sink in $\mathcal{F}$. Otherwise, $u'$ remains unlabeled.
	\end{itemize}
	The node $(\textnormal{root}(\mathcal{D}), \emptyset)$ is the root of $\mathcal{F}$. The edges in $\mathcal{F}$ are of three types: 
	
	\begin{enumerate}
		\item For each light edge $e = (u,v)$ in $\mathcal{D}$ and each $s \in S(u)$, add the edge $((u,s), (v,s \cup \{e\}))$ to $\mathcal{F}$.
		\item For each neutral edge $e = (u,v)$ in $\mathcal{D}$ and each $s \in S(u)$, add the edge $((u,s), (v,s))$ to $\mathcal{F}$.
		\item For each heavy edge $(u,v_r)$ with corresponding light edge $e = (u,v_l)$, each $s \in S(u)$ and each $1$-sink $w$ in $\mathcal{D}_{v_l}$, add the edge $((w, s \cup \{e\}),(v_r, s))$ to $\mathcal{F}$.
	\end{enumerate}
\end{simulation}

%\noindent
In Figure \ref{figure:example_s_dnnf_to_vee_obdd} we give an example for the adapted simulation. The resulting nondeterministic OBDD $\mathcal{F}$ respects the variable ordering given by the sequence $w, x, y, z$. Note that we would only get a nondeterministic FBDD by the original simulation since the light edge $e_2$ would be classified as a heavy edge. On the one hand, there would exist a path in the resulting nondeterministic FBDD where $w < z$ holds. On the other hand, there would also be a path where $z < w$ holds. Hence, we cannot find a corresponding variable ordering.

\begin{figure}[t]
	\centering
	\begin{subfigure}[b]{0.3\textwidth}
		\hspace*{-1.5cm}
		\resizebox{!}{0.5\textheight}{\begin{tikzpicture}[>= stealth']

\tikzset{
	invisible/.style={opacity=0},
	visible on/.style={alt=#1{}{invisible}},
	alt/.code args={<#1>#2#3}{%
		\alt<#1>{\pgfkeysalso{#2}}{\pgfkeysalso{#3}} % \pgfkeysalso doesn't change the path
	},
}

\tikzstyle{vertex} = [draw, circle]
\tikzstyle{sink} = [draw, rectangle]
\tikzstyle{vertexVtree} = [draw, circle, inner sep = 2pt]

% structured DNNF --------------------------------------------------------
% ------------------------------------------------------------------------

% 1. Ebene
\node[vertex] (or) at (0,0) [label={[black, font = \footnotesize] above right:{$1$}}] {$\vee$};

% 2. Ebene
\node[vertex] (and1) [below left = of or, label={[black, font = \footnotesize] above left:{$2$}}, , label={[itemizeblue, font = \small] right:{$I$}}] {$\wedge$};

\node[vertex] (and2) [below right = of or, label={[black, font = \footnotesize] above right:{$3$}}, , label={[itemizeblue, font = \small] left:{$I$}}] {$\wedge$};

% 3. Ebene
\node[vertex] (w) [below left = of and1, label={[black, font = \footnotesize] above left:{$4$}}] {$w$};

\node[vertex] (and3) [below right = of and1, label={[black, font = \footnotesize] above right:{$5$}}, label={[itemizeblue, font = \small] left:{$III$}}] {$\wedge$};

\node[vertex] (and4) [below right = of and2, label={[black, font = \footnotesize] above right:{$6$}}, label={[itemizeblue, font = \small] left:{$II$}}] {$\wedge$};

% 4. Ebene
\node[sink] (w_zero) [below left = 1 and 0.3 of w, label={[black, font = \footnotesize] above left:{$7$}}] {$0$};
\node[sink] (w_one) [below right = 1 and 0.3 of w, label={[black, font = \footnotesize] above right:{$8$}}] {$1$};

\node[vertex] (y) [below left = 1 and 0.3 of and3, label={[black, font = \footnotesize] above right:{$9$}}] {$y$};

\node[vertex] (z) [below right = 1 and 0.3 of and3, label={[black, font = \footnotesize] above right:{$10$}}] {$z$};

\node[vertex] (w_2) [below left = 1 and 0.3 of and4, label={[black, font = \footnotesize] above right:{$11$}}] {$w$};
\node[vertex] (x) [below right = 1 and 0.3 of and4, label={[black, font = \footnotesize] above right:{$12$}}] {$x$};

% 5. Ebene
\node[sink] (y_zero) [below left = 1 and 0.3 of y, label={[black, font = \footnotesize] below:{$13$}}] {$0$};
\node[sink] (y_one) [below = 0.9 of y, label={[black, font = \footnotesize] below:{$14$}}] {$1$};

\node[sink] (z_zero) [below left = 1 and 0.3 of z, label={[black, font = \footnotesize] below:{$15$}}] {$0$};
\node[sink] (z_one) [below = 0.9 of z, label={[black, font = \footnotesize] below:{$16$}}] {$1$};

\node[sink] (w_2_zero) [below left = 1 and 0.3 of w_2, label={[black, font = \footnotesize] below:{$17$}}] {$0$};
\node[sink] (w_2_one) [below = 0.9 of w_2, label={[black, font = \footnotesize] below:{$18$}}] {$1$};
\node[sink] (x_zero) [below left = 1 and 0.3 of x, label={[black, font = \footnotesize] below:{$19$}}] {$0$};
\node[sink] (x_one) [below = 0.9 of x, label={[black, font = \footnotesize] below:{$20$}}] {$1$};

% 1. -> 2.
\draw (or) to (and1);
\draw (or) to (and2);

% 2. -> 3.
\draw (and1) to node [auto, swap] {$e_1$} (w);
\draw (and1) to (and3);
\draw (and2) to (z);
\draw (and2) to node [auto] {$e_2$} (and4);

% 3. -> 4.
\draw[dashed, ->] (w) to (w_zero);
\draw[->] (w) to (w_one);

\draw (and3) to node [auto, swap] {$e_3$} (y);
\draw (and3) to (z);
\draw (and4) to node [auto, swap] {$e_4$} (w_2);
\draw (and4) to (x);

% 4. -> 5.
\draw[->] (y) to (y_one);
\draw[->, dashed] (y) to (y_zero);
\draw[dashed, ->] (z) to (z_zero);
\draw[->] (z) to (z_one);

\draw[->] (w_2) to (w_2_zero);
\draw[->, dashed] (w_2) to (w_2_one);
\draw[->] (x) to (x_zero);
\draw[->, dashed] (x) to (x_one);

% Variablenbaum ----------------------------------------------------------
% ------------------------------------------------------------------------
% Knoten
\node[vertexVtree] (v1) [below = 2.25 of z_zero, label={[itemizeblue, font = \small] above right:{$I$}}] {};
\node[vertexVtree] (v2) [below left = of v1, label={[itemizeblue, label distance = 0.25cm, font = \small] above left:{$II$}}] {};
\node[vertexVtree] (v3) [below right = of v1, label={[itemizeblue, label distance = 0.25cm, font = \small] above right:{$III$}}] {};
	
\node (w) [below left = 1 and 0.3 of v2] {$w$};
\node (x) [below right = 1 and 0.3 of v2] {$x$};
\node (y) [below left = 1 and 0.3 of v3] {$y$};
\node (z) [below right = 1 and 0.3 of v3] {$z$};

% Kanten
\draw (v1) to (v2);
\draw (v1) to (v3);
	
\draw (v2) to (w);
\draw (v2) to (x);
	
\draw (v3) to (y);
\draw (v3) to (z);

\end{tikzpicture}}
		\caption{$\textnormal{DNNF}_T$ $\mathcal{D}$ and vtree $T$}
		\label{figure:example_s_dnnf_to_vee_obdd_a}
	\end{subfigure}
	\begin{subfigure}[b]{0.6\textwidth}
		\hspace{3cm}\resizebox{!}{0.55\textheight}{\begin{tikzpicture}[>= stealth']

\tikzstyle{vertex} = [draw, circle, minimum size=0.7cm]
\tikzstyle{sink} = [draw, rectangle]

% 1. Ebene
\node[vertex] (or) at (0,0) [label={[black, font = \footnotesize] above right:{$(1, \emptyset)$}}] {$\vee$};

% 2. Ebene
\node[vertex] (noop_1) [below left = of or, label={[black, font = \footnotesize] above left:{$(2, \emptyset)$}}] {};

\node[vertex] (noop_2) [below right = of or, label={[black, font = \footnotesize] above right:{$(3, \emptyset)$}}] {};

% 3. Ebene
\node[vertex] (w_1) [below = of noop_1, label={[black, font = \footnotesize] above left:{$(4, e_1)$}}] {$w$};

\node[vertex] (w_2) [below = of noop_2, label={[black, font = \footnotesize] above right:{$(6, e_2)$}}] {};

% 4. Ebene
\node[sink] (w_1_zero) [below left = 1 and 0.7 of w_1, label={[black, font = \footnotesize] below:{$(7, e_1)$}}] {$0$};

\node[vertex] (noop_3) [below = of w_1, label={[black, font = \footnotesize] above right:{$(8, e_1)$}}] {};

\node[vertex] (noop_4) [below = of w_2, label={[black, font = \footnotesize] above right:{$(11, e_2e_4)$}}] {$w$};

\node[sink] (w_2_zero) [below left = 1 and 0.7 of noop_4, label={[black, font = \footnotesize] below:{$(17, e_2e_4)$}}] {$0$};

% 5. Ebene
\node[vertex] (noop_5) [below = of noop_3, label={[black, font = \footnotesize] above left:{$(5, \emptyset)$}}] {};

\node[vertex] (noop_6) [below = of noop_4, label={[black, font = \footnotesize] above right:{$(18, e_2e_4)$}}] {};

% 5. Ebene
\node[vertex] (x_1) [below = of noop_5, label={[black, font = \footnotesize] above left:{$(9, e_3)$}}] {$y$};

\node[vertex] (x_2) [below = of noop_6, label={[black, font = \footnotesize] above right:{$(12, e_2)$}}] {$x$};

% 6. Ebene
\node[sink] (x_1_zero) [below left = 1 and 0.7 of x_1, label={[black, font = \footnotesize] below:{$(13, e_3)$}}] {$0$};

\node[sink] (x_2_zero) [below left = 1 and 0.7 of x_2, label={[black, font = \footnotesize] below:{$(19,e_2)$}}] {$0$};

\node[vertex] (noop_7) [below = of x_1, label={[black, font = \footnotesize] above right:{$(14, e_3)$}}] {};

\node[vertex] (noop_8) [below = of x_2, label={[black, font = \footnotesize] above right:{$(20, e_2)$}}] {};

% 7. Ebene
\node[vertex] (y) [below right = of noop_7, label={[black, font = \footnotesize] below:{$(10, \emptyset)$}}] {$z$};

% 8. Ebene
\node[sink] (y_zero) [left = 1 of y, label={[black, font = \footnotesize] below:{$(15, \emptyset)$}}] {$0$};
\node[sink] (y_one) [right = 1 of y, label={[black, font = \footnotesize] below:{$(16, \emptyset)$}}] {$1$};

% 1. -> 2.
\draw[->] (or) to (noop_1);
\draw[->] (or) to (noop_2);

% 2. -> 3.
\draw[->] (noop_1) to (w_1);
\draw[->] (noop_2) to (w_2);

% 2. -> 3.
\draw[->, dashed] (w_1) to (w_1_zero);
\draw[->] (w_1) to (noop_3);
\draw[->] (noop_4) to (w_2_zero);
\draw[->] (w_2) to (noop_4);

% 3. -> 4.
\draw[->] (noop_3) to (noop_5);
\draw[->, dashed] (noop_4) to (noop_6);

% 4. -> 5.
\draw[->] (noop_5) to (x_1);
\draw[->] (noop_6) to (x_2);

% 5. -> 6.
\draw[->, dashed] (x_1) to (x_1_zero);
\draw[->] (x_1) to (noop_7);
\draw[->] (x_2) to (x_2_zero);
\draw[->, dashed] (x_2) to (noop_8);

% 6. -> 7.
\draw[->] (noop_7) to (y);
\draw[->] (noop_8) to (y);

% 7. -> 8.
\draw[->, dashed] (y) to (y_zero);
\draw[->] (y) to (y_one);

\end{tikzpicture}}
		\caption{$\vee$-OBDD $\mathcal{F}$}
		\label{figure:example_s_dnnf_to_vee_obdd_b}
	\end{subfigure}	
	\caption{(a) A $\textnormal{DNNF}_T$ $\mathcal{D}$ whose leaves are interpreted as decision nodes respecting the depicted vtree $T$ for the set of variables $X = \{w,x,y,z\}$. $\mathcal{D}$ computes the Boolean function $\Phi_{\mathcal{D}}(w,x,y,z) = wyz \vee \overline{w} \hspace*{-1pt}\ \overline{x}z$. The light edges are marked by $e_1, \dots, e_4$ and the decomposition nodes are labeled by I, II and III as in the vtree. 
(b) The nondeterministic OBDD $\mathcal{F}$ resulting from the given simulation with input $\mathcal{D}$. The variable ordering of $\mathcal{F}$ is given by $\pi(\mathcal{D}, T)$ resulting in the sequence $w,x,y,z$.}
	\label{figure:example_s_dnnf_to_vee_obdd}
\end{figure}
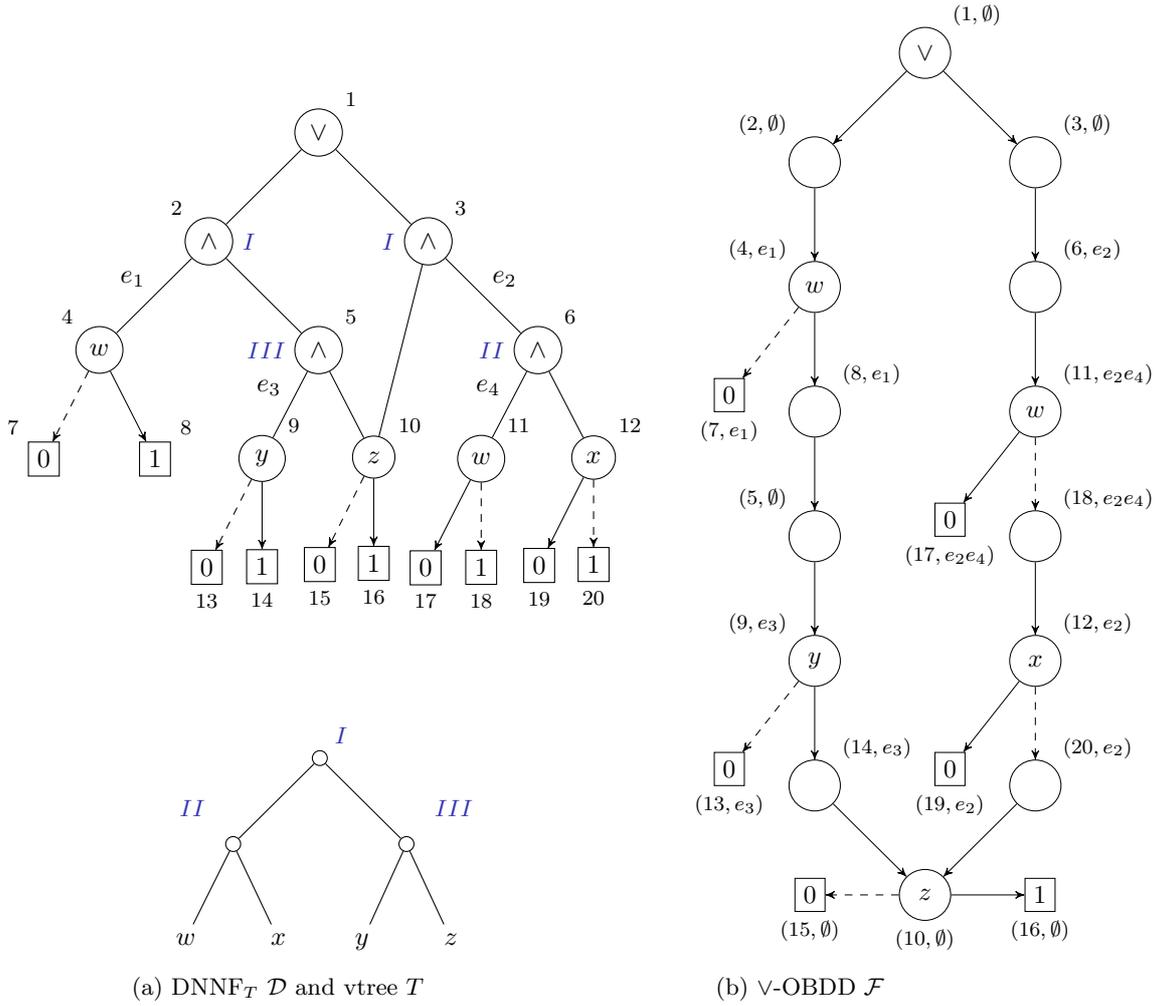

\subsection{Size and correctness}\label{subsection:simulation_sdnnfs_by_vee_obdds_size_and_correctness}

First, we have a look at the size of the constructed nondeterministic OBDD.

\begin{lemma}
	\label{lemma:sdnnfs_to_vee_obdds_size}
	Let $\mathcal{D}$ be a $\textnormal{DNNF}_T$ with $M$ $\wedge$-nodes, $N$ be the total number of nodes and $L$ the maximum number of light edges from the root to a leaf. Then, the constructed nondeterministic $\textnormal{OBDD}$ $\mathcal{F}$ of Simulation \ref{simulation:structured_dnnf_to_vee_obdd} contains at most $N(M+1)^L \leq N \cdot 2^{\log^2(N)}$ nodes.
\end{lemma}

\begin{proof}
	The upper bound of $|\mathcal{F}| \leq N(M+1)^L$ can be derived analogously to the upper bound of the simulation of $\textnormal{DNNF}$ by $\vee$-$\textnormal{FBDDs}$ from Beame and Liew \cite{BL15}. For that to happen, one has to determine the number of nodes that are created by the simulation. Now, we have a look at the second upper bound depending only on $N$.\\
	
	%\noindent
	Consider a path from the root of $\mathcal{D}$ to a leaf containing $L$ light edges that must exist by premise. For an $\wedge$-node $u$ on that path with children $u_l,u_r$ let $v$ be the node of $T$ such that $\dnode{u} = v$. Let also be $v_l, v_r$ the children of $v$. By definition there exist $M^v + M^v_l + M^v_r$ $\wedge$-nodes having a decomposition node which is located in the subtree $T_v$. The subgraph $\mathcal{D}_{u_l}$ contains at most $M^v_l$ $\wedge$-nodes. Assume to the contrary that there exists an $\wedge$-node $u'$ in $\mathcal{D}_{u_l}$ such that $\dnode{u'} = v'$ for a node $v'$ which is not located in $T_{v_l}$. Then, $\mathcal{D}_{u_l}$ would contain at least one node labeled by a variable $x \notin \vars{v_l}$ that would be a contradiction to the premise of $\mathcal{D}$ being a $\textnormal{DNNF}_T$. Analogously, the subgraph $\mathcal{D}_{u_r}$ contains at most $M^v_r$ $\wedge$-nodes.\\
	
	%\noindent
	W.l.o.g. let $(u,u_l)$ be the light edge of the $\wedge$-node $u$. Therefore, it holds that $M^v_l \leq M^v_r$. I.e., the number of $\wedge$-nodes which can be located in $\mathcal{D}_{u_l}$ is at most half the number of $\wedge$-nodes that can possibly be located in $\mathcal{D}_u$. If $(u,u_r)$ is the light edge of $u$, an analog result can be derived. Hence, each time we pass a light edge on the given path, the number of $\wedge$-nodes that can be located in the next lower subgraph is at least halved. Moreover, in addition to the $M$ $\wedge$-nodes there has to be at least one node labeled by a variable or literal because there must be $\wedge$-nodes which are connected to literals or variables as inputs. Altogether, we get $N > M \geq 2^L$. Now, we get the claimed upper bound by using the mentioned inequalities:
	\begin{eqnarray*}
		N(M + 1)^L &=& N \cdot 2^{\log((M + 1)^L)} = N \cdot 2^{L\log(M + 1)} \\
		&\leq& N \cdot 2^{\log(M)\log(M + 1)} \\
		&\leq& N \cdot 2^{\log^2(N)}\,. 
	\end{eqnarray*}
\end{proof}

%\noindent
Next up, we show an extension of Lemma 5.4 from Beame and Liew \cite{BL15}
which can subsequently used in order to show that the constructed nondeterministic OBDD is syntactically correct. Let $\mathcal{D}_1, \mathcal{D}_2$ be two SDNNFs. We use the notation $\mathcal{D}_1 \subset \mathcal{D}_2$ which means that $\mathcal{D}_1$ is a subgraph of $\mathcal{D}_2$. Moreover, for two variables $x_i, x_j \in X$ we have $x_i \leq x_j$ if and only if $x_i < x_j$ w.r.t. $\pi^*$ or $x_i = x_j$ holds.

\begin{lemma}
	\label{lemma:s_dnnf_to_vee_obdd_help_correct}
	Let $T$ be a vtree for the variable set $X = \{x_1, \dots, x_n\}$,
        $\mathcal{D}$ be a $\textnormal{DNNF}_T$, 
 and $\mathcal{F}$ be the nondeterministic $\textnormal{OBDD}$ resulting from Simulation \ref{simulation:structured_dnnf_to_vee_obdd}. Furthermore, let $\pi^* = \pi(\mathcal{D},T)$ be the induced variable ordering. If $u$ is a leaf in $\mathcal{D}$ labeled by a variable $x_i \in X$ and there exists a nontrivial path (consisting of at least one edge) between $(u,s)$ and $(v,s')$ in $\mathcal{F}$, then there exists no node in $\mathcal{D}_v$ labeled by a variable $x_j$ fulfilling $x_j \leq x_i$ w.r.t. $\pi^*$.
\end{lemma}

\begin{proofidea}
	If we assume to the contrary that there exists such a nontrivial path between $(u,s)$ and $(v,s')$ in $\mathcal{F}$ and there is a node labeled by a variable $x_j \leq x_i$ in $\mathcal{D}_v$, we either get a violation of the decomposability of $\mathcal{D}$ or that $\mathcal{D}$ does not respect $T$ which leads to a contradiction. The entire proof can be found in 
Appendix C.
%\ref{appendix:s_dnnf_to_vee_obdd_help_correct}. 
\end{proofidea}

%\noindent
Now, we are able to prove that the constructed nondeterministic OBDD is syntactically correct.

\begin{lemma}
	Let $T$ be a vtree for the set of variables $X = \{x_1, \dots, x_n\}$,
        $\mathcal{D}$ be a $\textnormal{DNNF}_T$, 
and $\mathcal{F}$ the nondeterministic $\textnormal{OBDD}$ resulting from Simulation \ref{simulation:structured_dnnf_to_vee_obdd}. Then, $\mathcal{F}$ is a syntactically correct nondeterministic $\textnormal{OBDD}$ respecting the variable ordering $\pi^* = \pi(\mathcal{D}, T)$.
\end{lemma}

\begin{proof}
	We have to show that $\mathcal{F}$ is a $\textnormal{BDD}$ which suffices the property that decision nodes are labeled by a subsequence of $\pi^*$ on each directed path.\\
	
	%\noindent
	First, we could show with the help of Lemma \ref{lemma:s_dnnf_to_vee_obdd_help_correct} that $\mathcal{F}$ is a syntactically correct nondeterministic $\textnormal{FBDD}$ with further unlabeled nodes. This can be done like in the proof of Lemma 5.4. from Beame and Liew \cite{BL15}. Now, we only have to show that $\mathcal{F}$ is respecting the variable ordering $\pi^*$.\\
	
	%\noindent
	Suppose there is a directed path $P$ in $\mathcal{F}$ such that the decision nodes appearing on $P$ 
are not labeled by a subsequence of $\pi^*$. Then, there also exists a subpath of $P$ with nodes
$(u,s), \dots, (v,s')$ fulfilling the following properties: $(u,s)$ is a decision node labeled by a variable $x_i$, $(v,s')$ is a decision node labeled by $x_j$ with $i \neq j$, $x_j < x_i$ w.r.t. $\pi^*$. The node $(u,s)$ is labeled by $x_i$ in $\mathcal{F}$ because $u$ is a leaf in the given $\textnormal{DNNF}_T$ $\mathcal{D}$ labeled by the same variable. Analogously, we know that $v$ is a decision node labeled by $x_j$ in $\mathcal{D}$. By usage of Lemma \ref{lemma:s_dnnf_to_vee_obdd_help_correct} we know that the subgraph $\mathcal{D}_v$ does not contain a decision node labeled by a variable $x_j$ such that $x_j \leq x_i$ w.r.t. $\pi^*$. Now, we have the desired contradiction because $\mathcal{D}_v$ contains $v$ which is labeled by $x_j$ and $x_j < x_i$. 
\end{proof}

%\noindent
In the following we assume that $\vee$- and $\wedge$-nodes of the given $\textnormal{DNNF}$ do not have constants as inputs in order to simplify the proofs of correctness and completeness of the simulation. Otherwise, we could simplify a given $\textnormal{DNNF}$ by propagating the constants according to the semantics of $\vee$- and $\wedge$-nodes. Certificates are subgraphs of a given $\textnormal{DNNF}$ fulfilling the following properties.

\begin{definition}[\cite{BCM16}]
	Let $\mathcal{D}$ be a $\textnormal{DNNF}$ for the set of variables $X$. A \emph{certificate} of $\mathcal{D}$ is a $\textnormal{DNNF}$ $\mathcal{C}$ for $X$ with the following properties:
	\begin{itemize}
		\item[(i)] The \textnormal{DNNF} $\mathcal{C}$ is a subgraph of $\mathcal{D}$ ($\mathcal{C} \subset \mathcal{D}$).
		\item[(ii)] The roots (output gates) of $\mathcal{C}$ and $\mathcal{D}$ coincide.
		\item[(iii)] If $\mathcal{C}$ contains an $\wedge$-node $u$, $\mathcal{C}$ also contains each child node $v$ of $u$ and the edge $(u,v)$.
		\item[(iv)] If $\mathcal{C}$ contains an $\vee$-node $u$, $\mathcal{C}$ also contains exact one of the child nodes $v$ of $u$ and the edge $(u,v)$.
	\end{itemize}
\end{definition}

%\noindent
Since the fanin of $\wedge$-nodes is restricted by $2$ and because of the decomposability of $\mathcal{D}$ a certificate can be seen as a binary tree where each leaf is labeled by a different variable of $X$. Now, we define $1$-certificates in order to represent sets of satisfying inputs of a given $\textnormal{DNNF}$.

\begin{definition}
	A $1$\emph{-certificate} is a certificate with the following modifications: each leaf labeled by a literal $x$ is a decision node labeled by $x$ whose only outgoing edge labeled by $1$ leads to the $1$-sink and each leaf labeled by a literal $\overline{x}$ is a decision node labeled by $x$ whose only outgoing edge labeled by $0$ leads to the $1$-sink.  
\end{definition}

%\noindent
A $1$-certificate represents all assignments to the input variables where the labels of outgoing edges of decision nodes are chosen as assignments for the corresponding variables. Since a $1$-certificate does not have to contain a decision node for each input variable, the represented set of assignments can contain more than one element. Now, observe that according to the definition of $1$-certificates each $\vee$- and $\wedge$-node will evaluate to $1$ given an assignment of the defined set. Since the roots of a $1$-certificate and a given $\textnormal{DNNF}$ coincide, this set of assignments is also satisfying for the given $\textnormal{DNNF}$.\\

%\noindent
After introducing the notation of $1$-certificates we are ready to show the equivalence of the Boolean functions computed by $\mathcal{F}$ and $\mathcal{D}$. We will start with the correctness of the simulation, i.e., for each variable assignment $b$ we show that $\Phi_{\mathcal{F}}[b] = 1$ implies $\Phi_{\mathcal{D}}[b] = 1$.

\begin{lemma}
	\label{lemma:s_dnnf_to_vee_obdd_semantic_correctness}
	Let $\mathcal{F}$ be the nondeterministic $\textnormal{OBDD}$ resulting from Simulation \ref{simulation:structured_dnnf_to_vee_obdd} of a given $\textnormal{DNNF}_T$ $\mathcal{D}$. Then, for each accepting path for a (possibly partial) variable assignment $b$ in $\mathcal{F}$ there exists a $1$-certificate of $\mathcal{D}$ which represents $b$.
\end{lemma}

\begin{proofidea}
	Given an accepting path for a variable assignment $b$ in $\mathcal{F}$ we are able to reconstruct a $1$-certificate of $\mathcal{D}$ representing the same variable assignment by inspecting Simulation \ref{simulation:structured_dnnf_to_vee_obdd}. We give a formal proof by induction on the length $l$ of an accepting path in $\mathcal{F}$ in 
Appendix D.
%\ref{appendix:s_dnnf_to_vee_obdd_semantic_correctness}. 
\end{proofidea}

%\noindent
Next, we will show the completeness of the given simulation, i.e., for each variable assignment $b$ we show that $\Phi_{\mathcal{D}}[b] = 1$ implies $\Phi_{\mathcal{F}}[b] = 1$.

\begin{lemma}
	\label{lemma:s_dnnf_to_vee_obdd_vollstaendigkeit}
	Let $\mathcal{F}$ be the nondeterministic $\textnormal{OBDD}$ resulting from Simulation \ref{simulation:structured_dnnf_to_vee_obdd} of a given $\textnormal{DNNF}_T$ $\mathcal{D}$. Then, for each $1$-certificate of $\mathcal{D}$ representing a (possibly partial) variable assignment $b$ there exists an accepting path in $\mathcal{F}$ for $b$.
\end{lemma}

\begin{proofidea}
	Given a $1$-certificate $\mathcal{C}$ of $\mathcal{D}$ we can decompose $\mathcal{C}$ in order to get an accepting path in $\mathcal{F}$. We give a proof by induction on the depth $l$ (longest path from the root to a leaf) of a $1$-certificate of $\mathcal{D}$ in 
Appendix E.
%\ref{appendix:s_dnnf_to_vee_obdd_completeness}. 
\end{proofidea}

%\noindent
Now, we can derive the proposed equivalence of $\mathcal{F}$ and $\mathcal{D}$ by applying the last two lemmata.

\begin{lemma}
	The nondeterministic $\textnormal{OBDD}$ $\mathcal{F}$ computes the same Boolean function as the given $\textnormal{DNNF}_T$ $\mathcal{D}$. I.e., $\Phi_{\mathcal{F}}[b] = \Phi_{\mathcal{D}}[b]$ holds for each variable assignment $b$.
\end{lemma}

%\noindent
Altogether, we have shown that for each SDNNF there exists an equivalent nondeterministic OBDD with an increase in size that is at most quasipolynomial in $|\mathcal{D}|$. Let $L, M$ and $N$ be defined as in Lemma \ref{lemma:sdnnfs_to_vee_obdds_size}.

\begin{theorem}
	\label{theorem:sdnnf_to_vee_obdd}
	For any $\textnormal{DNNF}_T$ $\mathcal{D}$ there exists an equivalent nondeterministic $\textnormal{OBDD}$ $\mathcal{F}$ 
%computing the same Boolean function as $\mathcal{D}$, 
with at most $N(M+1)^L$ nodes and $\mathcal{F}$ can be constructed in time $\mathcal{O}(NM^L)$.
\end{theorem}

%\noindent
Using the described quasipolynomial simulation of SDNNF by nondeterministic OBDDs, 
we can derive lower bounds for SDNNFs (and also SDDs) from lower bounds for nondeterministic OBDDs. 
%In the next section, we have a short look at simulating (structured) d-DNNFs. Using the property that the given DNNFs are deterministic we will derive that the nondeterministic OBDDs constructed by Simulation \ref{simulation:structured_dnnf_to_vee_obdd} are unambiguous.

\section{Simulating (Structured) d-DNNFs}
\label{section:simulating_structured_d_dnnfs}
Independently, Beame and Liew and Razgon proved that $\textnormal{DNNFs}$ can be simulated by nondeterministic $\textnormal{FBDDs}$ with at most a quasipolynomial increase in size \cite{BL15,Raz15}. In the previous section, we have adapted this construction in order to get an analogous simulation of $\textnormal{SDNNFs}$ by nondeterministic $\textnormal{OBDDs}$. In this section, we will prove that both simulations can be used in order to simulate (structured) d-DNNFs by equivalent unambiguous nondeterministic FBDDs (OBDDs), respectively.\\

%\noindent
There are two key observations leading to the stated results. The first observation is that two different $1$-certificates of a given $\textnormal{d-DNNF}$ $\mathcal{D}$ do not represent a common satisfying input of $\mathcal{D}$.

\begin{lemma}
	\label{lemma:simulating_d_dnnfs_exactly_one_certificate}
	Let $\mathcal{D}$ be a  deterministic $\textnormal{DNNF}$ representing a Boolean function $\Phi_{\mathcal{D}}: \{0,1\}^n \rightarrow \{0,1\}$. Then, for each satisfying assignment $b \in \{0,1\}^n$ of $\Phi_{\mathcal{D}}$ there is exactly one $1$-certificate of $\mathcal{D}$ representing $b$.
\end{lemma}

\begin{proof}
	There has to be at least one $1$-certificate of $\mathcal{D}$ representing $b$. Otherwise, $b$ would not be a satisfying assignment of $\Phi_{\mathcal{D}}$. Now, suppose to the contrary there would be more $1$-certificates of $\mathcal{D}$ representing $b$. Let $\mathcal{C}_1$ and $\mathcal{C}_2$ be two of them. According to the definition of $1$-certificates we have $\rootset{\mathcal{C}_1} = \rootset{\mathcal{C}_2} = \rootset{\mathcal{D}}$. Hence, consider $\mathcal{C}_1$ and $\mathcal{C}_2$ starting from their common root. By definition of certificates we know that there has to be a common $\vee$-node $u$ of $\mathcal{C}_1$ and $\mathcal{C}_2$ such that $\mathcal{C}_1$ only contains the left child $u_l$ and $\mathcal{C}_2$ only contains the right child $u_r$ in order that $\mathcal{C}_1$ and $\mathcal{C}_2$ differ. The subtree $\mathcal{C}_{u_l}$ of $\mathcal{C}_1$ is a $1$-certificate of $\mathcal{D}_{u_l}$ representing $b$ because otherwise $C_1$ would be none of $\mathcal{D}$. Analogously, the subtree $C_{u_r}$ of $\mathcal{C}_2$ has to be a $1$-certificate of $\mathcal{D}_{u_r}$. However, this implies that the Boolean functions represented by $\mathcal{D}_{u_l}$ and $\mathcal{D}_{u_r}$ are not disjoint since $b$ is a satisfying assignment for both functions. This is a contradiction to the assumption of $\mathcal{D}$ being a $\textnormal{d-DNNF}$. 
\end{proof}

%\noindent
Now, the second observation is that the simulation from Beame and Liew (which is essentially given by Simulation \ref{simulation:structured_dnnf_to_vee_obdd}) maps each $1$-certificate of a given $\textnormal{DNNF}$ to a corresponding accepting path in the constructed nondeterministic $\textnormal{FBDD}$.

\begin{lemma}
	\label{lemma:simulating_d_dnnfs_same_count_of_certs_and_accepting_paths}
	Let $\mathcal{D}$ be a $\textnormal{DNNF}$ and $\mathcal{F}$ the nondeterministic \textnormal{FBDD} resulting from the simulation stated in \textnormal{\cite{BL15}}. Furthermore, let $b$ be a satisfying assignment. Then, $\mathcal{F}$ has as much accepting paths for $b$ as $\mathcal{D}$ has $1$-certificates representing $b$.
\end{lemma}

\begin{proof}
	Suppose to the contrary that there would exist more or less accepting paths for $b$ in $\mathcal{F}$ than $1$-certificates of $\mathcal{D}$ representing $b$.\\
	
	%\noindent
	\underline{Case 1:} There are less accepting paths in $\mathcal{F}$ than $1$-certificates of $\mathcal{D}$. Thus, according to Lemma \ref{lemma:s_dnnf_to_vee_obdd_vollstaendigkeit} (completeness) there exist two $1$-certificates $\mathcal{C}_1$ and $\mathcal{C}_2$ of $\mathcal{D}$ representing $b$ which are mapped to the same accepting path $P$ of $\mathcal{F}$ by the given simulation. Since $\mathcal{C}_1$ and $\mathcal{C}_2$ are different $1$-certificates of $\mathcal{D}$, one of the certificates must contain a node $u$ which is not contained in the other certificate. Otherwise, suppose they would consist of the same set of nodes. Then, $\mathcal{C}_1$ and $\mathcal{C}_2$ had to differ in their set of edges. But, the edge set of a $1$-certificate is determined by its node set according to the definition. W.l.o.g. let $\mathcal{C}_1$ be the certificate containing $u$. Now, we know that $\mathcal{C}_1$ was mapped to an accepting path of $\mathcal{F}$ by the given simulation containing a node $(u,s)$ for $s \in S(u)$. Since $\mathcal{C}_2$ does not contain $u$, $\mathcal{C}_2$ was mapped to an accepting path in $\mathcal{F}$ which does not contain a node $(u,s)$. However, this is a contradiction to the fact that $\mathcal{C}_1$ and $\mathcal{C}_2$ were both mapped to $P$.\\
	
	%\noindent
	\underline{Case 2:} There are more accepting paths for $b$ in $\mathcal{F}$ than $1$-certificates representing $b$. According to Lemma \ref{lemma:s_dnnf_to_vee_obdd_semantic_correctness} (correctness) for each accepting path in $\mathcal{F}$ there has to be a corresponding $1$-certificate of $\mathcal{D}$. Since there are more accepting paths for $b$ in $\mathcal{F}$ than $1$-certificates representing $b$, there have to be two different accepting path $P_1$ and $P_2$ which emerged from the same $1$-certificate of $\mathcal{D}$. However, the given simulation is a function which maps nodes and edges of $\mathcal{D}$ to nodes and edges of $\mathcal{F}$. Therefore, $P_1$ and $P_2$ have to be equal which leads to a contradiction. 
\end{proof}

%\noindent
By combining the last two lemmata we get the following result.

\begin{proposition}
	Let $\mathcal{D}$ be a $\textnormal{d-DNNF}$ and $\mathcal{F}$ be the nondeterministic \textnormal{FBDD} resulting from the simulation stated in \textnormal{\cite{BL15}}. Then, $\mathcal{F}$ is an unambiguous nondeterministic \textnormal{FBDD}.
\end{proposition}

\begin{proof}
	We have to show that for each variable assignment $b$ there exists at most one accepting path in $\mathcal{F}$. If $b$ is a non-satisfying assignment, we know from the equivalence of $\mathcal{D}$ and $\mathcal{F}$ that there is no accepting path for $b$ in $\mathcal{F}$. Now, let $b$ be a satisfying assignment of $\mathcal{D}$. By Lemma \ref{lemma:simulating_d_dnnfs_exactly_one_certificate} we know that there is exactly one $1$-certificate of $\mathcal{D}$ representing $b$. Furthermore, by Lemma \ref{lemma:simulating_d_dnnfs_same_count_of_certs_and_accepting_paths} we know that there is exactly one accepting path for $b$ in $\mathcal{F}$. In conclusion, for each variable assignment $b$ there exists at most one accepting path in $\mathcal{F}$. Therefore, $\mathcal{F}$ is an unambiguous nondeterministic $\textnormal{FBDD}$. 
\end{proof}

%\noindent
Since we only changed the definition of light and heavy edges in our simulation of SDNNFs by nondeterministic OBDDs, we easily obtain the next result analogously to Lemma \ref{lemma:simulating_d_dnnfs_same_count_of_certs_and_accepting_paths}.

\begin{lemma}
	\label{lemma:s_dnnf_to_vee_obdd_same_count_cert_and_calc}
	Let $\mathcal{D}$ be a $\textnormal{DNNF}_T$ and $\mathcal{F}$ be the nondeterministic \textnormal{OBDD} resulting from Simulation \ref{simulation:structured_dnnf_to_vee_obdd}. 
Besides, let $b$ be a satisfying assignment for the represented function. 
Then, there exists as many accepting paths for $b$ in $\mathcal{F}$ as there exists $1$-certificates in $\mathcal{D}$ representing $b$.
\end{lemma}

%\noindent
Therefore, given a $\textnormal{d-DNNF}_T$ our simulation yields an unambiguous nondeterministic OBDD.

\begin{proposition}
	\label{satz:simulation_s_d_dnnf_to_vee_one_obdd_result}
	Let $\mathcal{D}$ be a $\textnormal{d-DNNF}_T$ and $\mathcal{F}$ be the nondeterministic \textnormal{OBDD} resulting from Simulation \ref{simulation:structured_dnnf_to_vee_obdd}. Then, $\mathcal{F}$ is an unambiguous nondeterministic \textnormal{OBDD}. 
\end{proposition}

%%%%%%%%%%%%%%%%%%%%%%%%%%%%%%%%%%%%%%%%%%%%%%%%%%%%%%%%%%%%%%%%%%%%%%%%%%%
% Section 3
\section{On the \textup{SDD} Size of Some Storage Access Functions}\label{sec:storage}

The following representations for the Boolean function HWB$_n$ and 
its negation $\overline{\textup{HWB}}_n$
were presented in \cite{BLSW99} in order to prove 
that generalizations of OBDDs used in applications lead to representations of small polynomial size. 

\begin{align}\label{HWB}
   \textup{HWB}_n(x)=\bigvee\limits_{1\leq k \leq n} E^n_k(x)\wedge x_k \text{ and}
\end{align}

\begin{align}\label{notHWB} 
   \overline{\textup{HWB}}_n(x)=\bigvee\limits_{1\leq k \leq n} (E^n_k(x)\wedge \overline{x}_k) \vee E^n_0(x),
\end{align}
where $E^n_j$, $j\in \{0, \ldots, n\}$, is the symmetric Boolean function on $n$ 
variables computing $1$ iff the number of ones in the input, that is the number of variables set to $1$,
is exactly $j$.
Using equation \ref{HWB} and \ref{notHWB}
it is easy to see (and was already shown in \cite{BLSW99}) that HWB$_n$ and $\overline{\textup{HWB}}_n$ 
can be represented w.r.t.\ every variable ordering 
by unambiguous nondeterministic OBDDs of size $\mathcal{O}(n^2)$
with only one nondeterministic node at the beginning.
Later on a similar construction was used in \cite{Bov16} in order to prove that the SDD size of the function HWB$_n$ is 
%$\mathcal{O}(n^3)$.
polynomial.

Now, the crucial observation is that the storage access functions defined in Section \ref{sec2} 
can all be represented
in this way. The indirect storage access function is equal to
\begin{align*} 
   \textup{ISA}_n(a,x)=\bigvee\limits_{0\leq j \leq n-1} 
      (|x(a)|_2=j) \wedge x_j \text{ or}
\end{align*}
\begin{align*} 
   \textup{ISA}_n(a,x)=
      %\bigvee\limits_{1\leq i \leq m-1} \bigvee\limits_{0\leq j \leq n-1} 
      \bigvee\limits_{\substack{1\leq i \leq m-1\\
                      0\leq j \leq n-1}}
      (|a|_2=i)\wedge (|(x_{ik}, \ldots, x_{(i+1)k-1})|_2=j) \wedge x_j.
\end{align*}
This characterization of ISA$_n$ leads easily to a similar one for its negated function.
\begin{align*} 
   \overline{\textup{ISA}}_n(a,x)=
       %\bigvee\limits_{1\leq i \leq m-1} \bigvee\limits_{0\leq j \leq n-1} 
       \bigvee\limits_{\substack{1\leq i \leq m-1\\ 
                       0\leq j \leq n-1}}
      (|a|_2=i)\wedge (|(x_{ik}, \ldots, x_{(i+1)k-1})|_2=j) \wedge \overline{x}_j.
\end{align*}

The weighted sum function can be written as 
\begin{align*} 
   \textup{WS}_n(x)=\bigvee\limits_{1\leq i \leq n} 
      ((S=i) \wedge x_i) \vee ((S=0)\wedge x_1) \vee ((S>n)\wedge x_1),
\end{align*}
where $S$ is the sum of all $ix_i$ in $\mathbb{Z}_p$, $1\leq i \leq n$. 
The negated weighted sum function is defined in the following way.
\begin{align*} 
   \overline{\textup{WS}}_n(x)=\bigvee\limits_{1\leq i \leq n} 
      ((S=i) \wedge \overline{x}_i) \vee ((S=0)\wedge \overline{x}_1) \vee ((S>n)\wedge \overline{x}_1).
\end{align*}

It is easy to see that the conjunction of a Boolean function $f$ and a projective function both given as OBDDs can be 
done in time and space $\mathcal{O}(|G|)$ where $G$ is the given OBDD representing $f$.
W.l.o.g.\ let $p(X)=x_i$ be the projective function and $f$ defined on the variable set $X$.
Traverse the OBDD $G$ and redirect all $0$-edges leaving nodes labeled by $x_i$ to the $0$-sink.
Alternatively, for all nodes $v$ labeled by $x_i$ all incoming edges into $v$ are redirected to the
$1$-successors of $v$. Since $v$ is not longer reachable afterwards, the nodes labeled by $x_i$ can be deleted.
Obviously, the size of the resulting OBDD is at most $|G|$.
For more details see, e.g., \cite{Weg00}.

%All storage access functions where the address can be computed by OBDDs of polynomial size (and possibly 
%a polynomial number of sinks) have the properties mentioned in Proposition \ref{sdd}.

%%Structured $d$-DNNFs are DNNFs that respect a vtree in a similar way as SDDs. 
%Since for $d$-DNNFs only determinism and not strong determinism is required,
%we can easily obtain the following.

Using the representations for HWB$_n$, ISA$_n$ and WS$_n$ mentioned above we can prove the following result
as a corollary from Theorem \ref{thm:transformation_into_sdd}.
%Proposition \ref{sdd}.
\begin{corollary}\label{WS}
The function \textup{ISA$_n$} can be represented by \textup{SDDs}
of size $\mathcal{O}(n^2)$,
the functions \textup{HWB$_n$} and \textup{WS$_n$} by \textup{SDDs} of size $\mathcal{O}(n^3)$. 
\end{corollary}

Corollary \ref{WS} is an improvement on a result of 
Bova and Szeider that ISA$_n$ can be represented by SDDs of size $\mathcal{O}(n^{13/5})$ \cite{BS17}.
Beame and Liew showed that SDDs are sometimes exponentially less concise than FBDDs \cite{BL15}. 
For this result they analyzed Boolean functions derived from a natural class of database queries and proved 
that there exists a Boolean function whose FBDD size is $\mathcal{O}(m^2)$ but its SDD size is 
at least $2^{\sqrt{m/3}-1}$, where the number of Boolean variables the investigated function depends on
is $m^2+2m$. 
Since the weighted sum function WS$_n$ has exponential FBDD size \cite{SZ00}, we complement Beame's and Liew's
result using Corollary \ref{WS}.
\begin{corollary}\label{cor:incomparable}
The complexity classes $\mathcal{P}(\textup{FBDD})$ and 
$\mathcal{P}(\textup{SDD})$ are incomparable
which means that $\mathcal{P}(\textup{FBDD})\not\subseteq \mathcal{P}(\textup{SDD})$ and vice versa.
%, there are Boolean functions representable by 
%one model in polynomial size but the other one needs exponential size and vice versa.
\end{corollary}

Note that there exist Boolean functions representable by polynomial-size FBDDs  
but every unambiguous nondeterministic \textup{OBDD} with only one nondeterministic node at the beginning
has exponential size and vice versa (see, e.g., \cite{BW99}). 
Therefore, Corollary \ref{cor:incomparable} is not really astonishing.
%%%%%%%%%%%%%%%%%%%%%%%%%%%%%%%%%%%%%%%%%%%%%%%%%%%%%%%%%%%%%%%%%%%%%%%%%%
% Section 4
\section{On the Succinctness of \textup{SDDs} and More General \textup{BDD} Variants}\label{sec:comparison}

%Our aim is to prove 
In this section, we prove 
that every function representable by $k$-OBDDs of polynomial size, 
where $k$ is a constant, can also be represented by SDDs of polynomial size. Moreover, there exist
Boolean functions representable by SDDs of polynomial size whose $k$-OBDD size is exponential.

\begin{theorem}\label{thm:kOBDD}
The complexity class $\mathcal{P}(k$-\textup{OBDD)} is a proper subclass of $\mathcal{P}(\textup{SDD})$
which means that $\mathcal{P}(k$-\textup{OBDD)}$\subsetneq \mathcal{P}(\textup{SDD})$.
\end{theorem}

The proof of Theorem \ref{thm:kOBDD} is technically not too involved. We only need the following observations.

\begin{lemma}\label{lem:kOBDDtransform}
Each function representable by a $k$-\textup{OBDD} of polynomial size can be represented by an unambiguous 
nondeterministic \textup{OBDD} of polynomial size w.r.t.\ the same variable ordering and
with only one nondeterministic node at the beginning.
\end{lemma}

Lemma \ref{lem:kOBDDtransform} can be proved by a polynomial transformation from $k$-OBDDs into equivalent
unambiguous nondeterministic \textup{OBDDs} with only one nondeterministic node at the beginning.
For this we can use a construction first used in \cite{BSSW98} and later on in \cite{BW99}.
For the sake of completeness we provide the proof of Lemma \ref{lem:kOBDDtransform}
in Appendix F.

%Negation can be done for $k$-OBDD representations by changing the labels of the $0$- and the $1$-sink. 
By changing the labels of the $0$- and the $1$-sink a $k$-OBDD representing a function $f$ can easily
be transformed into a $k$-OBDD for the negated function $\overline{f}$.
Therefore, for every function $f$ representable by $k$-OBDDs 
of polynomial size also the negated function $\overline{f}$ can be represented 
by $k$-OBDDs of polynomial size w.r.t.\ the same variable ordering as $f$. 
%The polynomial transformation from $k$-OBDDs into equivalent unambiguous OBDDs 
%(see Appendix A) does not change the variable ordering.
Hence, using Lemma \ref{lem:kOBDDtransform} together with Theorem \ref{thm:transformation_into_sdd}
we 
%can already conclude that 
obtain the result
$\mathcal{P}(k$-OBDD)$\subseteq \mathcal{P}($SDD).
Next, we prove that 
$\mathcal{P}(k$-OBDD) is even a proper subclass of $\mathcal{P}($SDD).

\begin{lemma}\label{lem:nondetnodes}
There exists Boolean functions $f$ such that $f$ and $\overline{f}$ can be represented
by unambiguous nondeterministic \textup{OBDDs} of polynomial size w.r.t.\ the same variable ordering 
but nondeterministic \textup{OBDDs} where the nondeterministic nodes are only at the beginning
need exponential size for $f$.
\end{lemma}

\begin{proof}[Sketch of proof]
Sauerhoff proved that there is a Boolean functions $f$ representable by nondeterministic OBDDs
of polynomial size but nondeterministic OBDDs for $f$ where nondeterministic nodes are only
at the beginning need exponential size \cite{Sau03b}. A careful analysis of his proof shows that 
the nondeterministic OBDD for the function $f$ which is a generalized storage access function is 
an unambiguous nondeterministic OBDD. Moreover, it is not too difficult but exhausting and tedious 
to prove that $\overline{f}$ can also be represented by unambiguous OBDDs of polynomial size
w.r.t.\ the same variable ordering as $f$. 
\end{proof}

Combining Lemma \ref{lem:kOBDDtransform} and \ref{lem:nondetnodes} with Theorem \ref{thm:transformation_into_sdd}
we  can prove Theorem \ref{thm:kOBDD}.

%%%%%%%%%%%%%%%%%%%%%%%%%%%%%%%%%%%%%%%%%%%%%%%%%%%%%%%%%%%%%%%%%%%%%%%%%%%
% Section 5
\section*{Concluding Remarks}

It is still open whether the complexity class $\mathcal{P}(k\textup{-OBDD})$, where $k$ is a constant, is a proper subset
of the complexity class that consists of all Boolean functions representable in polynomial size by  
unambiguous nondeterministic OBDDs with only one nondeterministic node at the beginning.
Furthermore, to the best of our knowledge the question whether the complexity class that consists
of all Boolean functions representable
by polynomial-size unambiguous nondeterministic OBDDs is closed under negation
is open. 
For unrestricted nondeterministic OBDDs of polynomial size the answer is negative.
%We know that the class of Boolean functions representable
%by polynomial-size nondeterministic OBDDs is not closed under negation.
Examples are all Boolean functions $f$ for which there is an exponential gap in the so-called 
nondeterministic one-way communication complexity for $f$ and $\overline{f}$ (for communication complexity 
see, e.g., \cite{KN97}).
The existence of a Boolean function $f$ with polynomial-size unambiguous nondeterministic OBDDs 
but for which $\overline{f}$ has exponential unambiguous nondeterministic OBDD size would answer the question
whether structured d-DNNFs are more powerful w.r.t.\ polynomial-size representations than SDDs
in the affirmative.

%An example is the function $\overline{\textup{PERM}}_n$ which has nondeterministic OBDDs of linear size
%but even the nondeterministic FBDD size of PERM$_n$ is exponential.
%%%%%%%%%%%%%%%%%%%%%%%%%%%%%%%%%%%%%%%%%%%%%%%%%%%%%%%%%%%%%%%%%%%%%%%%%%%
%\begin{acknowledgements}

%\end{acknowledgements}
%%%%%%%%%%%%%%%%%%%%%%%%%%%%%%%%%%%%%%%%%%%%%%%%%%%%%%%%%%%%%%%%%%%%%%%%%%%
%References

\bibliographystyle{spmpsci}  

\bibliography{references}

%%%%%%%%%%%%%%%%%%%%%%%%%%%%%%%%%%%%%%%%%%%%%%%%%%%%%%%%%%%%%%%%%%%%%%%%%
%%\sloppy
%\newpage
\appendix 
\section*{Appendix A: Proof of Lemma \ref{lemma:partitionlemma_vee_one_obdd_to_sdd}}
\label{appendix:proof_of_partition_lemma}

\begin{proof}
	First, we will show that the set $\Phi$ consists of at least two elements. For this purpose, it will be shown that the children of the $\vee$-node $u$ are elements of $R^{+}(\mathcal{F}, \beta)$. As a consequence, $\Phi$ consists of at least two elements because $\mathcal{F}$ was assumed to be simple and therefore $u$ has at least two children.
	
	Let $Y$ be the set of variables that are not assigned by $\beta$. 
According to the definition of $\beta(u)$, the assignment $\beta$ can be extended such that there exists an accepting path for $\beta$ in $\mathcal{F}$ containing $u$. Suppose $u$ were not maximal w.r.t. $Y$. Then, there would exist another node $u'$ in $\mathcal{F}$ such that $\vars{u} \subset \vars{u'} \subseteq Y$ and $\mathcal{F}_u$ is a subgraph of $\mathcal{F}_{u'}$. Let $x_i$ be the smallest variable of $\vars{u}$ w.r.t. $\pi$. Then, we have $Y = \{x_i, \dots, x_n\}$ according to the definition of $\beta(u)$. Notice that the graph $\mathcal{F}_{u'}$ must contain a variable $x_j$ with $j < i$ because it was assumed that there are no edges between $\vee$-nodes and $\mathcal{F}_{u}$ is a subgraph of $\mathcal{F}_{u'}$. Hence, $\vars{u'} \not\subseteq Y$ would hold which is a contradiction to the assumption. Therefore, $u \in R(\mathcal{F}, \beta)$ and its children are in $R^{+}(\mathcal{F}, \beta)$ because the node $u$ meets both conditions of the set $R(\mathcal{F}, \beta)$.
	
	Now, we give a proof by contradiction in order to show that $\Phi$ is a partition. Suppose to the contrary that there would be an $\vee$-node $u$ of $\mathcal{F}$ and a (partial) assignment $\beta \in \beta(u)$ such that the described set of functions $\Phi$ is not a partition. So $\Phi$ has to violate at least one of the partition properties. It will be shown that the violation of at least one partition property leads to a contradiction.\\
	
	%\noindent
	\textbf{Satisfiability}
	Suppose there would be a function $\varphi \in \Phi$ with $\varphi = \bot$. By definition of $R^{+}(\mathcal{F}, \beta)$ and $R^{+}(\overline{\mathcal{F}}, \beta)$ the nodes $u_1, \dots, u_k, v_1, \dots, v_l$ are no sinks. Therefore, an inner node $u_i$ or $v_j$ of $\mathcal{F}$ or $\overline{\mathcal{F}}$, respectively, represents the constant function $\bot$. This is a contradiction to the assumption of $\mathcal{F}$ and $\overline{\mathcal{F}}$ being simple.\\
	
	%\noindent
	\textbf{Disjointness}
	Suppose there would be functions $\varphi_1, \varphi_2 \in \Phi$ with $\varphi_1 \wedge \varphi_2 \neq \bot$. For this purpose, consider the following cases.
	
	\begin{enumerate}
		\item The functions $\varphi_1, \varphi_2$ are represented by nodes of the same $\vee_1$-OBDD, i.e., either $\varphi_1 = \Phi_{u_i}$, $\varphi_2 = \Phi_{u_j}$ or $\varphi_1 = \Phi_{v_i}$, $\varphi_2 = \Phi_{v_j}$ holds for $i \neq j$. Suppose $\Phi_{u_i} \wedge \Phi_{u_j} \neq \bot$. According to the definition of $R^{+}(\mathcal{F}, \beta)$ the assignment $\beta$ can be extended (maybe differently) such that there are accepting paths for $\beta$ in $\mathcal{F}$ containing $u_i$ and $u_j$. As $\Phi_{u_i} \wedge \Phi_{u_j} \neq \bot$ holds, there is an assignment $\beta^*$ of $Y$ (variables not assigned by $\beta$) such that $\Phi_{u_i}[\beta^*] = 1$ and $\Phi_{u_j}[\beta^*] = 1$. However, if we extend $\beta$ by $\beta^*$ then there are accepting paths for $(\beta, \beta^*)$ in $\mathcal{F}$ containing $u_i$ and $u_j$ with $i \neq j$. Because of the maximality of $u_i$ and $u_j$ w.r.t. $Y$ ($\mathcal{F}_{u_i}$ can't be a subgraph of $\mathcal{F}_{u_j}$ or vice versa) we know that there must be two distinct accepting paths. This is a contradiction to the property of $\mathcal{F}$ being unambiguous. If $\Phi_{v_i} \wedge \Phi_{v_j} \neq \bot$ holds, the contradiction can be derived analogously.
		
		\item The functions $\varphi_1, \varphi_2$ are represented by nodes of $\mathcal{F}$ and $\overline{\mathcal{F}}$, i.e., $\Phi_{u_i} \wedge \Phi_{v_j} \neq \bot$. Hence, there is an assignment $\beta^*$ of $Y$ such that $\Phi_{u_i}[\beta^*] = 1$ and $\Phi_{v_j}[\beta^*] = 1$ leading to accepting paths for $\beta^*$ in the subgraphs $\mathcal{F}_{u_i}$ and $\overline{\mathcal{F}}_{v_j}$. By definition of $R^{+}(\mathcal{F}, \beta)$ and $R^{+}(\overline{\mathcal{F}}, \beta)$ the assignment $\beta$ can be extended such that there are accepting paths in $\mathcal{F}$ and $\overline{\mathcal{F}}$ containing $u_i$ and $v_j$, respectively. Like in the former case $\beta$ can be extended by $\beta^*$ such that there are accepting paths for $(\beta, \beta^*)$ in $\mathcal{F}$ and $\overline{\mathcal{F}}$ leading to a contradiction to $\Phi_{\mathcal{F}} = \overline{\Phi_{\overline{\mathcal{F}}}}$.
	\end{enumerate}
	
	%\noindent
	\textbf{Cover}
	Suppose $\Phi_{u_1} \vee \dots \vee \Phi_{u_k} \vee \Phi_{v_1} \vee \dots \vee \Phi_{v_l} \neq \top$. Then, there exists an assignment $\beta^*$ of $Y$ such that $\Phi_{u_1}[\beta^*] = \dots = \Phi_{u_k}[\beta^*] = \Phi_{v_1}[\beta^*] = \dots = \Phi_{v_l}[\beta^*] = 0$. Hence, there is no accepting path for $\beta^*$ in $\mathcal{F}_{u_1}, \dots, \mathcal{F}_{u_k}, \overline{\mathcal{F}}_{v_1}, \dots, \overline{\mathcal{F}}_{v_l}$. Because every accepting path for $\beta$ in $\mathcal{F}$ and $\overline{\mathcal{F}}$ contains exactly one node from $u_1, \dots, u_k, v_1, \dots, v_l$, it is not possible to extend $\beta$ by $\beta^*$ resulting in an accepting path in $\mathcal{F}$ or $\overline{\mathcal{F}}$. This is a contradiction to $\Phi_{\mathcal{F}} \vee \overline{\Phi_{\overline{\mathcal{F}}}} = \top$.
	
	Now, we get the claimed lemma because the violation of at least one partition property leads to a contradiction. 
\end{proof}
%%%%%%%%%%%%%%%%%%%%%%%%%%%%%%%%%%%%%%%%%%%%%%%%%%%%%%%%%%%%%%%%%%%%%%%%%%%%%%%%%%%%%%%%
\section*{Appendix B: Proof of Lemma \ref{lemma:hauptlemma_vee_one_obdd_to_sdd}}
\label{appendix:hauptlemma_vee_one_obdd_to_sdd}

\begin{proof}
	We give a proof by induction on the depth $l$ of the subgraph $C_{(u,\emptyset)}$ of the SDD $C$. Note that in the following proof we sometimes denote $C$ to be the Boolean function represented at the corresponding SDD. It will be clear from the context whether the SDD or the represented function is meant.\\
	
	%\noindent
	\textbf{Base case $(l = 0):$}
	Since the depth of the subgraph $C_{(u,\emptyset)}$ is zero, it only consists of the node $(u, \emptyset)$. Therefore, $(u, \emptyset)$ was added to $C$ because of rule (a) from Simulation \ref{simulation:vtree_vee_one_obdds_to_sdds}. Otherwise, in case (b) or (c) the node $(u, \emptyset)$ would be connected to other nodes by outgoing edges resulting in an increase of depth.\\
	
	%\noindent
	First, we will show that $C_{(u, \emptyset)}$ is a syntactically correct $\textnormal{SDD}$. According to rule (a) of Simulation \ref{simulation:vtree_vee_one_obdds_to_sdds} the node $(u, \emptyset)$ was added to $C$ because of a decision node $u \in (V \cup \overline{V})$ for a variable $x_i \in X$ that is connected only to sinks. In this particular case $(u, \emptyset)$ was labeled by a literal $x_i$ or $\overline{x_i}$ depending on the semantics of the decision node $u$. 
Then, we know that $\nodefunc{u} = v_i'$ and $C_{(u, \emptyset)}$ is an SDD 
representing a projective function as in the base case of Definition \ref{def:sdds}
%of type (ii) from definition \todo{Reference SDD definition...} 
respecting vtree $T_{v_i'}$ since it contains a leaf labeled by the variable $x_i$. It is evident from rule (a) that $C_{(u,\emptyset)}$ represents the same function as the node $u$ of $\mathcal{F}$ or $\overline{\mathcal{F}}$ because $(u,\emptyset)$ was labeled according to the semantics of $u$.\\
	
	%\noindent
	\textbf{Induction hypothesis:}
	Each subgraph $C_{(u,\emptyset)}$ of $C$ with depth of at most $l$ is a syntactically correct $\textnormal{SDD}$ respecting vtree $T_v$ with $v = \nodefunc{u}$. Moreover, it represents the same Boolean function as the node $u$ of $\mathcal{F}$ or $\overline{\mathcal{F}}$.\\
	
	%\noindent
	\textbf{Inductive step $(l \rightarrow l+1, l \geq 0)$:}
	In this particular case $(u, \emptyset)$ of $C$ was added because of rule (b) or (c). Otherwise, the depth of $C_{(u, \emptyset)}$ would be zero as mentioned in the base case. Subsequently, we will have a look at both cases.\\
	
	%\noindent
	\underline{Case 1:} The node $(u, \emptyset)$ was added to $C$ due to rule (b) because of the decision node $u \in (V \cup \overline{V})$ for a variable $x_i \in X$. Then, $(u, \emptyset)$ is an $\vee$-node which is connected to the $\wedge$-nodes $(u, \wedge_0)$ and $(u, \wedge_1)$. The node $(u, \wedge_0)$ is connected to the node $(u, \overline{x_i})$ labeled by $\overline{x_i}$ and $(u, \wedge_1)$ is connected to $(u, x_i)$ labeled by $x_i$. Let $(u, u_0)$ and $(u, u_1)$ be the outgoing $0$- and $1$-edges of $u$, respectively. Then, $C$ also contains the edges $((u, \wedge_0), (u_0, \emptyset))$ and $((u, \wedge_1), (u_1, \emptyset))$. Since we have this setup of nodes and edges, $C_{(u, \emptyset)}$ is an 
inductively defined SDD constructed by smaller SDDs (see Definition \ref{def:sdds}).
%SDD of type (iii) from \todo{Reference SDD definition}. 
Next, we will show that $C_{(u, \emptyset)}$ is a syntactically correct SDD respecting the vtree $T_{v_i'}$ with $v_i' = \nodefunc{u}$. For this purpose, we show that the smaller SDDs are syntactically correct and that they represent Boolean functions which form a partition.\\

$C_{p_1} = C_{(u,\overline{x_i})}$ and $C_{p_2} = C_{(u, x_i)}$ are SDDs representing a projective function and they consist of a single node labeled by $\overline{x_i}$ or $x_i$, respectively. According to the construction of $T$ in Simulation \ref{simulation:vtree_vee_one_obdds_to_sdds} the left subtree of $T_{v_i'}$ is a leaf labeled by $x_i$. Hence, $C_{p_1}$ and $C_{p_2}$ are SDDs respecting this left subtree. $C_{s_1} = C_{(u_0, \emptyset)}$ and $C_{s_2} = C_{(u_1, \emptyset)}$ are subgraphs of $C$ with a depth of at most $l-1$ since $C_{(u, \emptyset)}$ is a subgraph with depth of at most $l+1$ and $(u, \emptyset)$ is connected to the nodes $(u_0, \emptyset)$, $(u_1, \emptyset)$ by paths of length two. By induction hypothesis $C_{(u_0, \emptyset)}$ and $C_{(u_1, \emptyset)}$ are syntactically correct SDDs respecting vtrees $T_{\nodefunc{u_0}}$ and $T_{\nodefunc{u_1}}$, respectively. Since there are edges $(u, u_0)$ and $(u, u_1)$ in $\mathcal{F}$ or $\overline{\mathcal{F}}$ and the variable ordering is given by $x_1, \dots, x_n$, we know that $\nodefunc{u_0} = v_j$ or $\nodefunc{u_0} = v_j'$ holds for $j > i$. Otherwise, the variable ordering of $\mathcal{F}$ or $\overline{\mathcal{F}}$ would be violated. Analogously, we can derive $\nodefunc{u_1} = v_h$ or $\nodefunc{u_1} = v_h'$ for $h > i$. Therefore, both SDDs respect the right subtree $T_{v_{i+1}}$. Moreover, we know that the set of functions $\{C_{p_1}, C_{p_2}\}$ yield a partition since the following conditions are satisfied:

	\begin{itemize}
		\item $C_{p_1} = C_{(u,\overline{x_i})} = \overline{x_i} \neq \bot$, $C_{p_2} = C_{(u, x_i)} = x_i \neq \bot$, \hfill (satisfiability)
		\item $C_{p_1} \wedge C_{p_2} = \overline{x_i} \wedge x_i = \bot$, and \hfill (disjointness)
		\item $C_{p_1} \vee C_{p_2} = \overline{x_i} \vee x_i = \top$. \hfill (cover)
	\end{itemize}

Now, we want to show the equivalence of the represented functions. According to rule (b) of the simulation we have $C_{(u, \emptyset)} = \overline{x_i} C_{(u_0, \emptyset)} \vee x_i C_{(u_1, \emptyset)}$. W.l.o.g. let $u \in V$. Since $u$ is a decision node for the variable $x_i$, we know that $\Phi_{\mathcal{F}_u} = \overline{x_i}\Phi_{\mathcal{F}_{u_0}} \vee x_i\Phi_{\mathcal{F}_{u_1}}$ because of the Shannon decomposition rule. By induction hypothesis we get $C_{(u_0, \emptyset)} = \Phi_{\mathcal{F}_{u_0}}$ and $C_{(u_1, \emptyset)} = \Phi_{\mathcal{F}_{u_1}}$. Hence, $C_{(u_0, \emptyset)} = \Phi_{\mathcal{F}_{u}}$ holds. If $u \in \overline{V}$, we can derive the equivalence the same way.\\
		
	%\noindent
	\underline{Case 2:} The node $(u, \emptyset)$ was added to $C$ due to rule (c) because of the $\vee$-node $u \in (V \cup \overline{V})$. W.l.o.g. suppose that $u \in V$ holds. According to rule (c) $(u, \emptyset)$ is an $\vee$-node which is connected to an $\wedge$-node $(u,v)$ for each $v \in (R^+ \cup \overline{R}^+)$. These $\wedge$-nodes are connected to further nodes based on rule (c). Thus, $C_{(u, \emptyset)}$ is an 
inductively defined SDD constructed by smaller SDDs.
%SDD of type (iii) from definition \todo{Reference definition of SDDs.} 
Next, we will show that $C_{(u, \emptyset)}$ is a syntactically correct SDD respecting the vtree $T_{v_i}$ with $v_i = \nodefunc{u}$. For this purpose, we show that the smaller SDDs are syntactically correct and that they represent Boolean functions which form a partition.\\

The subgraph $C_{(v,\emptyset)}$ has at most depth $l-1$ for each $v \in (R^+ \cup \overline{R}^+)$ because by assumption $C_{(u, \emptyset)}$ is a subgraph of depth at most $l+1$ and $(u, \emptyset)$ is connected to $(v, \emptyset)$ by paths of length two. Thus, by the use of the inductive hypothesis $C_{(v, \emptyset)}$ is a syntactically correct $\textnormal{SDD}$ respecting the vtree $T_{\nodefunc{v}}$ for each $v \in (R^+ \cup \overline{R}^+)$. Since we have the edge $(u,v)$ in $\mathcal{F}$ and the given variable ordering is $x_1, \dots, x_n$, we know that $\nodefunc{v} = v_j'$ holds for $j \geq i$ because by assumption $v$ cannot be an $\vee$-node. Therefore, $C_{(v, \emptyset)}$ is an $\textnormal{SDD}$ respecting the vtree $T_{v_i'}$ as well. $C_{(u,\bot)}$ and $C_{(u, \top)}$ are $\textnormal{SDDs}$ representing $\bot$ and $\top$, respectively. By definition the right subtree of $T_{v_i}$ is a leaf labeled by the help variable $h_{x_i, \dots, x_n}$. Hence, $C_{(u,\bot)}$ and $C_{(u, \top)}$ are SDDs respecting this right subtree. Furthermore, the partition properties are satisfied because the set of functions $\{C_{(v, \emptyset)} \;|\; v \in (R^+ \cup \overline{R}^+) \}$ yield a partition: By induction hypothesis we have $C_{(v, \emptyset)} = \Phi_{\mathcal{F}_v}$ for each $v \in R^+$ and $C_{(v, \emptyset)} = \Phi_{\overline{\mathcal{F}}_v}$ for each $v \in \overline{R}^+$. Thus, we know that $\{C_{(v, \emptyset)} \;|\; v \in (R^+ \cup \overline{R}^+) \}$ is a partition using Lemma \ref{lemma:partitionlemma_vee_one_obdd_to_sdd}. Therefore, the desired properties are fulfilled:

\begin{itemize}
	\item for each $v \in (R^+ \cup \overline{R}^+):$ $C_{(v, \emptyset)} \neq \bot$, \hfill (satisfiability)
	\item for each $v, v' \in (R^+ \cup \overline{R}^+)$ with $v \neq v':$ $C_{(v, \emptyset)} \wedge C_{(v', \emptyset)} = \bot$, and \hfill (disjointness)
	\item we have $\bigvee_{v \in (R^+ \cup \overline{R}^+)} C_{(v, \emptyset)} = \top$. \hfill (cover)
\end{itemize}

Finally, we get the equivalence of $C_{(u, \emptyset)}$ and $\Phi_{\mathcal{F}_u}$ by applying the inductive hypothesis on the representation of $C_{(v, \emptyset)}$ for each $v \in (R^+ \cup \overline{R}^+)$. Since $C_{(u, \emptyset)}$ was constructed by rule (c), $C_{(u, \emptyset)}$ represents the following Boolean function:
\begin{eqnarray*}
	C_{(u, \emptyset)}  &=& \bigvee_{\substack{v \,\in\, R^+, \\ (u,v) \,\in\, E}}  C_{(v, \emptyset)}   C_{(u,\top)}  \,\vee\, \bigvee_{\substack{v \,\in\, R^+, \\ (u,v) \,\notin\, E}}  C_{(v, \emptyset)}   C_{(u,\bot)}  \,\vee\, \bigvee_{v \in \overline{R}^+}  C_{(v, \emptyset)}   C_{(u,\bot)}  \\
	&=& \bigvee_{\substack{v \,\in\, R^+, \\ (u,v) \,\in\, E}} ( C_{(v, \emptyset)}  \wedge \top) \,\vee\, \bigvee_{\substack{v \,\in\, R^+, \\ (u,v) \,\notin\, E}} ( C_{(v, \emptyset)}  \wedge \bot) \,\vee\, \bigvee_{v \in \overline{R}^+} ( C_{(v, \emptyset)}  \wedge \bot) \\
	&=& \bigvee_{\substack{v \,\in\, R^+, \\ (u,v) \,\in\, E}} ( C_{(v, \emptyset)}  \wedge \top) = \bigvee_{(u,v) \,\in\, E} C_{(v, \emptyset)} \overset{\textnormal{(ind.)}}{=} \bigvee_{(u,v) \,\in\, E} \Phi_{\mathcal{F}_{v}} \;=\; \Phi_{\mathcal{F}_{u}}
\end{eqnarray*}
	
\end{proof}
%%%%%%%%%%%%%%%%%%%%%%%%%%%%%%%%%%%%%%%%%%%%%%%%%%%%%%%%%%%%%%%%%%%%%%%%%%%%%%%%%%%%%%%%%%%%%%%%%%%%%%%%%%%%%%
\section*{Appendix C: Proof of Lemma \ref{lemma:s_dnnf_to_vee_obdd_help_correct}}
\label{appendix:s_dnnf_to_vee_obdd_help_correct}

\begin{proof}
	We give a proof by contradiction adapting the proof of Lemma 5.4. from Beame and Liew \cite{BL15}. If necessary, we distinguish whether $i = j$ or $i \neq j$ holds.\\
	
	%\noindent
	Suppose to the contrary that $u$ is a leaf of the given $\textnormal{DNNF}_T$ $\mathcal{D}$ labeled by a variable $x_i \in X$ and there exists a nontrivial path between $(u,s)$ and $(v,s')$ in $\mathcal{F}$ such that there exists a node in $\mathcal{D}_v$ labeled by a variable $x_j$ fulfilling $x_j \leq x_i$ w.r.t. $\pi^*$. We choose $v$ such that there exists no other node $v'$ in $\mathcal{D}$ for which there is a path from $(u,s)$ to $(v',s'')$ and $\mathcal{D}_v \subset \mathcal{D}_{v'}$ holds. Therefore, we call the chosen subgraph $\mathcal{D}_v$ to be \emph{maximal}. We know that $\mathcal{D}_v$ exists because by assumption $(v,s')$ is a node in $\mathcal{F}$ resulting from the node $v$ in $\mathcal{D}$.\\
	
	%\noindent
	If the path from $(u,s)$ to $(v,s')$ only consists of one edge, then $v$ has to be a sink in $\mathcal{D}$ because $u$ is a leaf node and therefore $((u,s),(v,s'))$ was added to $\mathcal{F}$ because of the neutral edge $(u,v)$ in $\mathcal{D}$. This leads directly to a contradiction to the assumption that $\mathcal{D}_v$ contains a node labeled by a variable $x_j$. Now, consider paths from $(u,s)$ to $(v,s')$ in $\mathcal{F}$ consisting of at least two edges. Especially, consider the last edge of the path:
	\begin{eqnarray*}
		(u,s), \dots, (w,s''), (v,s')\,.
	\end{eqnarray*}
	Suppose that there would exist the edge $(w,v)$ in $\mathcal{D}$. This would lead to a contradiction to the assumed maximality of $\mathcal{D}_v$ because we had $\mathcal{D}_v \subset \mathcal{D}_w$ and $x_j$ would also occur in $\mathcal{D}_w$. Therefore, we know that the edge between $(w,s'')$ and $(v,s')$ has to be of the third type and was added to $\mathcal{F}$ because of a heavy edge in $\mathcal{D}$. Let $z$ be the corresponding $\wedge$-node in $\mathcal{D}$, $e = (z, v_l)$ the light edge and $e' = (z,v_r)$ the heavy edge. Since the edge is of the third type and $z$ is the corresponding $\wedge$-node, we have $v = v_r$ because the edge between $(w,s'')$ and $(v,s')$ was added to $\mathcal{F}$ by mapping the heavy edge $(z,v_r)$. Furthermore, for that reason we have $s'' = s' \cup \{e\}$. In the following we distinguish two cases at which point the light edge $e$ was added to the set of light edges $s''$.
(See Figure \ref{figure:example_help_lemma_syntactically_correctness} for a visualization of the two cases.)\\
	
	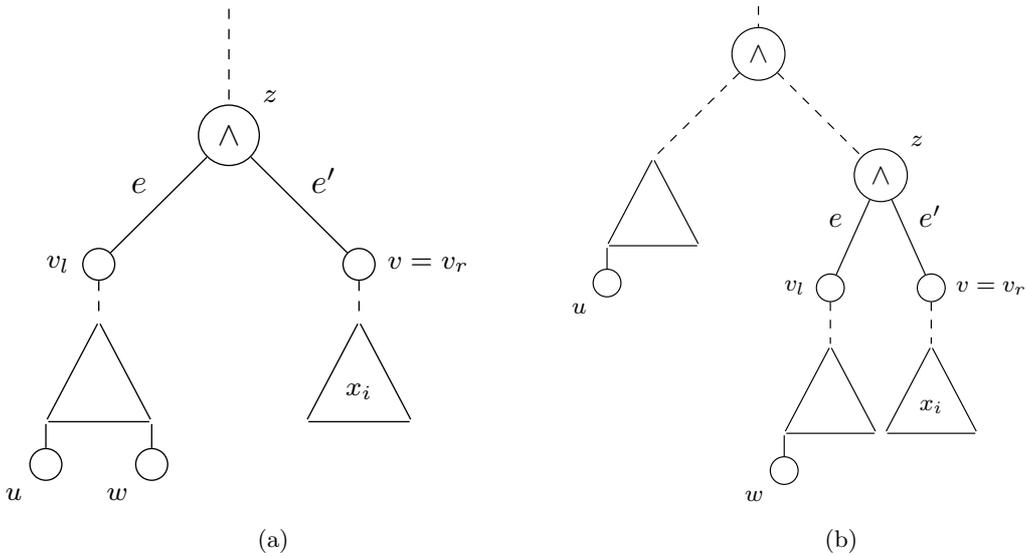
\begin{figure}[!b]
		\centering
		\begin{subfigure}[b]{0.49\textwidth}
			\resizebox{!}{0.3\textheight}{\begin{tikzpicture}[>= stealth', ]

\tikzstyle{vertex} = [draw, circle]
\tikzstyle{sink} = [draw, rectangle]
\tikzstyle{empty} = [inner sep = 0pt]

% 1. Ebene
\node[vertex] (and) at (0,0) [label={[font = \footnotesize] above right:{$z$}}] {$\wedge$};

% 0. Ebene
\node (top) [above = of and] {};

% 2. Ebene
\node[vertex] (v_l) [below left = of and, label={[font = \footnotesize] left:{$v_l$}}] {};
\node[vertex] (v_r) [below right = of and, label={[font = \footnotesize] right:{$v = v_r$}}] {};

% 3. Ebene
\node[empty] (v1_top) [below = 0.4 of v_l] {}; 
\node[empty] (v2_top) [below = 0.4 of v_r, label={[font = \footnotesize, label distance = 0.5cm] below :{$x_i$}}] {};
\node[empty] (v1_left) [below left = 1 and 0.5 of v1_top] {};
\node[empty] (v2_left) [below left = 1 and 0.5 of v2_top] {};
\node[empty] (v1_right) [below right = 1 and 0.5 of v1_top] {};
\node[empty] (v2_right) [below right = 1 and 0.5 of v2_top] {};

% 4. Ebene
\node[vertex] (v1_zero) [below = 0.25 of v1_left, label={[font = \footnotesize] below left:{$u$}}] {};
\node[vertex] (v1_one) [below = 0.25 of v1_right, label={[font = \footnotesize] below left:{$w$}}] {};

% Kanten
\draw (and) to node [auto, swap] {$e$} (v_l);
\draw (and) to node [auto] {$e'$} (v_r);

\draw[dashed] (v_l) to (v1_top);
\draw[dashed] (v_r) to (v2_top);

\draw (v1_top) to (v1_left);
\draw (v1_top) to (v1_right);
\draw (v2_top) to (v2_left);
\draw (v2_top) to (v2_right);
\draw (v1_left) to (v1_right);
\draw (v2_left) to (v2_right);

\draw (v1_left) to (v1_zero);
\draw (v1_right) to (v1_one);

\draw[dashed] (top) to (and);

\end{tikzpicture}}
			\caption{}
			\label{figure:example_help_lemma_syntactically_correctness_case_1}
		\end{subfigure}
		\begin{subfigure}[b]{0.49\textwidth}
			\resizebox{!}{0.3\textheight}{\begin{tikzpicture}[>= stealth', ]

\tikzstyle{vertex} = [draw, circle]
\tikzstyle{sink} = [draw, rectangle]
\tikzstyle{empty} = [inner sep = 0pt]

% 1. Ebene
\node[vertex] (and1) at (0,0) {$\wedge$};

% 0. Ebene
\node (top) [above = 0.25 of and] {};

% 2. Ebene
\node[vertex] (and2) [below right = of and1, label={[font = \footnotesize] above right:{$z$}}] {$\wedge$};

% 3. Ebene
\node[empty] (v1_top) [below left = of and1] {}; 
\node[empty] (v1_left) [below left = 1 and 0.5 of v1_top] {};
\node[empty] (v1_right) [below right = 1 and 0.5 of v1_top] {};

% 4. Ebene
\node[vertex] (v1_zero) [below = 0.25 of v1_left, label={[font = \footnotesize] below left:{$u$}}] {};
\node[vertex] (v_l) [below left = 1 and 0.25 of and2, label={[font = \footnotesize] left:{$v_l$}}] {};
\node[vertex] (v_r) [below right = 1 and 0.25 of and2, label={[font = \footnotesize] right:{$v = v_r$}}] {};

% 5. Ebene
\node[empty] (v2_top) [below = 0.5 of v_l] {}; 
\node[empty] (v2_left) [below left = 1 and 0.5 of v2_top] {};
\node[empty] (v2_right) [below right = 1 and 0.5 of v2_top] {};
\node[vertex] (v2_zero) [below = 0.25 of v2_left, label={[font = \footnotesize] below left:{$w$}}] {};

\node[empty] (v3_top) [below = 0.5 of v_r, label={[font = \footnotesize, label distance = 0.5cm] below :{$x_i$}}] {}; 
\node[empty] (v3_left) [below left = 1 and 0.5 of v3_top] {};
\node[empty] (v3_right) [below right = 1 and 0.5 of v3_top] {};

% Kanten
\draw[dashed] (and1) to (v1_top);
\draw[dashed] (and1) to (and2);

\draw (v1_top) to (v1_left);
\draw (v1_top) to (v1_right);
\draw (v1_left) to (v1_right);
\draw (v1_left) to (v1_zero);

\draw (v2_top) to (v2_left);
\draw (v2_top) to (v2_right);
\draw (v2_left) to (v2_right);

\draw (v3_top) to (v3_left);
\draw (v3_top) to (v3_right);
\draw (v3_left) to (v3_right);

\draw (and2) to node [auto, swap] {$e$} (v_l);
\draw (and2) to node [auto] {$e'$} (v_r);

\draw[dashed] (top) to (and);
\draw[dashed] (v_l) to (v2_top);
\draw[dashed] (v_r) to (v3_top);

\draw (v2_left) to (v2_zero);

\end{tikzpicture}}
			\caption{}
			\label{figure:example_help_lemma_syntactically_correctness_case_2}
		\end{subfigure}	
		\caption{Subgraphs of $\mathcal{D}$ visualizing both cases concerning the proof of 
                         Lemma \ref{lemma:s_dnnf_to_vee_obdd_help_correct}.}
		\label{figure:example_help_lemma_syntactically_correctness}
	\end{figure}
	
	%\noindent
	\textbf{At the beginning of the path (a)}
	Suppose $e \in s$ holds. Hence, we know that there is a path containing the light edge $e$ from the root of $\mathcal{D}$ to $u$. There is a path from $\mathcal{D}$ to $v$ containing the heavy edge $e'$ as well. Subsequently, we differentiate whether $i = j$ or $i \neq j$ holds.\\
	
	%\noindent
	Assume we have $i = j$. There is a node labeled by $x_i$ in the left subgraph $\mathcal{D}_{v_l}$, namely $u$. Additionally, by assumption the same variable $x_i = x_j$ appears in the right subgraph $\mathcal{D}_{v} = \mathcal{D}_{v_r}$. This is contradiction to the premise of $\mathcal{D}$ being a $\textnormal{DNNF}_T$ because for the $\wedge$-node $z$ we have: $\vars{v_l} \cap \vars{v_r} \neq \emptyset$, i.e., the decomposability is violated.\\
	
	%\noindent 
	Now, assume we have $i \neq j$. Let $\tau$ be the node of the vtree $T$ such that $\dnode{z} = \tau$. We can find such a node because $\mathcal{D}$ is a $\textnormal{DNNF}_T$. Let $\tau_l, \tau_r$ be the children of $\tau$. W.l.o.g. suppose $\vars{v_l} \subseteq \vars{\tau_l}$ and $\vars{v_r} \subseteq \vars{\tau_r}$. Otherwise, we could label the children of $\tau$ vice versa. Like in the preceding case we know that there is a node labeled by $x_i$ in $\mathcal{D}_{v_l}$ and a node labeled by $x_j$ in $\mathcal{D}_{v_r}$. So, we know that $x_i \in \vars{\tau_l}$ and $x_j \in \vars{\tau_r}$. By assumption we have $x_j < x_i$ w.r.t. $\pi^*$. Therefore, it must hold that $M^{\tau}_r < M^{\tau}_l$ by definition of $\pi^* = \pi(\mathcal{D},T)$. But now, we have a contradiction to the premise of $e = (z,v_l)$ being marked as a light edge which only holds for $M^{\tau}_l \leq M^{\tau}_r$.\\
	
	%\noindent
	\textbf{During the course of the path (b)}
	Suppose $e \notin s$ holds. Since the edge from $(w,s'')$ to $(v,s')$ is one of the third type, we know $e \in s''$. Hence, there must exist an edge of the first type $((z,s_1),(v_l,s_1 \cup \{e\}))$ on the path $(u,s), \dots, (w,s''), (v,s')$. Therefore, there is also a path from $(u,s)$ to $(z,s_1)$ in $\mathcal{F}$ and $\mathcal{D}_v \subset \mathcal{D}_z$ holds because of the heavy edge $e' = (z,v_r) = (z,v)$ in $\mathcal{D}$. The subgraph $\mathcal{D}_z$ contains a node labeled by $x_j$ as well because we assumed that $\mathcal{D}_v$ contains such a node. Altogether, we get a contradiction to the maximality of $\mathcal{D}_v$.\\
	
	%\noindent
	Now, the claimed lemma results from the contradictions of the individual cases. 
\end{proof}
%%%%%%%%%%%%%%%%%%%%%%%%%%%%%%%%%%%%%%%%%%%%%%%%%%%%%%%%%%%%%%%%%%%%%%%%%%%%%%%%%%%%%%%%%%%%%
\section*{Appendix D: Proof of Lemma \ref{lemma:s_dnnf_to_vee_obdd_semantic_correctness}}
\label{appendix:s_dnnf_to_vee_obdd_semantic_correctness}

\begin{proof}
	W.l.o.g. we assume that there is no $\wedge$- or $\vee$-node in $\mathcal{D}$ which uses constants as input. Otherwise, we could simplify $\mathcal{D}$ by propagating the constant according to the semantics of $\wedge$- and $\vee$-nodes. We give a proof by induction on the length $l$ (number of edges) of an accepting path and we represent a path by a list of its nodes.\\
	
	%\noindent
	\textbf{Base case $(l = 1):$}
	Let $P = (u_1, s_1), (u_2, s_2)$ be an accepting path in $\mathcal{F}$ for a variable assignment $b$. Since $P$ is an accepting path, the node $u_2' := (u_2, s_2)$ has to be a $1$-sink of $\mathcal{F}$. Hence, by rule (v) of Simulation \ref{simulation:structured_dnnf_to_vee_obdd} the node $u_2$ is also a $1$-sink in $\mathcal{D}$ and $s_2 = \emptyset$ holds. Furthermore, we know that $u_1$ cannot be an $\vee$- or $\wedge$-node because we assumed that there are no $\vee$- or $\wedge$-nodes with constant inputs. Thus, $u_1$ is a decision node for a variable $x_i \in X$ in $\mathcal{D}$ and by rule (i) we know that $u_1'$ is a decision node for the same variable. The edge $(u_1',u_2')$ was added to $\mathcal{F}$ because of the neutral edge $(u_1,u_2)$ in $\mathcal{D}$. The node $u_1$ has to be the root of $\mathcal{D}$ by construction of $\mathcal{F}$. Thus, the decision node $u_1$, the $1$-sink $u_2$ and the edge $(u_1,u_2)$ form a $1$-certificate representing $b$. If the edge $(u_1,u_2)$ in $\mathcal{D}$ is labeled by $0$, we have $b_i = 0$. Otherwise, we have $b_i = 1$.\\
	
	%\noindent
	\textbf{Induction hypothesis:}
	For each accepting path for a variable assignment $b$ in $\mathcal{F}$ with length at most $l$ there is a $1$-certificate of $\mathcal{D}$ representing $b$.\\
	
	%\noindent
	\textbf{Inductive step $(l-1 \rightarrow l, l \geq 2)$:}
	Let $P = (u_1, s_1), \dots, (u_l, s_l), (u_{l+1}, s_{l+1})$ be an accepting path for a variable assignment $b$ in $\mathcal{F}$. We do the inductive step by considering the following two cases.\\
	
	%\noindent
	\underline{Case 1:} $u_1' := (u_1, s_1)$ is an $\vee$-node of $\mathcal{F}$. We know that $u_1'$ is the root of $\mathcal{F}$ and by construction $u_1' = (\textnormal{root}(\mathcal{D}), \emptyset)$ holds. Furthermore, $u_1$ has to be an $\vee$-node in $\mathcal{D}$ as well since only $\vee$-nodes of $\mathcal{D}$ are mapped to $\vee$-nodes of $\mathcal{F}$ by the given simulation. Therefore, the edge $((u_1, s_1), (u_2, s_2))$ was added to $\mathcal{F}$ because of the neutral edge $(u_1,u_2)$ in $\mathcal{D}$.\\
	
	%\noindent
	Consider the nondeterministic OBDD $\mathcal{F}'$ which results from the given simulation with input $\mathcal{D}_{u_2}$. $\mathcal{F}'$ corresponds to the nondeterministic OBDD with root $(u_2, s_2)$ which arises from $\mathcal{F}$ by removing all nodes and edges that cannot be reached from $(u_2, s_2)$. Now, consider the subpath $P' = (u_2, s_2), \dots, (u_{l+1}, s_{l+1})$ of $P$. The subpath $P'$ is an accepting path for $b$ in $\mathcal{F}'$. Otherwise, $P$ would be no accepting path for $b$ in $\mathcal{F}$. Furthermore, $P'$ contains an edge less than $P$. Thus, by the inductive hypothesis there exists a $1$-certificate of $\mathcal{D}_{u_2}$ representing $b$. Since $u_1$ is an $\vee$-node and the root of $\mathcal{D}$, we can expand the $1$-certificate of $\mathcal{D}_{u_2}$ by $u_1$ and the edge $(u_1, u_2)$ in order to get a $1$-certificate of $\mathcal{D}$.\\
	
	%\noindent
	\underline{Case 2:} $u_1' := (u_1, s_1)$ is not an $\vee$-node of $\mathcal{F}$. Then, $u_1'$ has to be an unlabeled node resulting from the $\wedge$-node $u_1$ in $\mathcal{D}$. Suppose to the contrary that $u_1'$ would be a sink. Then, $P$ would be no computing path because $P$ contains two sinks. Moreover, suppose $u_1'$ would be a decision node for a variable $x_i \in X$. Then, by rule (i) of Simulation \ref{simulation:structured_dnnf_to_vee_obdd} the node $u_1$ is also a decision node for the same variable. But now, $u_1$ would be a leaf in $\mathcal{D}$ and therefore the length of the accepting path had to be $1$. Finally, consider $u_1'$ would be an unlabeled node which was created because of a $1$-sink in $\mathcal{D}$. Then, $\mathcal{D}$ would only consist of this $1$-sink and $P$ had length 0.\\
	
	%\noindent
	Since $u_1'$ is the root of $\mathcal{F}$, we have $u_1' = (\textnormal{root}(D), \emptyset)$. As $P$ is an accepting path in $\mathcal{F}$, the node $u_{l+1}' := (u_{l+1}, s_{l+1})$ is a $1$-sink. We know by rule (v) of Simulation \ref{simulation:structured_dnnf_to_vee_obdd} that $s_{l+1} = \emptyset$ and $u_{l+1}$ is also a $1$-sink in $\mathcal{D}$. The edge $((u_1, s_1), (u_2, s_2))$ was added to $\mathcal{F}$ because of the light edge $e = (u_1,u_2)$ in $\mathcal{D}$ since $u_1$ is an $\wedge$-node. Therefore, we have $s_2 = s_1 \cup \{e\} = \{e\}$. Since $s_{l+1} = \emptyset$ holds, there must exist an edge $((u_i, s_i), (u_{i+1}, s_{i+1}))$ in $\mathcal{F}$ with $3 \leq i \leq l-1$ which was added because of the corresponding heavy edge $e' = (u_1, v_r)$ in $\mathcal{D}$. The bounds of $i$ emerge from the first and last position of an unlabeled node $(u_i, s_i)$ on $P$ that is connected to $(u_{i+1}, s_{i+1})$ by an edge of the third type. Otherwise, we would have $e \in s_{l+1}$ resulting in $(u_{l+1}, s_{l+1})$ not being a $1$-sink and $P$ not being accepting. As $((u_i, s_i), (u_{i+1}, s_{i+1}))$ is an edge of the third type, we have $u_{i+1} = v_r$ and $u_i$ is a $1$-sink in $\mathcal{D}$.\\
	
	%\noindent
	Now, let $P' = (u_2, s_2), \dots, (u_i, s_i)$ and $P'' = (u_{i+1}, s_{i+1}), \dots, (u_{l+1}, s_{l+1})$ be subpaths of $P$ such that $i$ is chosen as described in the previous paragraph. Consider the nondeterministic OBDD $\mathcal{F}'$ resulting from the given simulation of the left subgraph $\mathcal{D}_{u_2}$. Alternatively, we can get $\mathcal{F}'$ from $\mathcal{F}$ by removing all nodes $(u,s)$ fulfilling $e \notin s$ and replacing unlabeled nodes without outgoing edges by $1$-sinks. Moreover, consider the nondeterministic OBDD $\mathcal{F}''$ resulting from the given simulation of the right subgraph $\mathcal{D}_{u_{i+1}}$. We can get $\mathcal{F}''$ from $\mathcal{F}$ by removing the root $(u_1, s_1)$ and each node $(u,s)$ for which $e \in s$ holds.\\
	
	%\noindent
	Next, we want to derive $1$-certificates of $\mathcal{D}_{u_2}$ and $\mathcal{D}_{u_{i+1}}$ representing $b$ from the given subpaths $P'$ and $P''$, respectively. The root of $\mathcal{F}'$ is the fist node $(u_2, s_2)$ of $P'$. Each edge of $P'$ exists in $\mathcal{F}'$ since we only removed nodes $(u,s)$ for which $e \notin s$ holds. Furthermore, the node $(u_i, s_i)$ is a $1$-sink in $\mathcal{F}'$ because we split up $P$ such that $((u_i, s_i), (u_{i+1}, s_{i+1}))$ is an edge of the third type. Hence, $P'$ is an accepting path for $b$ in $\mathcal{F}'$ which is shorter than $P$. By induction hypothesis there is a $1$-certificate of $\mathcal{D}_{u_2}$ representing $b$.\\
	
	%\noindent
	The root of $\mathcal{F}''$ is the first node $(u_{i+1}, s_{i+1})$ of $P''$. The path $P''$ is a proper subpath of $P$ and has to be an accepting path for $b$ in $\mathcal{F}''$ since otherwise $P$ would be no accepting path for $b$ in $\mathcal{F}$. By induction hypothesis there is a $1$-certificate of $\mathcal{D}_{u_{i+1}}$ representing $b$.\\
	
	%\noindent
	Finally, we will combine the $1$-certificates of $\mathcal{D}_{u_2}$ and $\mathcal{D}_{u_{i+1}}$ in order to get a $1$-certificate of $\mathcal{D}$ representing $b$. At the beginning we observed that $u_1$ has to be the root of $\mathcal{D}$. The edge $((u_1, s_1),(u_2,s_2))$ of $P$ was added to $\mathcal{F}$ because of the light edge $(u_1,u_2)$ and $((u_i, s_i),(u_{i+1},s_{i+1}))$ was added because of the heavy edge $(u_1, u_{i+1})$. Thus, the node $u_1$, both edges $(u_1,u_2)$, $(u_1, u_{i+1})$, and the $1$-certificates of $\mathcal{D}_{u_2}$ and $\mathcal{D}_{u_{i+1}}$ give a $1$-certificate of $\mathcal{D}$ representing $b$. 
\end{proof}
%%%%%%%%%%%%%%%%%%%%%%%%%%%%%%%%%%%%%%%%%%%%%%%%%%%%%%%%%%%%%%%%%%%%%%%%%%%%%%%%%%%%%%%%%%%%%%%%
\section*{Appendix E: Proof of Lemma \ref{lemma:s_dnnf_to_vee_obdd_vollstaendigkeit}}
\label{appendix:s_dnnf_to_vee_obdd_completeness}

\begin{proof}
	W.l.o.g. we assume that there is no $\wedge$- or $\vee$-node in $\mathcal{D}$ which uses constants as input. Otherwise, we could simplify $\mathcal{D}$ by propagating the constant according to the semantics of $\wedge$- and $\vee$-nodes. Furthermore, we assume that $\mathcal{D}$ consists not only of a sink. We give a proof by induction on the depth $l$ (longest path from the root to a leaf) of a $1$-certificate of $\mathcal{D}$.\\
	
	%\noindent
	\textbf{Base case $(l = 1):$}
	Let $\mathcal{C}$ be a $1$-certificate of $\mathcal{D}$ of depth one representing the satisfying variable assignment $b$. By definition of a certificate we have $\rootset{\mathcal{C}} = \rootset{\mathcal{D}} =: u$. The root $u$ has to be a decision node for a variable $x_i \in X$. Suppose to the contrary that $u$ would be an $\wedge$- or an $\vee$-node. Then, the inputs of $u$ had to be constants as $\mathcal{C}$ is of depth one. This was precluded by assumption. Moreover, $u$ is not a sink since we also precluded it by assumption. Therefore, $\mathcal{C}$ consists of the root $u$, a $1$-sink $v$, and an edge $(u,v)$ which is labeled consistently with $b$. So, $P = (u,\emptyset), (v, \emptyset)$ is an accepting path for $b$ in $\mathcal{F}$.\\
	
	%\noindent
	\textbf{Induction hypothesis:}
	For each $1$-certificate of $\mathcal{D}$ representing $b$ with depth of at most $l$, there exists an accepting path for $b$ in $\mathcal{F}$.\\
	
	%\noindent
	\textbf{Inductive step $(l \rightarrow l+1, l \geq 1)$:}
	Let $\mathcal{C}$ be a $1$-certificate of $\mathcal{D}$ with depth $l+1$ representing the satisfying variable assignment $b$. Let $u:= \rootset{\mathcal{C}} = \rootset{\mathcal{D}}$. We do the inductive step by considering the following two cases.\\
	
	%\noindent
	\underline{Case 1:} $u$ is an $\vee$-node. By definition of certificates, $\mathcal{C}$ contains exactly one child node of $u$, called $v$, and the edge $(u,v)$. The subtree $\mathcal{C}_v$ of $\mathcal{C}$ has to be a $1$-certificate of $\mathcal{D}_v$ since $\mathcal{C}$ would not be one of $\mathcal{D}$. Moreover, the depth of $\mathcal{C}_v$ is $l$. By induction hypothesis there exists an accepting path for $b$ in the nondeterministic OBDD $\mathcal{F}_v$ which results from the given simulation by input of $\mathcal{D}_v$. Since we have $S(v) = \{\emptyset\}$ for $v$ in $\mathcal{D}_v$ and $\emptyset \in S(v)$ for $v$ in $\mathcal{D}$, we know that $\mathcal{F}_v$ is a subgraph of $\mathcal{F}$. Apart from the nodes and edges of $\mathcal{F}_v$, $\mathcal{F}$ also contains the edge $((u,\emptyset), (v, \emptyset))$ because of the neutral edge $(u,v)$ in $\mathcal{D}$. We can extend $P_v$ to be an accepting path of $\mathcal{F}$ by adding $((u,\emptyset), (v, \emptyset))$ as a prefix.\\
	
	%\noindent
	\underline{Case 2:} $u$ is not an $\vee$-node. The node $u$ has to be an $\wedge$-node. Suppose to the contrary that $u$ is a decision node. Then, the depth of $\mathcal{C}$ would be $1$ as in the base case. By definition of certificates, $\mathcal{C}$ contains both children of $u$, called $u_l$ and $u_r$. We assume that $(u,u_l)$ is the light edge. Otherwise, we rename the child nodes of $u$. The subtrees $\mathcal{C}_{u_l}$ and $\mathcal{C}_{u_r}$ have to be $1$-certificates of $\mathcal{D}_{u_l}$ and $\mathcal{D}_{u_r}$, respectively, because otherwise $\mathcal{C}$ would be no $1$-certificate of $\mathcal{D}$. Moreover, we know that $\mathcal{C}_{u_l}$ and $\mathcal{C}_{u_r}$ have a depth of at most $l$. By induction hypothesis there are accepting paths for $b$ in $\mathcal{F}_{u_l}$ and $\mathcal{F}_{u_r}$ which are nondeterministic OBDDs resulting from the simulation of $\mathcal{D}_{u_l}$ and $\mathcal{D}_{u_r}$, respectively.\\
	
	%\noindent
	Let $P' = (u_1', s_1'), \dots, (u_g', s_g')$ and $P'' = (u_1'', s_1''), \dots, (u_h'', s_h'')$ be the accepting paths for $b$ in $\mathcal{F}_{u_l}$ and $\mathcal{F}_{u_r}$, respectively. According to the simulation we know that $(u_1', s_1') = (\rootset{\mathcal{D}_{u_l}}, \emptyset) = (u_l, \emptyset)$ and $(u_1'', s_1'') = (\rootset{\mathcal{D}_{u_r}}, \emptyset) = (u_r, \emptyset)$. Furthermore, $(u_g', s_g')$ and $(u_h'', s_h'')$ have to be $1$-sinks and $s_g' = s_h'' = \emptyset$. It is our aim to identify $P'$ and $P''$ in $\mathcal{F}$ and to extend them with two further edges to an accepting path for $b$.\\
	
	%\noindent
	Since $(u, u_r)$ is a heavy edge of $\mathcal{D}$ leading to the root of $\mathcal{D}_{u_r}$, we have $S(u_r) = \{\emptyset\}$ in $\mathcal{D}$. If there was any other set of light edges in $S(u_r)$, then the decomposability property would be violated at the $\wedge$-node $u$: one of the light edges of a set of $S(u_r)$ has to connect a node of $\mathcal{D}_{u_l}$ with $u_r$. Otherwise, $\mathcal{D}_{u_r}$ would be cyclic. Furthermore, we have $S(u_r) = \{\emptyset\}$ in $\mathcal{D}_{u_r}$ since $u_r$ is the root of $\mathcal{D}_{u_r}$. Hence, $\mathcal{F}_{u_r}$ is a subgraph of $\mathcal{F}$. Thus, $P''$ is a path from $(u_r, \emptyset)$ to a $1$-sink in $\mathcal{F}$.\\
	
	%\noindent
	However, $(u, u_l)$ is a light edge in $\mathcal{D}$ such that $\{e\} \in S(u_l)$ holds in $\mathcal{D}$. Further, we also know that $S(u_l) = \{\{e\}\}$ holds in $\mathcal{D}$ because otherwise the decomposability of $\mathcal{D}$ would be violated. But we have $S(u_l) = \emptyset$ in $\mathcal{D}_{u_l}$ because $u_l$ is the root of $\mathcal{D}_{u_l}$. Hence, there exists an isomorphism between $\mathcal{F}_{u_l}$ and the subgraph of $\mathcal{F}$ which was added because of $\mathcal{D}_{u_l}$ since $((u, s), (v, s'))$ is an edge of $\mathcal{F}_{u_l}$ if and only if $((u, s \cup \{e\}), (v, s' \cup \{e\}))$ is an edge of $\mathcal{F}$.\\
	
	%\noindent
	Finally, consider $P = (u, \emptyset), (u_1', s_1' \cup \{e\}), \dots, (u_g', s_g' \cup \{e\}), (u_1'', s_1''), \dots, (u_h'', s_h'')$. We get $P$ by concatenating a modified version of $P'$, $P''$, and two more edges. The first edge $((u, \emptyset), (u_1', s_1' \cup \{e\})) = ((u, \emptyset), (u_l, \{e\}))$ exists in $\mathcal{F}$ because of the light edge $(u, u_l)$ in $\mathcal{D}$. The sequence of edges $(u_1', s_1' \cup \{e\}), \dots, (u_g', s_g' \cup \{e\})$ exist in $\mathcal{F}$ since $P'$ is an accepting path of $\mathcal{F}_{u_l}$ and there exists the isomorphism between the nodes of $\mathcal{F}$ and $\mathcal{F}_{u_l}$. Furthermore, we have $(u_g', s_g' \cup \{e\}) = (u_g', \{e\})$ since $P'$ is an accepting path and therefore $(u_g', s_g')$ is a $1$-sink in $\mathcal{F}_{u_l}$ with $s_g' = \emptyset$. So, the edge $((u_g', s_g' \cup \{e\}), (u_1'', s_1'')) = ((u_g', \{e\}), (u_r, \emptyset))$ exists because of the heavy edge $(u, u_r)$ in $\mathcal{D}$. Finally, the path $(u_1'', s_1''), \dots, (u_h'', s_h'')$ ends in a $1$-sink of $\mathcal{F}$. Hence, $P$ is an accepting path for $b$ in $\mathcal{F}$.
\end{proof}

\section*{Appendix F: Proof of Lemma \ref{lem:kOBDDtransform}\label{appendix:E}
}

\begin{proof}
Our aim is to prove that each function representable by a $k$-OBDD of polynomial size,
where $k$ is an arbitrary constant, 
can also be represented by an unambiguous nondeterministic OBDD of polynomial size 
with only one nondeterministic node at the beginning.
For this reason we present a polynomial transformation from $k$-OBDDs into 
equivalent restricted unambiguous nondeterministic OBDDs. The following construction was first used in
\cite{BSSW98} proving that the satisfiability problem can be solved in polynomial time
for functions represented by $k$-OBDDs. Later it was also used in \cite{BW99} in order 
to prove that $k$-OBDDs can be polynomially transformed into OBDDs which use 
so-called parity nondeterminism.

Let $f$ be the function represented by a given $k$-OBDD $G$ and let $k$ be a constant.
We start with the observation 
that there is exactly one accepting path for each $1$-input
in a $k$-OBDD since it is a deterministic model. 
%Since $k$-OBDDs are a deterministic model, 
%there is exactly one accepting path for each $1$-input.
Now, the crucial idea is a suitable decomposition of a given 
$k$-OBDD $G$.
% into an unambiguous nondeterministic OBDD for $f$ that uses
%only one nondeterministic node at the beginning.
For this we consider the at most $s = |G|^{k-1}$ possibilities to
 switch between the layers of $G$.
The $i$-th auxiliary function, $1 \le i \le s$, equals 1 for the $1$-inputs of $f$
that choose the $i$-th possibility which means that the accepting paths for these inputs 
run through the layers of the given $k$-OBDD $G$ in the chosen way.
Such an auxiliary function can be represented by an OBDD of size $|G|^k$
by combining parts of the $k$-OBDD via conjunction. 
Here we use the fact that in a $k$-OBDD all layers respect
the same variable ordering.
(OBDDs in general do not have nice algorithmic properties. There are examples known
such that $g_n$ and $h_n$ are two Boolean functions which have OBDDs of linear size (for
different variable orderings) but $f_n=g_n \wedge h_n$ has even exponential nondeterministic FBDD
size. 
%The function PERM 
The so-called {\it permutation test function} 
is an example of such a function $f_n$. 
If only OBDDs respecting the same variable ordering
are considered, all important operations can be performed efficiently.
For more details see, e.g., \cite{Weg00}.)

Next, we describe these ideas more precisely. Let $G_1, \ldots, G_k$ be the layers of $G$.
If $b$ is a $1$-input, the accepting path for $b$ leads 
through some layers $\ell(1)=1< \ell(2)< \cdots < \ell(r)\leq k$ of $G$, 
where $v_1$ is the source of $G$, $G_{\ell(i)}$ is reached at some node $v_i$, and from some node in $G_{\ell(r)}$ the 
sink labeled by $1$ is reached. There are at most 
$|G|^{k-1}$ possibilities to choose $r, \ell(2), \ldots, \ell(r), v_2, \ldots, v_r$.
For an arbitrary but fixed choice of these parameters we consider the layers $G_{\ell(1)}, \ldots, G_{\ell(r)}$ and 
the sinks. We transform $G_{\ell(i)}$, $i\in \{1, \ldots, r\}$, into an OBDD $G'_{\ell(i)}$ 
with source $v_i$ in the following way.
An edge leaving $G_{\ell(i)}$ is replaced by an edge to a $1$-sink if either $i <r$ 
and the edge leads to $v_{i+1}$ or $i=r$ and the edge leads to the $1$-sink. All other edges 
leaving a node in $G_{\ell(i)}$ are replaced by edges to the $0$-sink. Now, $G'_{\ell(i)}$ consists of all
nodes (and corresponding edges) reachable from $v_i$.
The function represented by $G$ has a $1$-input iff for some 
$r, \ell(2), \ldots, \ell(r), v_2, \ldots, v_r$ the corresponding OBDDs $G'_{\ell(1)}, \ldots, G'_{\ell(r)}$
have a common $1$-input. Since all these OBDDs respect the same variable ordering,
Bryant's apply algorithm \cite{Bry86} can be used to obtain 
an OBDD of size $\mathcal{O}(|G|^k)$ for the conjunction  of the functions represented
by $G'_{\ell(1)}, \ldots, G'_{\ell(r)}$
in time $\mathcal{O}(|G|^k)$.
Considering all choices of the parameters 
$r, \ell(2), \ldots, \ell(r), v_2, \ldots, v_r$
we obtain a unambiguous nondeterministic OBDD of size $\mathcal{O}(|G|^{2k-1})$
which has only one nondeterministic node at the beginning.
\end{proof}

%%%%%%%%%%%%%%%%%%%%%%%%%%%%%%%%%%%%%%%%%%%%%%%%%%%%%%%%%%%%%%%%%%%%%%%%%
\end{document}